\documentclass[aps,10pt,twocolumn,pra]{revtex4-2} 

\usepackage{epsfig,amsmath,amssymb,mathrsfs,mathtools,amscd}
\usepackage[linkcolor=red]{hyperref}
\usepackage[dvipsnames]{xcolor}

\usepackage{soul}
\usepackage{multirow}
\usepackage{enumerate}
\usepackage{amsmath} 
\usepackage{amsthm}
\usepackage{verbatim} 
\usepackage{mathtools} 
\usepackage{tikz}
\usepackage{tikzit}
\usepackage{scalerel}
\usepackage{physics}
\usepackage{booktabs, multirow, makecell}

\newtheorem{theorem}{Theorem}
\newtheorem{lemma}{Lemma}
\newtheorem{Corollary}{Corollary}

\newtheorem{definition}{Definition}

\newtheorem{proposition}{Proposition}
\newtheorem{remark}{Remark}

\def\ind{\ensuremath{\operatorname{ind}}}
\def\rank{\ensuremath{\operatorname{rank}}}
\def\transf#1{\ensuremath{\mathcal{#1}}}
\def\tH{\transf H}
\def\tI{\transf I}
\def\tE{\transf E}
\def\tF{\transf J}
\def\tJ{\transf E ^\dag}
\def\tK{\transf K}

\def\tP{\transf P}

\def\tS{\transf S}
\def\tT{\transf T}
\def\tU{\transf U}
\def\tV{\transf E}
\def\tor{\ensuremath{\bigtimes_{t=1}^s\mathbb Z_{m_t}}}

\def\Tred{(\tT^N)_r}

\def\Tredmat{(U_{\tT^N})_r}
\def\Tmat{(U_\tT)}
\def\Tnmat{U_{\tT^N}}

\def\Aplocc{\ensuremath{\mathsf{A}^{\Pi}(\mathbb{L})}}

\def\Aredloc{\ensuremath{\mathsf A'(\mathbb{L'})}}

\mathchardef\minus="002D

\def\<{\langle}
\def\>{\rangle}
 \def\ket#1{| #1 \rangle}

\def\bra#1{\langle #1 |}
\def\ketbra#1#2{| #1 \rangle\!\langle#2 |}

\def\braket#1#2{\langle #1 | #2 \rangle}

\def\vc#1{\boldsymbol{\mathrm #1}}
\def\Z{\mathbb Z}
\def\N{\mathbb N}

\def\L{\ensuremath{\mathbb{L}}}

\def\Mod{\operatorname{mod}}

\def\gt{\ensuremath{\boxtimes}}

\def\L{\ensuremath{\mathbb{L}}}
\def\<{\langle}
\def\>{\rangle}

\def\acomm#1#2{\ensuremath{\{#1,#2\}}}
\def\gc#1#2{\ensuremath{\{[#1,#2]\}}}

\def\Z{\mathbb Z}

\def\C{\mathbb C}
\def\N{\mathbb N}
\def\R{\mathbb R}

\def\A{\mathsf A}

\def\G{\mathbb G}
\def\M{\mathsf M}
\def\K{\mathbb K}
\def\P{\mathcal{P}}
\def\Mod{\operatorname{mod}}

\newcommand{\Eq}{Eq.~}

\newcommand\Cl{\ensuremath{\mathcal{C}\ell}}
\def\m#1{\ensuremath{\mathsf{#1}}}
\def\M#1#2{\ensuremath{\operatorname{Mat}(\C^{#1\lvert#2})}}
\def\dim{\operatorname{dim}}

\def\s{\shortmid}

\DeclareMathOperator*{\bigboxtimes}{\scalerel*{\boxtimes}{\sum}}
\DeclareMathOperator*{\bigboxplus}{\scalerel*{\boxplus}{\sum}}

\newcommand\Item[1][]{%
  \ifx\relax#1\relax  \item \else \item[#1] \fi
  \abovedisplayskip=0pt\abovedisplayshortskip=0pt~\vspace*{-\baselineskip}}
\def\Tr{\operatorname{Tr}}
\usetikzlibrary {arrows.meta}
\usetikzlibrary{%
    decorations.pathreplacing,%
    decorations.pathmorphing,%
    arrows,
    arrows.meta,
    positioning,
    calc,
    shapes,
    shadows,
    shapes.geometric,
   math
    }

\definecolor{myblue1}{HTML}{3A9831}
\definecolor{myblue2}{HTML}{08543A}
\definecolor{mycolor}{HTML}{FFC56E}
\definecolor{cifcolor}{HTML}{D50032}
\definecolor{Mswapcolor}{HTML}{FF571F}
\definecolor{Ucolor}{HTML}{C4D600}
\definecolor{phasecolor}{HTML}{FE5000}
\definecolor{phasecolorbis}{HTML}{C8102E}
\definecolor{phasecolortris}{HTML}{EAAA00}
\definecolor{tassellationcol}{HTML}{E4CC00}
\definecolor{important}{HTML}{E55121}
\definecolor{bluenonren}{rgb}{0.21,0.48,0.80}
\definecolor{blueren}{RGB}{166, 216, 236}
\newcommand\bitsize{0.6 cm}
\newcommand\projspace{0.12cm}
\def\LR{\hspace{\projspace}$\m{L}$\hspace{\projspace} \nodepart{two}\hspace{\projspace}$\m{R}$\hspace{\projspace}}

\def\LRbar{\hspace{\projspace}$\bar{\m{L}}$\hspace{\projspace} \nodepart{two}\hspace{\projspace}$\bar{ \m{R}}$\hspace{\projspace}}
\def\LtilR{\hspace{\projspace}$\tilde{\m{L}}$\hspace{\projspace} \nodepart{two}\hspace{\projspace}$\m{R}$\hspace{\projspace}}
\def\LRtilde{\hspace{\projspace}${\m{L}}$\hspace{\projspace} \nodepart{two}\hspace{\projspace}$\tilde{\m{R}}$\hspace{\projspace}}
\newcommand\bitdistance{1 cm}
\newcommand\transfborder{0.1cm}
\newcommand\transfsize{\bitsize+\bitdistance-\transfborder}
\tikzmath{
\transfdistance=2*\bitdistance +2*\bitsize;
}
\tikzmath{
\transfoffset= 0.25*\bitdistance - \transfborder*0.5;
\halfsize=\bitsize*0.5;
\halfdistance=\bitdistance*0.5;
}
\newcommand\hdist{\bitsize+\bitdistance+\halfsize}
\newcommand\nhdist{-\bitsize-\bitdistance-\halfsize}
\newcommand\halfleft{-\bitsize-\transfoffset + \transfsize*0.5-\transfborder}
\tikzset{
  local/.style 2 args={
    rectangle, 
    rounded corners=10pt, 
    minimum height={#1},
    minimum width={#2},
    align=center, 
    line width=1pt,
    draw=white,
    font=\color{black}\sffamily, 
    fill=myblue1
  },
  local/.default={0.6 cm}{0.6 cm},
  M1/.style 2 args={
 	rectangle,
 	rounded corners=10pt ,
 	align=center,
 	minimum height={#1},
   	 minimum width={#2},
 	line width =1pt,
 	draw=black,
    font=\color{black}\sffamily, 
    fill=myblue1,
  },
  M1/.default={\bitsize}{\transfsize},
 M2/.style 2 args={
 	rectangle,
 	rounded corners=10pt ,
 	align=center,
 	minimum height={#1},
   	 minimum width={#2},
 	line width =1pt,
 	draw=black,
    font=\color{black}\sffamily, 
    fill=myblue2
  },
  M2/.default={\bitsize}{\transfsize},
 G/.style 2 args={
 	rectangle,
 	rounded corners=10pt ,
 	align=center,
 	minimum height={#1},
   	 minimum width={#2},
 	line width =1pt,
 	draw=black,
    font=\color{black}\sffamily, 
    fill=mycolor
  },
  G/.default={\bitsize}{\transfsize},
  bit/.style 2 args= {  
    rectangle, 
    rounded corners=10pt, 
    minimum height={#1},
    minimum width={#2},
    align=center, 
    line width=1pt,
    draw=black,
    font=\color{black}\sffamily, 
    fill=White
   },
 bit/.default={\bitsize}{\bitsize},
 bittil/.style 2 args= {  
    rectangle, 
    rounded corners=10pt,
    drop shadow, 
    minimum height={#1},
    minimum width={#1},
    align=center, 
    line width=1pt,
    draw=black,
    font=\color{black}\sffamily, 
    fill={#2}
   },
 tassel/.style 2 args= {  
    rectangle, 
    rounded corners=10pt,
    minimum height={#1},
    minimum width={#1},
    align=center, 
    line width=1pt,
    opacity={#2},
    draw=black,
    font=\color{black}\sffamily, 
    fill=tassellationcol
   },
  U/.style= {  
    rectangle, 
    rounded corners=7pt, 
    minimum height={\bitsize},
    minimum width={\bitsize},
    align=center, 
    line width=1pt,
    draw=black,
    fill=Ucolor,
    font=\color{black}\sffamily, 
   },
  LRbit/.style = {  
    rectangle, 
    rounded corners=10pt, 
    minimum height={\bitsize},
    minimum width={\transfsize},
    align=center, 
    line width=1pt,
    draw=black,
    fill=gray,
    font=\color{black}\sffamily, 
   },
   pi/.style 2 args={
 	rectangle,
 	rounded corners=10pt ,
 	align=center,
 	minimum height={#1},
   	 minimum width={#2},
 	line width =1pt,
 	draw=black,
	opacity=1,
    font=\color{black}\sffamily, 
    fill=yellow,
  },
  pi/.default={\bitsize}{\transfsize},
LR/.style ={
 	rectangle split,
	minimum height=0.6 cm,
   	minimum width =7.5cm,
	rectangle split horizontal,
	rectangle split parts=2,
 	rounded corners=10pt ,
 	align=center,
 	line width =1pt,
 	draw=black,
	opacity=1,
    font=\color{black}\sffamily, 
    fill=yellow,
 	text centered,
  },
  swap/.pic ={
	\draw (-\bitdistance,-1)--(-\bitdistance,-0.75)--(\bitdistance,0.75)--(\bitdistance,1);
	\draw (\bitdistance,-1)--(\bitdistance,-0.75)--(-\bitdistance,0.75)--(-\bitdistance,1);  
  },
Cif/.style 2 args={
 	rectangle,
 	rounded corners=10pt ,
 	align=center,
 	minimum height={#1},
   	 minimum width={#2},
 	line width =1pt,
 	draw=black,
    font=\color{black}\sffamily, 
    fill=cifcolor
  },
Cif/.default={\bitsize}{\transfsize},
Mswap/.style 2 args={
 	rectangle,
 	rounded corners=10pt ,
 	align=center,
 	minimum height={#1},
   	 minimum width={#2},
 	line width =1pt,
 	draw=black,
    font=\color{black}\sffamily, 
    fill=Mswapcolor
  },
  Mswap/.default={\bitsize}{\transfsize},
phase/.style 2 args={
 	rectangle,
 	rounded corners=10pt ,
 	align=center,
 	minimum height={\bitsize},
   	 minimum width={#1},
 	line width =1pt,
 	draw=black,
    font=\color{black}\sffamily, 
    fill={#2}
  },
Op/.style 2 args={
	circle,
	draw=black,
	line width=1pt,
	font=\color{black}\sffamily ,
	minimum height={#1},
   	 minimum width={#2},
   	 fill=White,
 },
 Op/.default={\bitsize*0.25}{\bitsize*0.25},
 J/.pic={
  \draw (-\bitdistance,\nhdist)--(0,\nhdist*0.5)--(\bitdistance,\nhdist),},
 }
\usetikzlibrary {arrows.meta}
\usetikzlibrary{%
    decorations.pathreplacing,%
    decorations.pathmorphing,%
    arrows,
    arrows.meta,
    positioning,
    calc,
    shapes,
    shadows,
    shapes.geometric
    }
    
\tikzset{
  bit/.style={
  rectangle,
  minimum height=0.4cm,
  minimum width=0.4cm,
  draw=black,
  line width=1pt
  }
}
\tikzstyle{new edge style 0}=[-, fill=white, draw={rgb,255: red,249; green,20; blue,0}]
\tikzstyle{new edge style 1}=[-, fill=black]
\tikzstyle{new edge style 2}=[-, fill={rgb,255: red,255; green,247; blue,0}, draw={rgb,255: red,162; green,164; blue,60}]
\tikzstyle{new edge style 3}=[-, fill={rgb,255: red,131; green,131; blue,131}]
\tikzstyle{new edge style 4}=[-, fill={rgb,255: red,0; green,209; blue,255}, draw={rgb,255: red,49; green,57; blue,208}]

\begin{document}
\title{Renormalisation of Fermionic Cellular Automata}

\author{Lorenzo Siro
  \surname{Trezzini}} \email[]{lorenzosiro.trezzini01@universitadipavia.it}
\affiliation{QUIT Group, Dipartimento di Fisica, Universit\`a di Pavia, via
  Bassi 6, 27100 Pavia} \affiliation{Istituto Nazionale di Fisica
  Nucleare, Gruppo IV, via Bassi 6, 27100 Pavia} 
  \author{Andrea
  \surname{Pizzamiglio}} \email[]{andrea.pizzamiglio01@universitadipavia.it}
\affiliation{QUIT Group, Dipartimento di Fisica, Universit\`a di Pavia, via
  Bassi 6, 27100 Pavia} \affiliation{Istituto Nazionale di Fisica
  Nucleare, Gruppo IV, via Bassi 6, 27100 Pavia}
\author{Alessandro 
  \surname{Bisio}} \email[]{alessandro.bisio@unipv.it}
\affiliation{QUIT Group, Dipartimento di Fisica, Universit\`a di Pavia, via
  Bassi 6, 27100 Pavia} \affiliation{Istituto Nazionale di Fisica
  Nucleare, Gruppo IV, via Bassi 6, 27100 Pavia} 
\author{Paolo
  \surname{Perinotti}} \email[]{paolo.perinotti@unipv.it}
\affiliation{QUIT Group, Dipartimento di Fisica, Universit\`a di Pavia, via
  Bassi 6, 27100 Pavia} \affiliation{Istituto Nazionale di Fisica
  Nucleare, Gruppo IV, via Bassi 6, 27100 Pavia} 

\begin{abstract} 
We present an exact renormalisation scheme for fermionic cellular automata on hypercubic lattices. By grouping neighbouring cells into tiles and selecting subspaces within them, multiple evolution steps on the original system correspond to a single step of an effective automaton acting on the subspaces. We derive a necessary and sufficient condition for renormalisability and fully characterise the renormalisation flow for two-cell tiles and two time steps of nearest-neighbour fermionic automata on a chain of spinless modes, identifying all fixed points.
\end{abstract} 

\maketitle
\section{Introduction}

Fermionic cellular automata (FCA) represent local, discrete-time, and reversible dynamics of fermionic modes on a lattice. FCA naturally arise in a variety of foundational and applied contexts---from digital formulations of quantum field theory \cite{bisio2015free, bisio2021scattering, bisio2025perturbative, bakircioglu2025fermion, arrighi2014dirac, brun2025, gupta2025dirac}  to the engineering of quantum phases and simulation platforms \cite{fidkowski2019interacting, stephen2019subsystem, ballarin2024digital, jotzu2014experimental, farrelly2020review}. Their fully discrete nature makes them ideally suited for physically implementable, digital realisations of quantum many-body dynamics.

A natural step forward is to understand how the dynamics of an FCA transform across scales. Renormalisation theory has long addressed the problem of multiscale dynamics \cite{10.1093/acprof:oso/9780199227198.001.0001,PhysRevB.4.3174, PhysRevB.4.3184, PhysicsPhysiqueFizika.2.263, PhysRevLett.99.220405}. Crucially, we are interested not in flowing to a continuum theory, but in mapping an FCA into another effective FCA. This choice has several advantages. First, it preserves the discrete, local features of the dynamics, ensuring that the coarse-grained theory remains physically implementable on digital quantum hardware. Second, it turns renormalisation into an algebraic operation on the space of FCA, allowing large-scale behavior to be characterised through finite, computable data rather than asymptotic limits. Finally, the resulting coarse-grained update rules provide a direct route to simplified simulations and cost-efficient control of complex fermionic dynamics. This perspective motivates the construction of an exact blocking-and-projection scheme that remains fully within the FCA framework.

Specifically, in this work, we extend the renormalisation scheme for quantum cellular automata introduced in \cite{trezzini2024renormalisationquantumcellularautomata} to fermionic systems. The prescription---when applies---maps a microscopic FCA into an effective coarse-grained FCA, defined on a blocked space-time lattice and acting on a reduced fermionic algebra. 
This, however, comes at a cost. According to our scheme, for the finer automaton to be renormalisable into a coarser one, it must preserve---on the chosen spacetime blocking---the subspace of degrees of freedom on which the effective macroscopic automaton acts. Consequently, both the renormalisability of the original FCA and the resulting renormalised automaton (if any) depend sensitively on the space-time blocking scale and on the specific choice of the subspace of degrees of freedom. When the choice is consistent with our prescription, one obtains a renormalisation flow in the space of FCA.

A key result of our construction is that renormalisability of an FCA is governed by a single finite algebraic condition:  renormalisability reduces to a matrix equation linking the local evolution operator of the FCA to the projection selecting the macroscopic degrees of freedom. This provides a sharp and computationally tractable criterion for identifying renormalisable automata and their corresponding renormalisation flows, turning  a structural property of the dynamics into a concrete algebraic test.

We analyse in detail the renormalisation flow for the case of one-dimensional chains of single (spinless) fermionic modes per site, a minimal setting that still supports rich dynamical behavior, as non-trivial topological phases. Spinless fermions offer a computationally efficient effective description for capturing the universal features of fermionic systems where spin is not dynamically relevant, as suppressed or decoupled by strong external fields, symmetry constraints, or energy scale separation. This scenario arises, for instance, in one-dimensional systems subject to strong Zeeman splitting, spin-selective interactions, or strong spin-orbit coupling, where only a single spin-polarised subband remains near the Fermi level. Under these conditions, the low-energy physics is accurately captured by projecting out the spin sector, resulting in an effective model of spinless fermions. Such a description underlies the paradigmatic models of topological superconductivity, such as the Kitaev chain and spinless p-wave superconductors  (for a comprehensive discussion, see Ref.\cite{alicea2012new} and references therein).

\vspace{0.3cm}
This framework opens up several concrete applications, outlined below and reserved for future investigation.

First, in the context of quantum simulation, many platforms---such as cold atoms in optical lattices or trapped ions---naturally realise Floquet evolutions that can be modeled as FCA \cite{jotzu2014experimental, rudner2020band}. Our coarse-graining procedure may provide a systematic way to derive effective stroboscopic dynamics over larger temporal and spatial scales for a restricted set of degrees of freedom. This can inform the design of Floquet dynamics that isolate and reproduce only the selected emergent behavior. For example, in periodically driven optical lattices, where fermionic modes experience engineered hopping and interaction terms, the microscopic control pulses can be complex, but their collective effect over multiple time steps may be approximated by a simpler effective FCA. 

Second, in digital quantum simulation, fermionic dynamics such as the Fermi-Hubbard model or lattice gauge theories are encoded into quantum circuits, often using Jordan-Wigner~\cite{Jordan:1928aa} or Bravyi-Kitaev~\cite{BRAVYI2002210} transformations, Suzuki-Trotter decomposition, and rishonic representation \cite{PhysRevX.10.041040, ballarin2024digital}. The resulting circuits are resource-intensive, both in depth and in fermionic parity management. By applying our renormalisation scheme to these  simulations, one can identify coarse-grained FCA that approximate the long-time behavior of the system while acting on fewer effective modes. This can lead to reduced circuit depth and gate count, making the simulation more feasible on near-term quantum hardware. 

Third, our approach may be relevant for studying dynamical (Floquet) topological phases in fermionic systems. Some FCA are known to be characterised by nontrivial topological indices (as winding numbers),  support protected edge modes, and exhibit chiral transport \cite{fidkowski2019interacting, PhysRevB.82.235114}. 
Since our procedure maps FCA to FCA, it preserves the algebraic structure needed to track such topological invariants across scales.
Our renormalisation scheme can alter the topological index for suitable choices of coarse-graining, thus in general it may not be appropriate for identifying scale-invariant phases (e.g. as those discussed in \cite{haah2013latticequantumcodesexotic,RevModPhys.93.045003,PhysRevB.82.155138}). Nonetheless, this very feature can be exploited constructively: by design, the scheme can synthesise specific phases as effective collective dynamics of microscopic ones, e.g., generating nontrivial topological indices from shallow circuits \footnote{
Shallow circuits realising invertible topological order---e.g., by applying an FCA and its inverse in parallel to two stacked copies of a system---can yield nontrivial topological indices once selected stacked degrees of freedom are projected out. Despite this common goal, the two index-distillation schemes constitute distinct protocols.}\cite{PhysRevB.109.075116, PhysRevLett.127.220503, PhysRevLett.97.050401}.

Finally, the renormalisation of FCA provides a foundation for multi-scale algorithm design in fermionic quantum computation.  Our coarse-graining procedure yields a systematic way to build hierarchical quantum circuits that encode fermionic dynamics at different resolutions. This could be exploited to design hybrid analog-digital simulation protocols, where fine-grained dynamics (as short-range interactions) is executed in hardware, while coarse-grained behavior (as collective modes) is captured algorithmically through the renormalised circuit layers.
This separation enables more scalable and resource-efficient simulations, allowing large or strongly correlated fermionic systems to be studied with significantly reduced quantum hardware requirements.

Importantly, because our renormalisation flow remains within the space of FCA, the resulting effective dynamics are always local and unitary and allows for direct implementation in quantum systems with finite control resolution. This distinguishes our approach from continuum renormalisation group schemes, which may produce non-unitary or non-local effective theories.

\vspace{0.3cm}
The paper is organised as follows. In Sec. \ref{sec:qca}, we introduce fermionic cellular automata, defined as quantum cellular automata in the sense of Ref. \cite{schumacher2004reversiblequantumcellularautomata}, acting on fermionic systems. We begin by recalling the basic definitions and properties of FCA and then summarise the classification of nearest-neighbour FCA of spinless fermionic modes on a one-dimensional chain \cite{trezzini2025fermioniccellularautomatadimension}. We further rigorously present the wrapping lemma \cite{schumacher2004reversiblequantumcellularautomata}---a standard technique that allows the evolution of an automaton on an infinite lattice to be represented on a finite one. We provide this lemma at this stage as it plays a crucial role in formulating precise statements about renormalisation.
In Sec. \ref{sec:CG}, we first define the coarse-graining of lattice algebras and the renormalisation equation of FCA, and then we derive the corresponding renormalisation conditions adjusting Ref.\cite{trezzini2024renormalisationquantumcellularautomata} to the fermionic framework. In Sec. \ref{sec:FDQC}, we specialise to one-dimensional fermionic circuits, deriving stronger constraints and structural insights from the general renormalisation equation of Sec. \ref{sec:CG} for specific choices of coarse-graining. Finally, in Sec. \ref{sec:Qubits}, using the same coarse-graining scheme of the former section, we solve the renormalisability equation for nearest-neighbour FCA of spinless fermions in one spatial dimension and compute the resulting renormalisation flow in FCA space.

\section{Fermionic Cellular Automata}\label{sec:qca}
A Fermionic Cellular Automaton (FCA) is a discrete-time, reversible, local, and translation-invariant evolution of a lattice $\mathbb G$ of cells, each hosting a finite number of fermionic modes.
Locality condition requires that, at each step of
the evolution, a cell $x \in \mathbb G$ interacts
with a finite set $\mathcal{N}_x$ of
cells, called the \emph{neighourhood} of $x$. Thus, the speed of information propagation in a FCA is finite,
defining past and future cones of causal influence.  
Translation invariance means  that every cell is updated with the same rule, hence $\mathcal N_x=\mathcal N$~\footnote{FCA can be defined 
without the translation-invariance property (see e.g.~Ref.~\cite{Index}).}.

In this paper we focus on the case where the 
lattice $\mathbb G$ is a hypercubic lattice~\footnote{Typically, one considers the lattice as a graph, the edges corresponding to neighbourhood relations. One can prove that the translation-invariance requirement on an FCA makes the graph
of causal connections the \emph{Cayley graph} of some finitely presented group $G$~\cite{arrighi:hal-01785458,hadeler2017cellular,PhysRevA.90.062106,DAriano:2016aa,Perinotti2020cellularautomatain}, that in 
the present case is $\mathbb G$ imagined as an abelian  group.}
\begin{align}
\mathbb G \coloneqq \tor, \quad  m_t \in \mathbb{N} \cup \infty, \quad s < \infty, \label{eq:latticedefinition}
\end{align}
where $\mathbb Z_{m_t}$ denotes the additive group
modulo $m_t$, and $\mathbb Z_\infty \coloneqq \mathbb Z$. Each cell of the lattice
is labeled by a vector $x \in \tor $ of
integer numbers.
This notation allows for a unified treatment of
$\mathbb G$, regardless of whether it is a finite or infinite lattice.

Since we will deal with lattices of infinitely many fermionic systems, it is convenient to introduce FCA in the Heisenberg picture~\cite{schumacher2004reversiblequantumcellularautomata}, defining the evolution by its action on the algebra of operators rather than on states. The quasi-local algebra, which is the norm completion of the union of all operator algebras on finite lattice regions, provides a natural and general setting; whereas states realise inequivalent representations of the algebra of operators in the infinite-system limit~\cite{bratteli1987operator, haag1992local}.

We associate to each cell $x\in \mathbb G$ a $C^*$-algebra $\m A_x$, with $\norm{\cdot}$ the operator norm. When $\Lambda \subset \mathbb G$ is a finite subset, we denote by $\mathsf{A}(\Lambda)=\bigotimes_{x\in\Lambda}\mathsf{A}_x$ the algebra of $\Lambda$, where $\otimes$ is a suitable algebraic tensor product. If
we consider two different finite subsets $\Lambda_1,\Lambda_2\subset\G$, we can regard
$\mathsf{A}(\Lambda_1)$ as a subalgebra of $\mathsf{A}(\Lambda_1\cup\Lambda_2)$ by tensoring
each element $O\in\mathsf{A}(\Lambda_1)$ with the identity over $\Lambda_2\setminus \Lambda_1$, namely $O\otimes I_{\Lambda_2 \setminus \Lambda_1}\in\mathsf{A}(\Lambda_1\cup\Lambda_2)$. In
this way the product $O_1O_2$ with $O_i\in\mathsf{A}(\Lambda_i)$ is a well defined
element of $\mathsf{A}(\Lambda_1\cup\Lambda_2)$.   The algebra $\mathsf{A}^\mathrm{loc}(\G)$ of all elements that are non-trivial over any finite set of cells is called \textit{local algebra}
and is naturally endowed with the norm $\|O\|_\Lambda$ for $O\in\mathsf A(\Lambda)$. 
Its completion in the operator norm (uniform) topology is called \emph{quasi-local} algebra $\mathsf{A}(\G) \coloneqq \overline{\bigcup_{\Lambda \subset \G}\m A(\Lambda)}^{\norm{\cdot}}$ \footnote{
The local algebra may be completed with respect to different topologies.  Norm topology is a natural and representation-independent choice, as it does not rely on specifying an underlying Hilbert space. Physically, the operator norm is related to the maximum success probability in the discrimination of the two processes represented by the operators. In any case, because the coarse-graining map is continuous and the renormalisability condition Eq.~\eqref{prp:thecoarsegrainingisfinite} can be verified directly on the local algebra, our renormalisation scheme remains independent of the chosen completion.}.  Clearly, for finite lattices $\mathsf{A}(\G)=\mathsf{A}^\mathrm{loc}(\G)$ is a finite dimensional algebra.

In this paper the algebra of a cell will be a finite-dimensional $\Z_2$-graded complex $C^*$-algebra $\mathsf A_x\coloneqq(\mathsf A^0_x, \mathsf A_x^1)$, and the tensor product will be the graded tensor product $\gt$.  We shall refer to as subalgebras those subalgebras that maintain $C^*$ and $\Z_2$-grading properties.
Both the local and quasi-local algebras are $\Z_2$-graded $C^*$-algebras \cite{bratteli1987operator}. In this framework the standard commutator is replaced by a graded commutator $\gc{\cdot}{\cdot}$. We denote by $g(O)=p$ the parity (grade) of an operator $O\in \m A^p_\Lambda$.  We refer the reader to Appendix~\ref{app:z2algebras} for a detailed discussion of $\Z_2$-graded and $C^*$ algebras.

The algebra $\mathsf A_x$ can be seen as the CAR algebra generated by $d$ fermionic modes $\{b_{x,k},\,b_{x,k}^{\dagger} \;|\; k=1,\ldots,d\}$  satisfying the canonical anticommutation relations
\begin{gather*}
  \{b_{x,j},b_{y,k}^{\dagger}\} =\delta_{xy} \delta_{jk} I_{x,k}\quad,\quad \{b_{x,j},b_{y,k}\}=0\,,
\end{gather*}
with $x,y \in\G$ and $I_{x,k}$ the identity on the subalgebra generated by the $k$-th fermionic mode in the cell $x$. By definition $\A_x$ is identified by a pair $\mathsf A_x\coloneqq(\mathsf A^0_x, \mathsf A_x^1)$, where $\mathsf A^0_x$ ($\mathsf A^1_x$) is the span of even (odd) monomials in $b_{x,k},b_{x,j}^{ \dagger}$, and $O \in \m A_x$ means $O \in \m A^p_x$ for $p=0$ (in which case we say that $O$ is even) or $p=1$ (in which case we say that $O$ is odd).

It is convenient to introduce the following linear combinations
\begin{gather}\label{eq:Fpauli}
	X_{x,k} \coloneqq b_{x,k}^{\dagger}+b_{x,k}, \quad\quad Y_{x,k} \coloneqq i(b^{\,\dag}_{x,k}-b_{x,k}),\\
    Z_{x,k} \coloneqq -i X_{x,k} Y_{x,k}  =  b_{x,k} b_{x,k}^{\dagger}-b_{x,k}^{\dagger} b_{x,k}\,,
\end{gather}
which are called \emph{Majorana mode generators}, $X_{x,k}, Y_{x,k} \in \A^1_x$ and $Z_{x,k} \in \A^0_x$. Although the latter can be represented through Pauli matrices in the Jordan-Wigner representation~\cite{Jordan:1928aa}, they should not be mistaken for conventional qubit Pauli operators, which commute on distinct lattice sites. In contrast, Majorana generators on different sites graded-commute
\begin{align*}
    M^J_{x,a}M^K_{y,b}=&\,\delta_{xy}\delta_{ab}\left[\delta_{JK} I_{x,a} + i \varepsilon_{JKL}\,M^L_{x,a}\right]\,+\\
    &+(-)^{g(M^J)g(M^K)}(1-\delta_{xy}\delta_{ab})M^K_{y,a} M^J_{x,b}\,,
\end{align*}
where $x,y \in \mathbb G$, $a,b\in\{1,\ldots d\}$, $J,K,L\in\{1,2,3\}$ label the components of $\mathbf{M}_{x,a}=(X_{x,a},Y_{x,a},Z_{x,a})$, and $\varepsilon_{JKL}$ is the Levi-Civita symbol. Any set of $n$ odd Majorana generators $M$ generates a complex Clifford algebra (see Appendix \ref{app:z2algebras}) that we denote by $\Cl^M_n(\C)$, and we have the following isomorphisms \[\A_x \cong \Cl_{2d}(\C)\cong \bigboxtimes_{k=1}^d \Cl^{X_{x,k}}_1(\C) \gt \Cl^{Y_{x,k}}_1(\C)\,.\]
We now introduce two standard operators for later use.
We denote by \[n_{x,k} \coloneqq b_{x,k}^{\dagger}b_{x,k}=\frac{1}{2}\left(I_{x,k} - Z_{x,k} \right)\] the number operator of mode $k$ at site $x$, and by
\begin{align*}
    Q_x \coloneqq (-)^{\sum_{k=1}^d n_{x,k}}= \bigboxtimes_k Z_{x,k}
\end{align*}
the parity operator of a cell, which acts as $Q_xOQ_x=(-)^{g(O)}O$ for any $O \in \A_x$. We denote by $Q_\Lambda= \bigboxtimes_{x \in \Lambda} Q_x$ the parity operator of the algebra $\m A(\Lambda)$.

By defining
\begin{gather*}
    \ketbra{0}{0}\coloneqq \frac{1}{2} (I+Z)\,,\quad \ketbra{0}{1}  \coloneqq \frac{1}{2} (X+iY)\,,\\
    \ketbra{1}{0}\coloneqq \frac{1}{2} (X-iY)\,,\quad \ketbra{1}{1} \coloneqq  \frac{1}{2} (I-Z)\,,\\
    \ketbra{a}{b}\ketbra{c}{d}=\delta_{bc}\ketbra{a}{d}\,,
\end{gather*}
 any $O \in \m A_x$ can be written as
\begin{gather*}
    O= \sum_{\mathbf a,\mathbf b\in\{0,1\}^d} c_{\mathbf a\mathbf b} \ketbra{\mathbf{a}}{\mathbf b}\,,\\
    \ketbra{\mathbf a}{\mathbf b} \coloneqq \bigboxtimes_{k=1}^d \ketbra{a_k}{b_k}\,,
\end{gather*}
with $c_{\mathbf a\mathbf b} \in \C$, $a_k,b_k \in \{0,1\}$, and $g(\ketbra{a_k}{b_k})=a_k\oplus b_k$. 
 If the coefficients $c_{\mathbf a\mathbf b}$ of the expansion of $O$ are factorised as $c_{\mathbf a\mathbf b}=f_\mathbf ag_\mathbf b$, we will write
\begin{align*}
O=\ketbra{f}{g}.
\end{align*}    
It is worth noting that the operators 
$\ketbra{{f}}{{g}}$ obey the standard ket–bra algebra, i.e.
\begin{align*}
    \ketbra{{f}}{{g}}\ketbra{{h}}{{k}}=\left(\sum_{\mathbf a\in\{0,1\}^d}g_\mathbf ah_\mathbf a\right)\ketbra{f}{k},
\end{align*}
and we can even define
\begin{align*}
    \braket{f}{g}=\sum_{\mathbf a\in\{0,1\}^d}g_\mathbf ah_\mathbf a,
\end{align*}
so that
\begin{align*}
    \ketbra{{f}}{{g}}\ketbra{{h}}{{k}}=\braket{g}{h}\ketbra{f}{k},
\end{align*}
although the individual bra and ket elements are not defined and carry no individual meaning within this formalism \footnote{Although any finite-dimensional $C^*$-algebra can be represented (uniquely up to isomorphisms) as bounded operators on a finite-dimensional Hilbert space, it is often preferable to treat the algebra abstractly. The abstract description makes all constructions intrinsic to the algebra itself, without committing to a particular Hilbert space. In particular, for $\Z_2$-graded algebras, the grading is naturally encoded as an internal automorphism of the algebra rather than as a decomposition of a specific Hilbert space. Introducing abstract elements such as $\ketbra{f}{g}$ that satisfy the usual algebraic relations captures the full operator structure without requiring the individual kets or bras to have meaning as vectors in a Hilbert space. This viewpoint defines the algebraic content while avoiding arbitrary representational choices, extending seamlessly to more general (e.g. infinite-dimensional or categorical) settings.}.

\vspace{0.4cm}
In order to define a discrete-time, translation-invariant dynamics on $\A(\G)$, we first define the map that implements lattice translations.
\begin{definition}[Shift]\label{def:shift} The shift by $x\in \G$ $\tau_x:\A(\mathbb G)\to \A(\mathbb G)$ is defined by its action on Majorana mode generators as
\begin{align}
    \tau_x(M^J_{y,a})\coloneqq M^J_{x+y,a}.
\end{align}
\end{definition}
If $O \in \mathsf{A}_\Lambda$ is
supported on the region $\Lambda$, then
$\tau_x(O) \in \mathsf{A}_{\Lambda + x}$ is supported on the region
$\Lambda + x := \{y+x\; |\; y \in
\Lambda \}$.
We are now in position to define a Fermonic Cellular Automaton. 
 \begin{definition}\label{def:QCA} 
   A $\emph{Fermionic Cellular Automaton}$ (FCA)
   over a lattice $\G$ with finite neighbourhood
   scheme $\mathcal{N} \subset\G$ is an
  automorphism of the quasi-local
   algebra
   $\tT:\mathsf{A}(\G) \longrightarrow
   \mathsf{A}(\G)$  such that
   \begin{itemize}
   \item  $\tT(\mathsf{A}(\Lambda))\subset
   \mathsf{A}(\Lambda+\mathcal{N}) \quad
   \forall \Lambda\subset \G$ (locality),
   \item $\tT\circ\tau_x=\tau_x\circ\tT \quad \forall x\in \G$ (translation-invariance),
 \end{itemize}
 where $\Lambda+\mathcal{N} := \{ y+x\; |\; y \in
\Lambda,\, x \in \mathcal{N}  \}$, and $\circ$ denotes the composition of two algebra homomorphisms.
   The homomorphism $\tT_0:\mathsf{A}_0\longrightarrow
   \mathsf{A}(\mathcal{N})$ given by the
   restriction of the FCA to the cell $x=0$ is called the
   $\emph{transition rule}$.
\end{definition}
Notice that $N$ iterations of the FCA $\tT$ define an FCA $\tT^N$ whose neighbourhood scheme is possibly $N$ times larger $ \sum_{t=1}^N\mathcal N$. In this paper we assume FCA to be translation-invariant; in general one can drop this assumption.

Owing to the counterpart of the Wedderburn-Artin theorem for $\mathbb{Z}_2$-graded algebras \cite{Brunshidle01112012}, the structure theorem of Ref. \cite{schumacher2004reversiblequantumcellularautomata} for Quantum Cellular Automata (QCA) extends to fermionic cellular automata, establishing a one-to-one correspondence between the local transition rule $\tT_0$ and the global evolution $\tT$.
\begin{proposition}
	\label{thm:locglo}
\hspace*{1pt}
\begin{enumerate}
	\item The FCA $\tT$ is uniquely determined
	by the transition rule $\tT_0$ as
    \begin{equation*}
    	\tT\left(\bigboxtimes_{x\in \mathbb G}\A_x\right)=\prod_{x\in \mathbb G}\tT_x(\A_x)\,,
    	\label{eq:autom}
\end{equation*}
where
$\tT_x(\mathsf{A}_x) \coloneqq {\tau}_x\tT_0{\tau}_{x}^{-1}(\mathsf{A}_x)$, and the same (immaterial) ordering of 
$\A_x$ is implicit in both products.
	\item A graded $*$-homomorphism $\tT_0 : \mathsf{A}_{0} \longrightarrow \mathsf{A}(\mathcal{N})$ is the transition
	rule of an FCA if and only if
    \begin{align*}
        \gc{\tT_{0}(\mathsf{A}_{0})}{\tau_x \tT_{0}\tau_x^{-1}(\mathsf{A}_{x})}=0 
    \end{align*}
    for all $x \in
	\mathbb G$ such that $\mathcal{N}\cap(\mathcal{N}+x)\neq \emptyset$.
\end{enumerate} 
\end{proposition}
The proof (analogous to \cite{schumacher2004reversiblequantumcellularautomata}) uses that, in the finite-dimensional case, any graded $*$-homomorphism from a simple  $\Z_2$-graded $C^*$-subalgebra to the full simple $\Z_2$-graded $C^*$-algebra is an isomorphic embedding through unitary conjugation (see Appendix \ref{app:z2algebras})
\begin{equation}
    \tT_0(O)=U^\dag(I_{\mathcal N\setminus \{0\}}\gt O)U\quad\forall O\in\mathsf A_0\,,
    \label{eq:ev}
\end{equation} 
for a suitable unitary element $U \in \mathsf A({\mathcal
  N\cup\{0\}})$.

 We now define the notion of a fermionic gate that features in the most important cases of FCA.
\begin{definition}[Fermionic Gate]
A \textit{Fermionic Gate} $\mathcal F$ acting over a local algebra $\m A({\Lambda})$ is a inner $*$-automorphism of $\m A({\Lambda})$:
\begin{align*}
\mathcal F(O) = F^\dag O F\,,\quad \forall\,O\in \m A({\Lambda})\,,
\end{align*}
where $F\in \m A({\Lambda})$ is a unitary operator.
\end{definition}
\begin{remark}\label{rmk:Kraus}
Notice that fermionic gates have a local action. This might seem false if, e.g.,  one considers the bipartite operator $O = O_1 \gt I \in \m A({\Lambda_1}) \gt \m A_{x}$ with $O_1$ an odd operator, and the unitary operator  $F=I \gt X\in \m A({\Lambda_1}) \gt \m A_{x}$, for which we have
\begin{align*}
&\mathcal F(O)  = ( I\gt {X})(O_1\gt I )(I \gt {X}) =\\
&= - O_1 \gt {X}^2= - O_1\gt  I.
\end{align*}
However, due to $\Z_2$-grading of the algebra---which corresponds to a superselection rule for physical states---the minus sign, as a global phase, is physically irrelevant. In fact, operators, viewed as Kraus operators of a completely positive map, are unaffected by any overall phase factor.
\end{remark}

The identity $\tI(O_x)=O_x$ and the shifts $\tau_x$ are trivial
examples of FCA. Another example of FCA is any (Finite-Depth) Fermionic Circuit (FDFC), whose action is decomposed in a sequence of local fermionic gates as follows
\begin{equation*}
\resizebox{0.85\hsize}{!}{
\begin{tikzpicture}[baseline=(O.base)]
\node (O) at (0,\hdist) {};
\foreach \t in  {0,...,8}{
\draw[black]  (\bitdistance*\t+\bitsize*\t,6*\hdist) -- (\bitdistance*\t+\bitsize*\t, \bitsize*0.5+2pt) ;
\node[U] at  (\bitdistance*\t+\bitsize*\t,7*\hdist) {$\mathcal U_d$} ;
}
\foreach \t in {0,...,2}
\node[phase={\transfsize+\bitsize}{phasecolor}] at (\transfdistance*0.5+\transfdistance*\t+\bitsize*\t+\bitdistance*\t,0) {$\mathcal U_a$};
\foreach \t in{0,...,3}{
\node[phase={\transfsize}{phasecolorbis}] at (\bitsize-\transfsize*0.5+\transfdistance*\t,\hdist) {$\mathcal U_b$};
}
\foreach \t in{0,1}{
\node[phase={2*\transfsize+\bitdistance}{phasecolortris}] at (\transfdistance*0.5+\bitdistance+\transfdistance*\t+2*\bitsize*\t+2*\bitdistance*\t,4*\hdist) {$\mathcal U_c$};
}
\draw[fill=phasecolorbis,rounded corners=2pt,draw=white,opacity=1] (-\bitsize+\transfdistance*4+\transfsize-\transfoffset,0.0001*\hdist)--(-\bitsize+\transfdistance*4,0.0001*\hdist)--(-\bitsize+\transfdistance*4,2*\hdist)--(-\bitsize+\transfdistance*4+\transfsize-\transfoffset,2*\hdist);
\draw[black,rounded corners=2pt, line width=1pt,opacity=1]  (-\bitsize+\transfdistance*4+\transfsize-\transfoffset,0.0001*\hdist)--(-\bitsize+\transfdistance*4,0.0001*\hdist)--(-\bitsize+\transfdistance*4,2*\hdist)--(-\bitsize+\transfdistance*4+\transfsize-\transfoffset,2*\hdist);
\draw[fill=phasecolortris,rounded corners=2pt,draw=white,opacity=1] (-\bitsize+\transfdistance*4+\transfsize-\transfoffset,3*\hdist)--(-\bitsize+\transfdistance*4,3*\hdist)--(-\bitsize+\transfdistance*4,5*\hdist)--(-\bitsize+\transfdistance*4+\transfsize-\transfoffset,5*\hdist);
\draw[black,rounded corners=2pt, line width=1pt,opacity=1]  (-\bitsize+\transfdistance*4+\transfsize-\transfoffset,3*\hdist)--(-\bitsize+\transfdistance*4,3*\hdist)--(-\bitsize+\transfdistance*4,5*\hdist)--(-\bitsize+\transfdistance*4+\transfsize-\transfoffset,5*\hdist);
\end{tikzpicture}}.
\end{equation*}
Wires represent systems, and the boxes represent unitary completely positive maps $\mathcal U_l(O)\coloneqq U_l^\dag O U_l$ for any $O \in \A(\Lambda)$, with $U_l \in \A(\Lambda)$ unitary. In the literature, the term \emph{finite-depth}---or \emph{constant-depth}, or \emph{shallow}---refers to circuit depth that is independent of lattice size and cell dimension.

\begin{remark} \label{rem:abstran}  
 The algebra $\mathsf A_x$ is isomorphic to the $\Z_2$-graded $C^*$-algebra $\M{2^{d-1}}{2^{d-1}}$, say via \[\transf H_\mathbb G:\mathsf A_x\to \M{2^{d-1}}{2^{d-1}}.\] 
 In the same way $\mathsf A(\mathcal N)$ is isomorphic to $\M{2^{d|\mathcal N|-1}}{2^{d|\mathcal N|-1}}$, via \[\tK_\mathbb G:\mathsf A(\mathcal N)\to\M{2^{d|\mathcal N|-1}}{2^{d|\mathcal N|-1}}.\] Thus, the transition rule of an FCA $\tT$ can be equivalently identified by the graded $*$-homomorphism $\tT_*\coloneqq \tK_\mathbb G\circ\tT_0\circ\tH^{-1}_\mathbb G$
\begin{align*}
\tT_* :\M{2^{d-1}}{2^{d-1}}\to\M{2^{d|\mathcal N|-1}}{2^{d|\mathcal N|-1}}.
\end{align*}
In this way, the restriction of $\tK_\mathbb G$ to $\tT_0(\mathsf A_0)$ is a graded $*$-isomorphism to $\tT_*(\M{2^{d-1}}{2^{d-1}})$.

\end{remark}

  In the following we use Eq.~\eqref{eq:ev} both for $\tT_0$ and $\tT_*$, leaving to the context the specification as to which is the case.

The algebraic approach avoids the need to specify the representation of the algebra of operators. For finitely many systems, i.e. when $\mathbb G$ is a finite lattice, this offers no benefits, as all representations are unitarily equivalent: $\tT $  acts as conjugation by some unitary operator with definite parity. However, when $\G$ is an
infinite lattice, $\mathsf{A}(\mathbb G)$ becomes an infinite-dimensional $C^*$-algebra that admits many inequivalent representations. 
In this case, working in a representation-independent framework is highly advantageous, as it allows us to prove results without being tied to a particular Hilbert space.

 \subsection{One Dimensional Nearest-Neighbour spinless Fermionic Cellular Automata}\label{NNFCA}
  In this section, we briefly review the classification of nearest-neighbour  (i.e. $\mathcal{N}=\{-1,0,1\}$) 
 FCA  on a line $\mathbb{Z}$, with cells containing a single (spinless) fermionic mode $\A_x\cong \M{1}{1}\cong \Cl_{2}(\C)$~\cite{trezzini2025fermioniccellularautomatadimension}. Those FCA can be fully classified by a topological invariant called \emph{index} \cite{fidkowski2019interacting}, which measures the information flux rate produced by the evolution $\tT$.
FCA with index $\ind( \tT ) = 1$ are finite-depth fermionic circuits (FDFCs) with one of the two following forms
 \begin{align}
   \label{eq:sw}
\tT=\begin{tikzpicture}[baseline=(O.base)]
\node (O) at (0,3.5*\hdist) {};
\foreach \t in  {0,...,5}{
\draw[black]  (\bitdistance*\t+\bitsize*\t,9*\hdist) -- (\bitdistance*\t+\bitsize*\t, \bitsize*0.5+2pt) ;
\node[U] at (\bitdistance*\t+\bitsize*\t,7*\hdist) {$\mathcal U$};
}
\foreach \t in {0,...,2}{
\node[Cif] at (\bitsize+\transfoffset+\transfdistance*\t,\hdist) {$\mathcal{C}_\phi$};
}
\foreach \t in {0,...,1}{
\node[Cif] at (\bitsize+\transfoffset+\transfsize*0.5+\transfdistance*\t,4*\hdist) {$\mathcal C_\phi$};
}

\draw[fill=cifcolor,rounded corners=10pt,draw=white,opacity=1] (\bitsize+\transfoffset+\transfsize*0.5+\transfdistance*2,3*\hdist)--(\bitsize+\transfoffset+\transfdistance*2,3*\hdist)--(\bitsize+\transfoffset+\transfdistance*2,5*\hdist)--(\bitsize+\transfoffset+\transfsize*0.5+\transfdistance*2,5*\hdist);
\draw[black,rounded corners=10pt, line width=1pt,opacity=1] (\bitsize+\transfoffset+\transfsize*0.5+\transfdistance*2,3*\hdist)--(\bitsize+\transfoffset+\transfdistance*2,3*\hdist)--(\bitsize+\transfoffset+\transfdistance*2,5*\hdist)--(\bitsize+\transfoffset+\transfsize*0.5+\transfdistance*2,5*\hdist);

\draw[fill=cifcolor,rounded corners=10pt,draw=white,opacity=1] (-\bitsize-\transfoffset,3*\hdist)--(\halfleft,3*\hdist)--(\halfleft,5*\hdist)--(-\bitsize-\transfoffset,5*\hdist);
\draw[black,rounded corners=10pt, line width=1pt,opacity=1] (-\bitsize-\transfoffset,3*\hdist)--(\halfleft,3*\hdist)--(\halfleft,5*\hdist)--(-\bitsize-\transfoffset,5*\hdist);

\end{tikzpicture}\,,
\\
 \label{eq:frk}
\tT=\begin{tikzpicture}[baseline=(O.base)]
\node (O) at (0,3.5*\hdist) {};
\foreach \t in  {0,...,5}{
\draw[black]  (\bitdistance*\t+\bitsize*\t,9*\hdist) -- (\bitdistance*\t+\bitsize*\t, \bitsize*0.5+2pt) ;
\node[U] at (\bitdistance*\t+\bitsize*\t,7*\hdist) {$\mathcal U$};
}
\foreach \t in {0,...,2}{
\node[Mswap] at (\bitsize+\transfoffset+\transfdistance*\t,\hdist) {$\mathcal{S}(X, Y)$};
}
\foreach \t in {0,...,1}{
\node[Mswap] at (\bitsize+\transfoffset+\transfsize*0.5+\transfdistance*\t,4*\hdist) {$\mathcal{S}(X, Y)$};
}

\draw[fill=Mswapcolor,rounded corners=10pt,draw=white,opacity=1] (\bitsize+\transfoffset+\transfsize*0.5+\transfdistance*2,3*\hdist)--(\bitsize+\transfoffset+\transfdistance*2,3*\hdist)--(\bitsize+\transfoffset+\transfdistance*2,5*\hdist)--(\bitsize+\transfoffset+\transfsize*0.5+\transfdistance*2,5*\hdist);
\draw[black,rounded corners=10pt, line width=1pt,opacity=1] (\bitsize+\transfoffset+\transfsize*0.5+\transfdistance*2,3*\hdist)--(\bitsize+\transfoffset+\transfdistance*2,3*\hdist)--(\bitsize+\transfoffset+\transfdistance*2,5*\hdist)--(\bitsize+\transfoffset+\transfsize*0.5+\transfdistance*2,5*\hdist);

\draw[fill=Mswapcolor,rounded corners=10pt,draw=white,opacity=1] (-\bitsize-\transfoffset,3*\hdist)--(\halfleft,3*\hdist)--(\halfleft,5*\hdist)--(-\bitsize-\transfoffset,5*\hdist);
\draw[black,rounded corners=10pt, line width=1pt,opacity=1] (-\bitsize-\transfoffset,3*\hdist)--(\halfleft,3*\hdist)--(\halfleft,5*\hdist)--(-\bitsize-\transfoffset,5*\hdist);

\end{tikzpicture}\,,
\end{align}
where $\mathcal U(O_x)\coloneqq U^\dag (O_x) U$ for a unitary operator $U\in \m A_x$, a cell-wise unitary gate, $\mathcal C_\phi(O_xO_{x+1})\coloneqq C_\phi^\dag (O_xO_{x+1}) C_\phi$, $\phi \in \R$, a \emph{controlled-phase} gate
\begin{align*}
    C_{\phi} \coloneqq& e^{i\phi\,n_x\gt\, n_{x+1}}=\\
    =&\frac{1}{2}\left[(I_x+Z_x)\gt I_{x+1}+(I_x-Z_x)\gt e^{\frac{\phi}{2}(I_{x+1}-Z_{x+1})}\right]\,,
\end{align*}
and $\mathcal{S}({X},{Y})(O_{\{x,{x+1}\}})\coloneqq S_{{X},{Y}}^\dag (O_{\{x,{x+1}\}}) S_{{X},{Y}}$ a \emph{Majorana Swap} gate
\begin{align}\label{eq:majo_swap}
S_{{X},{Y}}\coloneqq e^{\frac{\pi}{4}{X_x}\gt {Y_{x+1}}}\,,
\end{align}
with $O_{\{x,x+1\}}$ any operator supported on sites $x$ and $x+1$.
Notice that, while the first family corresponds to that of nearest-neighbour qubit QCA on $\Z$ with index equal to $1$, the second family is genuinely fermionic. We will refer to the former as the \textit{Schumacher-Werner FCA} \cite{schumacher2004reversiblequantumcellularautomata} and to the latter as the \textit{Forking automata} \cite{trezzini2025fermioniccellularautomatadimension}. 

On the other hand, FCA with $\ind\neq 1$ are either (modulo cell-wise unitaries) i) the \emph{shifts} $\tau_\pm$ with $\ind(\tau_\pm)=2^{\pm 1}$---consisting of the trivial embedding of each cell algebra into its left or right neighbouring cell---e.g. a right shift
\begin{equation}
\tau_+=\begin{tikzpicture}[baseline=(A0.base)]
\foreach \a in {0,...,6}{
\node (A\a ) at (\a cm,1.1) {};
\draw[black] ($(A\a)-(0,0.7)$)--($(A\a)-(0,0.2)$);
\draw[black] ($(A\a)-(0,0.2)$)--($(A\a)+(1,0.3)$);
\draw[black] ($(A\a)+(1,0.3)$)--($(A\a)+(1,0.7)$);
}
\node at ($(A0)-(0.8,0)$) {...};
\node at ($(A6)+(1.5,0)$) {...};
\end{tikzpicture};
\end{equation}
or ii) the \textit{Majorana shifts} $\sigma_\pm$ with $\ind(\sigma_\pm)=2^{\pm 1/2}$---which act as
\begin{gather}
\begin{gathered}
\sigma_\pm(  X  _x) =   Y  _{x}\,,\\
\sigma_\pm(  Y  _{x}) =  X  _{x\pm 1}\,.
\end{gathered}
\end{gather}
Note that, as the algebra of each cell is generated by two Majorana mode generators, providing the action of an FCA on the generators, e.g. $X_x$ and $Y_x$, completely specifies its action on the whole algebra $\m A_x$.

\subsection{The Wrapping Lemma}
The Wrapping Lemma (see Ref. \cite{Perinotti2020cellularautomatain}) relates  
FCA over infinite and finite lattices establishing a bijective correspondence of transition rules. By Proposition \ref{thm:locglo} an FCA is uniquely determined  by a finite set of graded commutators identified by neighbourhoods overlaps. One ``wraps'' a sufficiently large finite sub-lattice of an FCA on $\Z^s$ by imposing suitably periodic boundary conditions, obtaining an FCA on a discrete torus $\mathbb L$. 
The ``unwrapped'' and ``wrapped'' FCA  have the same transition rule if no new overlap between neighbourhoods is introduced after wrapping the lattice.
We denote the smallest wrapping that meets this condition by $\mathbb L_0$ and we will say that $\mathbb L_0$ and all larger wrappings are \emph{regular} wrappings for $\tT$.

Since we will heavily exploit the Wrapping Lemma in the next sections, we now provide its formal statement. A reader who is already familiar with this concept can skip this section.

\begin{definition}[Wrapping]
A \emph{wrapping map} on the additive group $\mathbb Z^s$ is a homomorphism 
\begin{gather*}
\varphi:\mathbb Z^s\to\mathbb L\coloneqq\bigtimes_{t=1}^s\mathbb Z_{m_t}, \hspace{5pt} m_t<\infty \hspace{5pt} \forall t\,,\\
x=(x_1,...,x_s) \mapsto \varphi(x)=((x_1\Mod{m_1}) ,..., (x_s\Mod{m_s}))\,,
\end{gather*}
where $\L$ is called wrapping.
\end{definition}
The following definition formalises the idea of ``sufficently large'' wrapping, that is, a wrapping such that the neighbourhoods of cells $x$ and $y$, after wrapping, have the same intersection as before wrapping, so that the commutation relations $\gc{\tT_y(A)}{\tT_x(B)}=0$ with $x,y\in\mathbb L$, $A,B\in \m A(\L)$, and $\mathcal{N}_y\cap\mathcal{N}_x\neq\emptyset$ are the same as in the infinite case.
\begin{definition}[Regular wrapping]\label{def:regnei}
Let $\tT:\A(\G)\to\A(\G)$ be an FCA with neighbourhood scheme $\mathcal N\subset\mathbb Z^s$. We say that $\varphi(\mathbb G)=\L$ is a \emph{regular wrapping} for $\tT$  if 
\begin{align*}
&\{\varphi^{-1}(\mathcal N_{\varphi(x)}\cap \mathcal N_{\varphi(y)})\}\cap
\{\mathcal N_x\cup \mathcal N_y\}=\mathcal N_x\cap \mathcal N_y 
\end{align*}
for every pair $x,y\in\mathbb G$ such that $\mathcal N_x\cap\mathcal N_y\neq\emptyset$.
\end{definition}
The intersection with $\mathcal N_x\cup \mathcal N_y$ on the rhs~in the definition is required because the inverse image of $\mathcal N_{\varphi(x)}\cap \mathcal N_{\varphi(y)}$ contains infinitely many intersections, due to the periodicity of the set $\varphi^{-1}(x)$. For a  nearest-neighbour automaton, any  lattice $\L$ comprising at least $5$ cells in each coordinate direction, i.e., $m_t \geq 5$ for all $t$, gives a regular wrapping. Given the specific features of the automaton, smaller lattices may also be considered \cite{trezzini2024renormalisationquantumcellularautomata}. Depending on the parity of $m_t$ an FCA might append a minus sign to odd operators, which is physically irrelevant (see Remark \ref{rmk:Kraus}).

Given a wrapping $\varphi:\mathbb Z^s\to\mathbb L$, we consider the quasi-local algebras on $\mathbb Z^s$ and $\mathbb L$, and, accordingly, we introduce the following isomorphisms, referring to those introduced in Remark~\eqref{rem:abstran}:
\begin{align*}
&\mathcal{U}_\varphi\coloneqq\tH^{-1}_{\mathbb Z^s}\circ\tH_\mathbb L:\mathsf{A}_0^{(w)}\rightarrow\mathsf{A}_0 \\
&\mathcal{W}_\varphi\coloneqq\tK^{-1}_\mathbb L\circ\tK_{\mathbb Z^s}:\mathsf{A}({\mathcal{N}})\rightarrow\mathsf{A}^{(w)}({\mathcal{N}})\,,
\end{align*}
where $\mathcal{U}_\varphi$ maps the single cell algebra $\mathsf{A}^{(w)}_0$ of the  wrapping $\L$ into the single cell algebra $\mathsf{A}_0$ of the infinite lattice, while the map $\mathcal{W}_\varphi$ maps the algebra of the neighbourhood $\mathsf{A}({\mathcal{N}})$ of the infinite lattice into that of the wrapping $\mathsf{A}^{(w)}({\mathcal{N}})$. 
The wrapping regularity requirement in Def.~\eqref{def:regnei} ensures that the application of 
$\varphi$ to $\mathcal{N}_x$ does not change the neighbourhood structure. 

If we now consider an FCA $\tT$ over $\mathbb{Z}^s$, with transition rule $\tT_0$, the following remarkable result holds
\begin{lemma}[Wrapping Lemma]\label{lem:wrap}
There is a bijective correspondence between nearest-neighbour FCA $\tT:\A(\Z^s)\to\A(\Z^s)$ on $\mathbb Z^s$ and nearest-neighbour FCA $\tT^{(w)}:\A(\L)\to\A(\L)$ on  any regular wrapping $\L=\varphi(\Z^s)$, given by 
\begin{align*}
\tK_\mathbb L\circ\tT^{(w)}_x\circ\tH_{\mathbb L}^{-1}=\tK_{\mathbb Z^s}\circ\tT_x\circ\tH_{\mathbb Z^s}^{-1}=\tT_*\,,
\end{align*}
which holds irrespective of $x\in \L$.
\end{lemma}
The above result then inspires the definition of the following equivalence relation.
\begin{definition}
Two FCA on wrappings $\L,\L'$ \emph{have the same transition rule} if 
\begin{align*}
\tK_\mathbb L\circ\tT^{(w)}_0\circ\tH_{\mathbb L}^{-1}=\tK_{\mathbb L'}\circ\tT^{(w')}_0\circ\tH_{\mathbb L'}^{-1}=\tT_*\,.
\end{align*}
\end{definition}
Considering a cellular automaton $\tT$ on $\mathbb Z^s$ we can finally define its wrappings as follows.
\begin{definition}[Wrapped FCA]\label{def:wrapped} The FCA $\tT^{(w)}: \A(\L)\to \A(\L)$ on $\mathbb{L}=\varphi(\mathbb Z^s)$ is a \emph{wrapping} of an FCA $\tT$ over $\mathbb{Z}^s$ if $\tT^{(w)}$ has the same transition rule as $\tT$:
\begin{align} \label{eq:wrapping}
\tT^{(w)}_0=\mathcal{W}_\varphi \circ\tT_0\circ\mathcal{U}_\varphi\,.
\end{align}
\end{definition} 

One of the advantages of working with wrapped FCA is that, for a finite lattice $\mathbb L$,
the *-automorphism $\tT^{(w)}$ corresponds with conjugation by a unitary element $U\in\m A(\mathbb L)$ (see Theorem~\ref{thm:skolem} in Appendix~\ref{app:z2algebras}), i.e.
\begin{align}
    &\tT^{(w)}(X)=UXU^\dag,\\
    &U^\dag U=UU^\dag =I_{\m A(\mathbb L)}.
\end{align}

\section{FCA renormalisation}\label{sec:CG}
In this section, we extend the study of renormalisation of quantum cellular automata in Ref.~\cite{trezzini2024renormalisationquantumcellularautomata} to the case of fermionic cellular automata. 
Then, we will study the explicit form of renormalisable FCA with cells made of a single fermionic mode, namely the simplest one-dimensional FCA---in many respects these represent the fermionic counterpart of quantum cellular automata of qubits. 

We define a procedure that yields a coarse-grained dynamics of a fermionic cellular automaton, ensuring that the resulting large-scale dynamics remains an FCA.
To establish this rigorously, we first formulate the problem for finite lattices and then apply the Wrapping Lemma to extend the results to infinite lattices. This approach provides the straightest way for a rigorous and systematic framework for FCA renormalisation~\cite{trezzini2024renormalisationquantumcellularautomata} (see remark \ref{rmk:wrapp}).
We begin by considering degrees of freedom on finite lattices $\mathbb{L}$ renormalised into coarser lattices $\mathbb{L}'$ and deriving a necessary and sufficient condition for an FCA on the lattice $\mathbb{L}$ to admit a renormalisation on $\mathbb{L}'$. We then prove that this condition depends solely on the local update rule and is independent of the lattice choice, reducing the problem to a single renormalisation equation defined on a minimal lattice $\mathbb{L}_0$. Finally, we extend the construction to infinite lattices by formulating the renormalisation condition purely in terms of the transition rule.

We say that an FCA that admits a solution to the renormalisation equation is \emph{renormalisable}. The
coarse-graining of FCA determines a displacement in the space of FCA as specified by 
their defining parameters, that represents a \emph{renormalisation flow}, whose fixed 
points play a distinguished role as in any renormalisation scenario. We shall return to this point in the Conclusions.

\vspace*{0.4cm}

We consider a homogeneous
tessellation  of $\mathbb Z^s$ such that each tile 
is a $s$-dimensional hypercube of side $N \in \N$, i.e.
\begin{align*}
&\bigcup_{{x}\in \mathbb{Z}^s}\Lambda_{x}= \mathbb{Z}^s,\\
&\Lambda_0\coloneqq [0,N-1]^s,&&\Lambda_{{x}}\coloneqq\Lambda_0+Nx,
\end{align*}
where $[a,b]$ denotes the set
$\{a,a+1,\ldots,b-1,b\}\subset\mathbb Z$. 
\begin{figure}
\centering

 \resizebox{0.4\paperwidth}{!}{
   \begin{tikzpicture}
\foreach \t in {0,1} {
\foreach \s in {0,1} {
\node[tassel={3*\bitsize+\bitdistance}{0.5}] (T) at (-5cm +\bitdistance+\bitsize +3*\t*\bitsize+4*\t*\bitdistance,1.1*\bitdistance+\bitsize+3*\s*\bitsize+4*\s*\bitdistance) {};
\node at (-5cm +\bitdistance+\bitsize +3*\t*\bitsize+4*\t*\bitdistance,1.1*\bitdistance+\bitsize+3*\s*\bitsize+4*\s*\bitdistance) {\large $\Lambda_{\t,\s}$};
}}
\foreach \t in {0,...,3}{
\foreach \s in {0,...,3}{
\node[bittil={1.5*\bitsize}{bluenonren}] at (-5cm+1.5*\t*\bitsize+2*\t*\bitdistance,1.5*\s*\bitsize+2*\s*\bitdistance) {$(\t,\s)$};}
}

\foreach \t in {0,1}{
\foreach \s in {0,1}{
\node[tassel={3*\bitsize+\bitdistance}{0.5}] at ($(T)+(6*\bitdistance+6*\bitsize+3*\t*\bitsize+4*\t*\bitdistance,-3*\bitsize-4*\bitdistance+3*\s*\bitsize+4*\s*\bitdistance)$) {};
\node[bittil={1.5*\bitsize}{blueren}] (Tprim) at ($(T)+(6*\bitdistance+6*\bitsize+3*\t*\bitsize+4*\t*\bitdistance,-3*\bitsize-4*\bitdistance+3*\s*\bitsize+4*\s*\bitdistance)$) {$(\t,\s)$};
}
}
\draw[->,line width=1pt] ($(T)+(1.5cm+1.5*\bitsize+0.5*\bitdistance,-1.5*\bitsize-2*\bitdistance)$)--($(T)+(4.5cm+2*\bitsize+\bitdistance,-1.5*\bitsize-2*\bitdistance)$);
\end{tikzpicture}}
        \caption{ The coarse-grained lattice $\mathbb{L'}$ of $\L$ is given by the identification $\Lambda_{\vc{x}}\rightarrow \vc{x}=(x,y)$.}
\label{fig:Tess}
\end{figure}
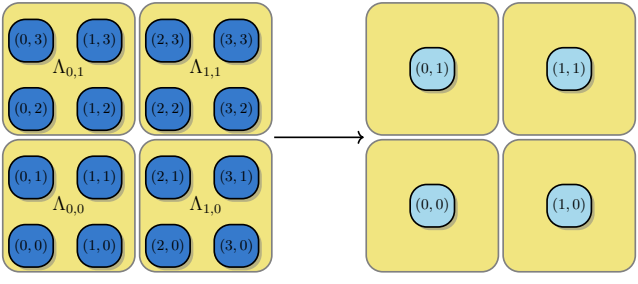
The hypercubes
$\{\Lambda_{{x}}\}$ are related to each other by translations.  
Consider now the wrapped FCA  $\tT_w$ of the FCA $\mathcal T$
on a regular wrapping for $\tT$
\begin{align}\label{eq:2}
\mathbb{L} = \mathbb{Z}_{Nm}^s   
\end{align}
where $m \in \N$ is large enough to guarantee regularity of the wrapping. Clearly, the lattice $\mathbb{L}$ admits a tessellation in terms of
the same $s$-dimensional hypercubes. Then we define a coarse-grained lattice (Fig. \ref{fig:Tess})
\begin{align}
  \label{eq:1}
  \mathbb{L}' :=
  \mathbb{Z}_{m}^s\quad \text{such that}\quad \bigcup_{{x}\in \mathbb{L}'}\Lambda_{x}=
  \mathbb{L}\,.
\end{align}

\begin{definition}
The \emph{projections} of the algebra of a tile $\A(\Lambda_x)$ are elements of the form 
\begin{align}
    \Pi_{{{\Lambda_x}}}=\sum_{{{g}\in S}} \dyad{{{g}}} ,\quad     \braket{g}{g'}=\delta_{g,g'}.
\end{align}
The cardinality $|S|$ of the index set $S$ is the \emph{rank} of the projection $\Pi_x$, denoted as $\operatorname{rank}(\Pi_x)$.    
\end{definition}

We define the coarse-grained algebra $\mathsf A^\Pi(\mathbb{L}) \subseteq \m A(\L)$  over $\mathbb L$  as
\begin{align*}
  \begin{aligned}
&\mathsf
A^{\Pi}(\mathbb{L})\coloneqq
\{ O \in \mathsf{A}(\mathbb{L}) \mbox{ s.t. }
\Pi O \Pi = O\} , \\
&\Pi\coloneqq\bigboxtimes_{x\in \mathbb{L'}}\Pi_{\Lambda_x}\,,  
  \end{aligned}
\end{align*}
where $\Pi_{\Lambda_x}$ are the same for all $x$. We set $\A'(\mathbb L')\cong\A^\Pi(\mathbb L)$,  
and $\dim(\m A'_x)=\operatorname{rank}(\Pi_x)^2$, with the cell $x\in \mathbb L'$ corresponding to the tile $\Lambda_x\in\mathbb L$ for any $x$.
Let us explicitly construct such algebra isomorphisms.
Let 

$\Pi_{{{\Lambda_x}}}=\sum_{{g}\in S} \dyad{g}$ with $\text{rank}(\Pi_x) \leq 2^{d\abs{\Lambda_x}}$.

We define isometry $\tV$ and  coisometry $\tV^\dag$ (Fig.\ref{fig:isom})
\begin{align}
  \begin{aligned}\label{eq:Jiso3}
&\tV:\Aredloc \to \mathsf{A}(\mathbb{L})
               ,\quad \tJ: \mathsf{A}(\mathbb{L}) \to\Aredloc\,, \\
&\tV(O') := \sum_{f,g\in S} \tF_{\psi_ff,\psi_gg}(O')\,,  \\
 & \tJ(O):= \sum_{f,g\in S} \tF^\dag_{\psi_ff,\psi_gg} (O)\,,
  \end{aligned}
\end{align}
where
\begin{equation*}
\tF^\dag_{\psi_ff,\psi_gg}(O)\coloneqq \ketbra{\psi_f}{{f}}O\ketbra{g}{\psi_{g}} 
\end{equation*}
(see Appendix \ref{app:z2algebras} for the details). It is useful to define also
\begin{align*}
       \mathcal{P} (O):= \Pi O \Pi  .
\end{align*}
One easily obtains
\begin{align*}
&\tE(I_{\m A'({\mathbb L'})})=\Pi,\quad \tE^\dag(I_{\m A(\mathbb L)})=I_{\m A'(\mathbb L')},\\
&\tJ\circ\tV=\tI, \quad
\tV\circ\tJ=\tP\,,
\end{align*}
i.e., $\tV$ is an
isomorphism between
$\Aredloc$ and its image $\Aplocc$, and  the
restriction of $\tJ$ to
$\Aplocc$ is the inverse isomorphism.

\begin{figure}
\begin{tikzpicture}[baseline=(O.base)]
\node (O) at (0,\hdist) {};
\draw[fill=blueren]  (-3,0) circle [radius=2cm];
\draw[fill=bluenonren] (5,0) circle [radius=4cm];
\draw[fill=blueren] (5,0) circle [radius=2cm];
\node at (5,0) {\large$\mathsf{A}^\Pi(\mathbb{L})$};
\node at (-3,0) {\large$\mathsf{A}'(\mathbb{L'})$};
\node at (5,3) { \large$\mathsf{A}(\mathbb{L})$};
\draw[->, line width=1pt] (-3,2) to [bend left] node [anchor=south] {\large $\tV$} (5,2);
\draw[->, line width=1pt] (5,-2) to [bend left] node [anchor=north] {\large$\tJ$} (-3,-2);
\end{tikzpicture}
\caption{Representation of the maps $\tJ, \tV$ between the algebra $\m A(\L)$ on the lattice $\L$ and its coarse-grained counterpart $\m A'(\L')\cong \m A^\Pi(\L)$ on the coarse-grained lattice $\L'$.}
\label{fig:isom}
\end{figure}
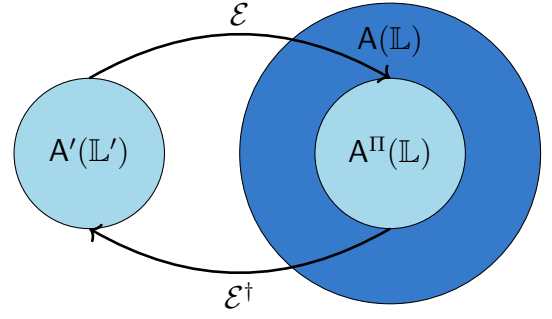

\begin{remark}
    While $\tP$ preserves the parity of the algebra, $\tE^\dag$ may alter it. This can be seen by noticing that the coarse-grained parity operator is \[Q^r_x=\tE^\dag(Q_{\Lambda_x}),\]
    where the action of $\tE^\dag$ selects a subset of eigenvalues of $Q_{\Lambda_x}$.
\end{remark}

\begin{remark}\label{rmk:wrapp}
   Notice that while $\Pi \in \A(\mathbb L) \subset \A^{\textrm{loc}}(\Z^s)$, the projection $\bigboxtimes_{x \in \Z^s} \Pi_{\Lambda_x}$ is not an element of the quasi-local algebra 
   $\A(\Z^s)$. For this reason, the renormalisability condition will be formulated on the wrapped FCA and then related to the corresponding original unwrapped FCA employing the Wrapping Lemma, thereby implementing a well-posed representation-independent definition of coarse-grained FCA on a coarse-grained quasi-local algebra (see Appendix \ref{sec:appendiceThm1}). 
\end{remark}

We are now ready to define renormalisable FCA (Fig.\ref{fig:ren_eq}, \ref{fig:CG}). The following definition is modelled on the analogous one
from~\cite{PhysRevE.73.026203,trezzini2024renormalisationquantumcellularautomata}.
\begin{definition}[Renormalisable FCA] Let $\tT$ and $\tS$ be two FCA on $\mathbb{Z}^s$,
  $N\in \N$,
  and $\tE^\dag$ a coisometry from $\m A(\Lambda_0)$ to $\m A'_0$ with $\tE (I_{\m A'_0})=
  \Pi_{\Lambda_0} \in \mathsf{A}(\Lambda_0)$ a
  projection of the algebra located on $\Lambda_0 =[0,N-1]^s \subset \Z^s$. 
  We say that $\tS$ is a $(N,\Pi_{\Lambda_0})$-renormalisation of $\tT$   if for any regular wrapping  $\mathbb{L} = \mathbb{Z}_{Nm}^s$ for $\tT^N$
\begin{equation}
\tT^N_{w}\circ \tV=\tV\circ\tS_w\,.
\label{eq:cg}
\end{equation}
 We then say that $\tT$ is $(N,\Pi_{\Lambda_0})$-\emph{renormalisable} (to $\tS$).
\end{definition}
 
  \begin{figure}[!]
  \begin{tikzpicture}[remember picture]
  \node (A) at (-2,2.25) {\includegraphics[scale=0.2]{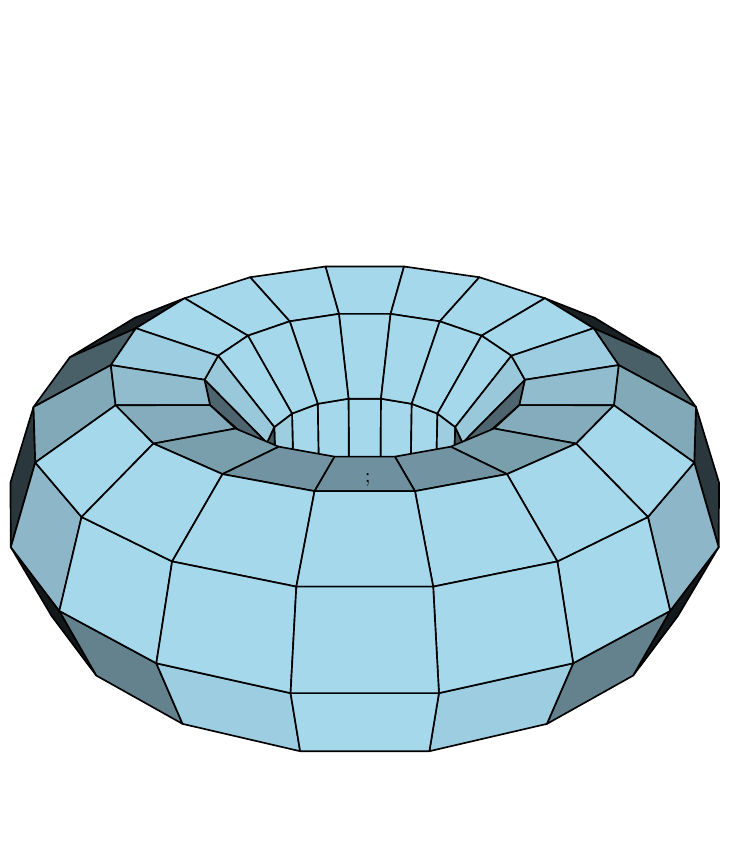}};
  \node (B) at (-2,-2.25) {\includegraphics[scale=0.2]{toruscg}};
  \draw[->,line width=2 pt] ($(A)-(0,1.25)$) to node[anchor=east] {\huge $\tS_w$} ($(B)+(0,1)$) ;
  \node (C) at (3.5, 2.25 ) {\includegraphics[scale=0.2]{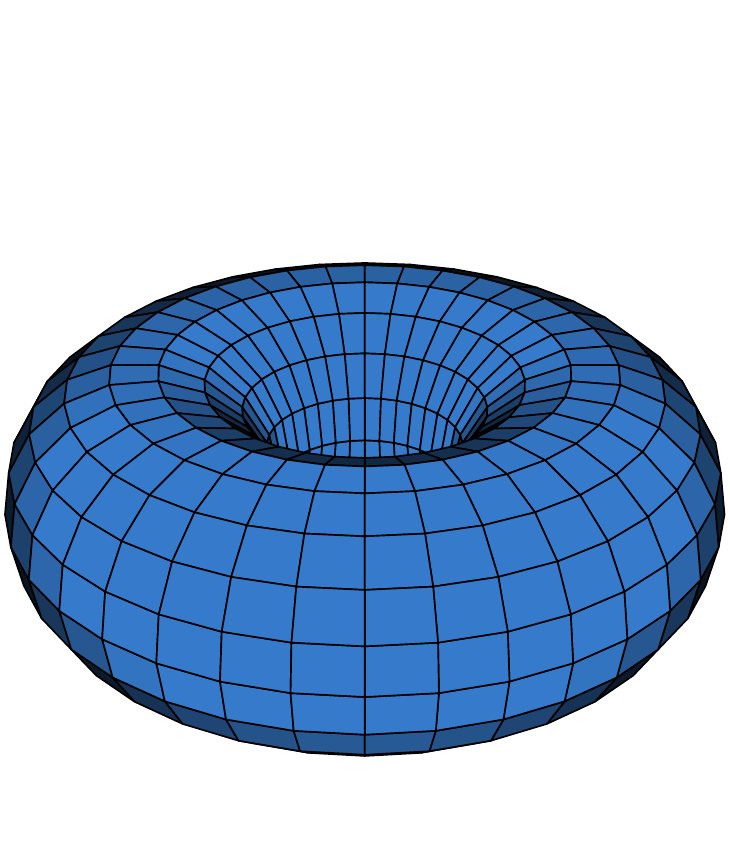}};
  \node (D) at (3.5,-2.25) {\includegraphics[scale=0.2]{torus}};
  \draw[->,line width=2 pt] ($(C)-(0,1.25)$) to node[anchor=west] {\huge $\tT_w^N$} ($(D)+(0,1)$) ;
  \draw[->,line width=2 pt]  ($(A)+(1.75,-0.25)$) to node[anchor=south] {\huge $\tV$}  ($(C)-(1.75,0.25)$) ;
  \draw[->,line width=2 pt]  ($(B)+(1.75,-0.25)$) to node[anchor=north] {\huge $\tV$} ($(D)-(1.75,0.25)$);
  \end{tikzpicture}
  \caption{ The renormalisation equation \eqref{eq:cg} amounts to ask the commutativity of the above diagram. If the FCA $\tS$ is a renormalisation of the FCA $\tT$, then evolving the coarse-grained algebra by one step by $\tS_w$ ($\tS$)  and subsequently embedding this algebra into the original one via $\tV$ is the same than first embedding the coarse-grained algebra in the original one and then evolving for $N$ steps by $\tT_w$ ($\tT$). In this way, a single step of $\tS$ reproduces the action of $\tT^N$ when restricted to the chosen degrees of freedom. The renormalisation equation is expressed on the wrapped FCA for technical convenience, but it applies to the unwrapped automata (see Appendix \ref{sec:appendiceThm1}).}
  \label{fig:ren_eq}
  \end{figure}

  \begin{figure}
\centering
(a)\resizebox{0.35\paperwidth}{!}{
\begin{tikzpicture}
  \node (A) at (-3,0)  {\resizebox{0.5\hsize}{!}{\begin{tikzpicture}
\foreach \t in  {0,...,5}{
\draw[black]  (\bitdistance*\t+1.5*\bitsize*\t,6*\hdist) -- (\bitdistance*\t+1.5*\bitsize*\t, \bitsize*0.5+2pt) ;
\node at (\bitdistance*\t+1.5*\bitsize*\t,-0.1) {\t}; 
}
\foreach \t in {0,...,5}{
\node[U] at (1.5*\bitsize*\t+\bitdistance*\t,0.0001*\hdist) {$\mathcal U$};
\node[U] at (1.5*\bitsize*\t+\bitdistance*\t,3*\hdist) {$\mathcal U$};
}
\foreach \t in {0,...,2}{
\node[tassel={1.25*\bitsize+\bitdistance}{0.5}] at (0.75*\bitsize+0.5*\bitdistance+\bitsize*2.5*\t+\bitdistance*2.25*\t, 1.5*\hdist) { };
}
\end{tikzpicture}}};
  \node (B) at (2,0) {\resizebox{0.4\hsize}{!}{\begin{tikzpicture}
\foreach \t in  {0,...,2}{
\draw[black]  (\bitdistance*\t+1.5*\bitsize*\t,3*\hdist) -- (\bitdistance*\t+1.5*\bitsize*\t, \bitsize*0.5+2pt) ;
\node at (\bitdistance*\t+1.5*\bitsize*\t,-0.1) {\t}; 
}
\foreach \t in {0,...,2}{
\node[U] at (1.5*\bitsize*\t+\bitdistance*\t,0.0001*\hdist) {$\mathcal U'$};
}

\end{tikzpicture}}};
  \draw[->,line width=2pt] ($(A)+(2.2,0)$)--($(B)-(1.8,0)$);
\end{tikzpicture}}
(b)\resizebox{0.35\paperwidth}{!}{\begin{tikzpicture}
  \node (A) at (-3,0)  {\resizebox{0.5\hsize}{!}{\begin{tikzpicture}[baseline=(A0.base)]
\foreach \a in {0,...,7}{
\node (A\a ) at (\bitsize*\a+\bitdistance*\a ,0.001*\hdist) {};
\draw[black, line width=1pt] ($(A\a)-(\bitsize*0.5+\bitdistance*0.5,\hdist)$)--($(A\a)-(\bitsize*0.5+\bitdistance*0.5,0.25*\hdist)$);
\node at (\bitdistance*\a+\bitsize*\a-\bitsize*0.5-\bitdistance*0.5,-1.4) {\huge \a}; 
\draw[black, line width=1pt] ($(A\a)-(\bitsize*0.5+\bitdistance*0.5,0.25*\hdist)$)--($(A\a)+(\bitsize*0.5+\bitdistance*0.5,0)$);
\draw[black, line width=1pt] ($(A\a)+(\bitsize*0.5+\bitdistance*0.5,0)$)--($(A\a)+(\bitsize*0.5+\bitdistance*0.5,0.6)$);
\node (B\a ) at ($(A\a)+(0,2*\hdist)$) {};
\draw[black, line width=1pt] ($(B\a)-(\bitsize*0.5+\bitdistance*0.5,\hdist)$)--($(B\a)-(\bitsize*0.5+\bitdistance*0.5,0.25*\hdist)$);
\draw[black, line width=1pt] ($(B\a)-(\bitsize*0.5+\bitdistance*0.5,0.25*\hdist)$)--($(B\a)+(\bitsize*0.5+\bitdistance*0.5,0)$);
\draw[black, line width=1pt] ($(B\a)+(\bitsize*0.5+\bitdistance*0.5,0)$)--($(B\a)+(\bitsize*0.5+\bitdistance*0.5,0.005*\hdist)$);
}
\foreach \t in {0,...,3}{
\node[tassel={1.2*\bitsize+0.9*\bitdistance}{0.5}] at (\bitsize*0.15+\bitdistance*0.15+2*\bitsize*\t+2*\bitdistance*\t,0.45*\hdist) {};
}
\node at ($(A0)-(2,\nhdist)$) {\huge ...};
\node at ($(A7)+(1.8,\hdist)$) {\huge ...};
\end{tikzpicture}}};
  \node (B) at (2.1,-0.1) {\resizebox{0.4\hsize}{!}{\begin{tikzpicture}[baseline=(A0.base)]
\foreach \a in {0,...,3}{
\node (A\a ) at (\bitsize*\a+\bitdistance*\a ,0.001*\hdist) {};
\draw[black, line width=1pt] ($(A\a)-(\bitsize*0.5+\bitdistance*0.5,\hdist)$)--($(A\a)-(\bitsize*0.5+\bitdistance*0.5,0.25*\hdist)$);
\node at (\bitdistance*\a+\bitsize*\a-\bitsize*0.5-\bitdistance*0.5,-1) {\a}; 
\draw[black, line width=1pt] ($(A\a)-(\bitsize*0.5+\bitdistance*0.5,0.25*\hdist)$)--($(A\a)+(\bitsize*0.5+\bitdistance*0.5,0)$);
\draw[black, line width=1pt] ($(A\a)+(\bitsize*0.5+\bitdistance*0.5,0)$)--($(A\a)+(\bitsize*0.5+\bitdistance*0.5,0.6)$);
}
\node at ($(A0)-(2,0.1)$) {\huge ...};
\node at ($(A3)+(1.8,-0.1)$) {\huge ...};
\end{tikzpicture}}};
  \draw[->,line width=2pt] ($(A)+(2.3,0.2)$)--($(B)-(1.9,-0.3)$);
\end{tikzpicture}}
        \caption{Trivial examples of $(2,\Pi_{\Lambda_0})$-renormalisation over a one-dimensional lattice $\Z$. Here we have $\Lambda_x=\{2x,2x+1\}$. (a) Renormalisation of cell-wise transformations. Two steps of gates $\mathcal U$ over two cells are mapped in a single step of new cell-wise gates $\mathcal U'$. (b) Renormalisation of a right shift. After two steps, all the content of $\Lambda_x$ is moved to $\Lambda_{x+1}$. The renormalised evolution is a right shift $\A'_x\mapsto \A'_{x+1}$ over the coarse-grained lattice $\L'$.}
\label{fig:CG}
\end{figure}
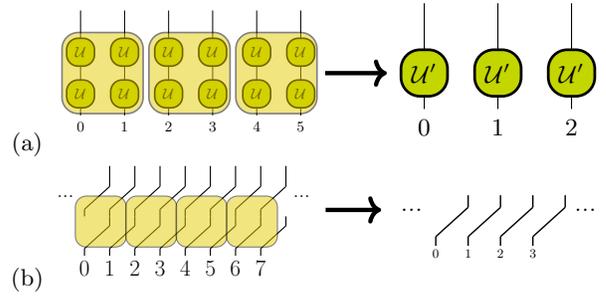

 One could regard the specification $(N,\Pi_{\Lambda_0})$ as identifying a resolution, or energy scale, of the renormalisation scheme: how degrees of freedom are coarse-grained while coarse-graining the space-time lattice. We choose a uniform coarse-graining of time steps and cells so that renormalised FCA of nearest-neighbour FCA remain nearest-neighbour. Clearly, one could conceive an uneven coarse-graining of time and space.  

 Note that the choice of $\Pi$ fixes the preserved degrees of freedom in the coarse-graining, while the choice of the basis in which we represent the renormalised FCA is arbitrary as $\bigboxtimes_{x \in \L'} \tU_x \circ\tE_x$ realises the same projection as $\bigboxtimes_{x \in \L'} \tE_x$ for any unitary map $\tU_x \in\m A'_x$ (Appendix \ref{sec:appendiceThm1}). 
 
 The main result of this section is the following.
 
\begin{theorem}
  \label{prp:thecoarsegrainingisfinite}
  An FCA $\tT$ on $\Z^s$ is $(N,\Pi_{\Lambda_0})$-\emph{renormalisable} to an FCA $\tS$ if and only if there exists a regular wrapping 
  $\mathbb{L} = \mathbb{Z}^s_{Nm}$ for $\tT^N$ such that 
  \[\mathcal P \circ \mathcal \tT_w^N=\tT_w^N \circ\mathcal P\,,\]
  or equivalently
  \begin{align}\label{eq:cons_proj}
      \tT_w^N(\Pi)=\Pi\quad \iff \quad[\Tnmat,\Pi]=0\,,
  \end{align}
  where $\Tnmat$ is the unitary operator defining the wrapped FCA $\tT^N_w$ and $\Pi=\bigboxtimes_{x\in \L'} \Pi_{\Lambda_x}$ with $\L'$ the coarse-grained lattice. In this case $\tS$ is the unique extension of $\tS_w=\tE^\dag \circ \tT_w^N \circ \tE$ to the quasi-local algebra $\A'(\Z^s)$.

\end{theorem}

This proves that the renormalisability equation \eqref{eq:cg}, reduces to a finite dimensional equation in terms of the unitary operator defining the FCA on a regular wrapping. In particular, the renormalisability condition can be checked by considering the smallest regular wrapping. 

The following corollary is an immediate consequence of Eq.\eqref{eq:cons_proj} and Hahn-Banach theorem \cite{bratteli1987operator}.

\begin{Corollary}
     An FCA $\tT$ on $\Z^s$ is $(N,\Pi_{\Lambda_0})$-\emph{renormalisable} if and only if there exists a regular wrapping 
  $\mathbb{L} = \mathbb{Z}^s_{Nm}$ for $\tT^N$ and an invariant product state $\omega_{\Pi}$ 
  \[ \omega_{\Pi} \circ \tT^N= \omega_{\Pi}\,,\]
  with $\omega_{\Pi}(O)\coloneqq\Tr_{\A(\L)}(\rho_{\Pi} O)$, $\rho_{\Pi}=\frac{1}{\Tr_{\A(\L)} \Pi}\bigboxtimes_{x \in \L'} \Pi_{\Lambda_x}$, and $\Tr$ the tracial state. The state $\omega_{\Pi}$ has a unique extension to $\A (\Z^s)$ and its coarse-grained state is the tracial state on $\A' (\Z^s)$.
\end{Corollary}

 The detailed proof of Theorem \ref{prp:thecoarsegrainingisfinite} and the definition of state can be found in Appendix \ref{sec:appendiceThm1} and Appendix \ref{app:z2algebras}, respectively.  

In the following sections, we specify Theorem \ref{prp:thecoarsegrainingisfinite} firstly for the case of fermionic circuits on a linear graph $\Z$ and then for FCA on a linear graph $\Z$ with only one fermionic mode per cell. In both cases, the renormalisation condition can be further simplified, and we explicitly compute the renormalised FCA for those FCA that are renormalisable.

\section{Renormalisation of
  one-dimensional fermionic circuits}\label{sec:FDQC}
In this section we specify our analysis to FCA $\tT$ on $\Z$ that can be realised as finite-depth fermionic circuits with
a Margolus partitioning scheme:
\begin{equation}
\tT =\begin{tikzpicture}[baseline=(O.base)]
\node (O) at (0,\hdist) {};
\foreach \t in  {0,...,5}{
\draw[black]  (\bitdistance*\t+\bitsize*\t,6*\hdist) -- (\bitdistance*\t+\bitsize*\t, \bitsize*0.5+10pt) ;
}
\foreach \t in {0,...,2}{
\node[M1] at (\bitsize+\transfoffset+\transfdistance*\t,\hdist) {$\mathcal{M}_1$};
}
\foreach \t in {0,...,1}{
\node[M2] at (\bitsize+\transfoffset+\transfsize*0.5+\transfdistance*\t,4*\hdist) {$\mathcal{M}_2$};
}

\draw[fill=myblue2,rounded corners=10pt,draw=white,opacity=1] (\bitsize+\transfoffset+\transfsize*0.5+\transfdistance*2,3*\hdist)--(\bitsize+\transfoffset+\transfdistance*2,3*\hdist)--(\bitsize+\transfoffset+\transfdistance*2,5*\hdist)--(\bitsize+\transfoffset+\transfsize*0.5+\transfdistance*2,5*\hdist);
\draw[black,rounded corners=10pt, line width=1pt,opacity=1] (\bitsize+\transfoffset+\transfsize*0.5+\transfdistance*2,3*\hdist)--(\bitsize+\transfoffset+\transfdistance*2,3*\hdist)--(\bitsize+\transfoffset+\transfdistance*2,5*\hdist)--(\bitsize+\transfoffset+\transfsize*0.5+\transfdistance*2,5*\hdist);

\draw[fill=myblue2,rounded corners=10pt,draw=white,opacity=1]  (-\bitsize-\transfoffset,3*\hdist)--(\halfleft,3*\hdist)--(\halfleft,5*\hdist)--(-\bitsize-\transfoffset,5*\hdist);
\draw[black,rounded corners=10pt, line width=1pt,opacity=1] (-\bitsize-\transfoffset,3*\hdist)--(\halfleft,3*\hdist)--(\halfleft,5*\hdist)--(-\bitsize-\transfoffset,5*\hdist);

\end{tikzpicture}\,,
\label{eq:marg}
\end{equation}
where $\mathcal M_i$ are fermionic gates acting on two cells.
We remark that, up to grouping neighbouring cells, any FDFC on $\Z$ can be cast in this Margolus form. However, in general, renormalising the grouped FCA is different from renormalising the original one, as the allowed projections would be different in the two cases.

We consider coarse-graining of two cells and two time steps into one cell and one time step.  Thanks to Theorem \ref{prp:thecoarsegrainingisfinite}, without loss of generality, the analysis can be carried out assuming $\tT$ being defined on a
\emph{finite} lattice $\mathbb L$, whose coarse-grained lattice will be denoted by $\mathbb L'$. The following Lemma specifies Theorem \ref{prp:thecoarsegrainingisfinite} to one-dimensional FDFCs.

\begin{lemma}
  Let $\tT$ be an FDFC on $\Z$.  $\tT$ is  $(2,\Pi_{\Lambda_0})$-renormalisable if and only if
\begin{equation}
\resizebox{0.8\hsize}{!}{
\begin{tikzpicture}[baseline=(O.base)]
\node (O) at (0,2*\hdist) {};
\foreach \t in  {0,...,5}{
\draw[black]  (\bitdistance*\t+\bitsize*\t,13*\hdist) -- (\bitdistance*\t+\bitsize*\t, \bitsize*0.5+10pt) ;
}
\foreach \t in {0,...,2}{
\node[LR] at (\bitsize+\transfoffset+\transfdistance*\t,0) {\LR};
}
\foreach \t in {0,...,1}{
\node[M2] at (\bitsize+\transfoffset+\transfsize*0.5+\transfdistance*\t,\hdist) {$\mathcal{M}_2$};
}
\foreach \t in {0,...,1}{
\node[M1] at (\bitsize+\transfoffset+\transfdistance*\t+\transfsize*0.5,4*\hdist) {$\mathcal{M}_1$};
}
\foreach \t in {0,...,2}{
\node[M2] at (\bitsize+\transfoffset+\transfdistance*\t,7*\hdist) {$\mathcal{M}_2$};
}
\foreach \t in {0,...,2}{
\node[M1] at (\bitsize+\transfoffset+\transfdistance*\t,10*\hdist) {$\mathcal{M}_1$};
}

\draw[fill=myblue2,rounded corners=10pt,draw=white,opacity=1] (\bitsize+\transfoffset+\transfsize*0.5+\transfdistance*2,\bitdistance+\halfsize)--(\bitsize+\transfoffset+\transfdistance*2,\bitdistance+\halfsize)--(\bitsize+\transfoffset+\transfdistance*2,2*\hdist)--(\bitsize+\transfoffset+\transfsize*0.5+\transfdistance*2,2*\hdist);
\draw[black,rounded corners=10pt, line width=1pt,opacity=1] (\bitsize+\transfoffset+\transfsize*0.5+\transfdistance*2,\bitdistance+\halfsize)--(\bitsize+\transfoffset+\transfdistance*2,\bitdistance+\halfsize)--(\bitsize+\transfoffset+\transfdistance*2,2*\hdist)--(\bitsize+\transfoffset+\transfsize*0.5+\transfdistance*2,2*\hdist);
\draw[fill=myblue1,rounded corners=10pt,draw=white,opacity=1] (\bitsize+\transfoffset+\transfsize*0.5+\transfdistance*2,3*\hdist)--(\bitsize+\transfoffset+\transfdistance*2,3*\hdist)--(\bitsize+\transfoffset+\transfdistance*2,5*\hdist)--(\bitsize+\transfoffset+\transfsize*0.5+\transfdistance*2,5*\hdist);
\draw[black,rounded corners=10pt, line width=1pt,opacity=1] (\bitsize+\transfoffset+\transfsize*0.5+\transfdistance*2,3*\hdist)--(\bitsize+\transfoffset+\transfdistance*2,3*\hdist)--(\bitsize+\transfoffset+\transfdistance*2,5*\hdist)--(\bitsize+\transfoffset+\transfsize*0.5+\transfdistance*2,5*\hdist);

\draw[fill=myblue2,rounded corners=10pt,draw=white,opacity=1]  (-\bitsize-\transfoffset,\bitdistance+\halfsize)--(\halfleft,\bitdistance+\halfsize)--(\halfleft,2*\hdist)--(-\bitsize-\transfoffset,2*\hdist);
\draw[black,rounded corners=10pt, line width=1pt,opacity=1] (-\bitsize-\transfoffset,\bitdistance+\halfsize)--(\halfleft,\bitdistance+\halfsize)--(\halfleft,2*\hdist)--(-\bitsize-\transfoffset,2*\hdist);
\draw[fill=myblue1,rounded corners=10pt,draw=white,opacity=1] (-\bitsize-\transfoffset,3*\hdist)--(\halfleft,3*\hdist)--(\halfleft,5*\hdist)--(-\bitsize-\transfoffset,5*\hdist);
\draw[black,rounded corners=10pt, line width=1pt,opacity=1] (-\bitsize-\transfoffset,3*\hdist)--(\halfleft,3*\hdist)--(\halfleft,5*\hdist)--(-\bitsize-\transfoffset,5*\hdist);
\end{tikzpicture}=\begin{tikzpicture}[baseline=(O.base)]
\node (O) at (0,2*\hdist) {};
\foreach \t in  {0,...,5}{
\draw[black]  (\bitdistance*\t+\bitsize*\t,11*\hdist) -- (\bitdistance*\t+\bitsize*\t, \bitsize*0.5+2pt) ;
}
\foreach \t in {0,...,2}{
\node[LR] at (\bitsize+\transfoffset+\transfdistance*\t,0) {\LR};
}

\end{tikzpicture}}
\label{eq:RenIndex}
\end{equation}
where we define the projection
\begin{align}
  \label{eq:4}
  P_{\Lambda_x} = \begin{tikzpicture}
\foreach \t in  {1,...,2}{
\draw[black]  (\bitdistance*\t+\bitsize*\t,1) -- (\bitdistance*\t+\bitsize*\t, \bitsize*0.5+2pt) ;
}
\node[LR] at (\bitsize+\transfoffset+\transfdistance*0.5,0) {\LR};
\end{tikzpicture}\coloneqq\begin{tikzpicture}[baseline=(O.base)]
\node (O) at (0,0.25*\hdist) {};
\foreach \t in  {1,...,2}{
\draw[black]  (\bitdistance*\t+\bitsize*\t,3*\hdist) -- (\bitdistance*\t+\bitsize*\t, \bitsize*0.5+2pt) ;
}
\node[M1] at (\bitsize+\transfoffset+\transfdistance*0.5,\hdist) {$\mathcal{M}_1$};
\node[pi] at  (\bitsize+\transfoffset+\transfdistance*0.5,0) {$\Pi$};
\end{tikzpicture}.
\end{align}
\end{lemma}
\begin{proof}
We can write
\begin{equation*}
\tT^{2}\left(\bigboxtimes_{x\in\mathbb{L}'}\Pi_{\Lambda_x}\right)=\input{T2.tikz}
\end{equation*}
Since each step is translationally invariant we can shift the second evolution step in order to have
\begin{equation*}
\tT^{2}\left(\bigboxtimes_{x\in\mathbb{L}'}\Pi_{\Lambda_x}\right)=\begin{tikzpicture}[baseline=(O.base)]
\node (O) at (0,3*\hdist) {};
\foreach \t in  {0,...,5}{
\draw[black]  (\bitdistance*\t+\bitsize*\t,10*\hdist) -- (\bitdistance*\t+\bitsize*\t, \bitsize*0.5+10pt) ;
\node[bit]  at (\bitdistance*\t+\bitsize*\t,0)   {};
}
\foreach \t in {0,...,2}{
\node[M1] at (\bitsize+\transfoffset+\transfdistance*\t,\hdist) {$\mathcal{M}_1$};
\node[pi] at (\bitsize+\transfoffset+\transfdistance*\t,0) {$\Pi$};
}
\foreach \t in {0,...,1}{
\node[M2] at (\bitsize+\transfoffset+\transfsize*0.5+\transfdistance*\t,4*\hdist) {$\mathcal{M}_2$};
}
\foreach \t in {0,...,1}{
\node[M1] at (\bitsize+\transfoffset+\transfdistance*\t+\transfsize*0.5,7*\hdist) {$\mathcal{M}_1$};
}
\foreach \t in {0,...,2}{
\node[M2] at (\bitsize+\transfoffset+\transfdistance*\t,10*\hdist) {$\mathcal{M}_2$};
}

\draw[fill=myblue2,rounded corners=10pt,draw=white,opacity=1] (\bitsize+\transfoffset+\transfsize*0.5+\transfdistance*2,3*\hdist)--(\bitsize+\transfoffset+\transfdistance*2,3*\hdist)--(\bitsize+\transfoffset+\transfdistance*2,5*\hdist)--(\bitsize+\transfoffset+\transfsize*0.5+\transfdistance*2,5*\hdist);
\draw[black,rounded corners=10pt, line width=1pt,opacity=1] (\bitsize+\transfoffset+\transfsize*0.5+\transfdistance*2,3*\hdist)--(\bitsize+\transfoffset+\transfdistance*2,3*\hdist)--(\bitsize+\transfoffset+\transfdistance*2,5*\hdist)--(\bitsize+\transfoffset+\transfsize*0.5+\transfdistance*2,5*\hdist);
\draw[fill=myblue1,rounded corners=10pt,draw=white,opacity=1] (\bitsize+\transfoffset+\transfsize*0.5+\transfdistance*2,6*\hdist)--(\bitsize+\transfoffset+\transfdistance*2,6*\hdist)--(\bitsize+\transfoffset+\transfdistance*2,8*\hdist)--(\bitsize+\transfoffset+\transfsize*0.5+\transfdistance*2,8*\hdist);
\draw[black,rounded corners=10pt, line width=1pt,opacity=1] (\bitsize+\transfoffset+\transfsize*0.5+\transfdistance*2,6*\hdist)--(\bitsize+\transfoffset+\transfdistance*2,6*\hdist)--(\bitsize+\transfoffset+\transfdistance*2,8*\hdist)--(\bitsize+\transfoffset+\transfsize*0.5+\transfdistance*2,8*\hdist);

\draw[fill=myblue2,rounded corners=10pt,draw=white,opacity=1] (-\bitsize-\transfoffset,3*\hdist)--(\halfleft,3*\hdist)--(\halfleft,5*\hdist)--(-\bitsize-\transfoffset,5*\hdist);
\draw[black,rounded corners=10pt, line width=1pt,opacity=1] (-\bitsize-\transfoffset,3*\hdist)--(\halfleft,3*\hdist)--(\halfleft,5*\hdist)--(-\bitsize-\transfoffset,5*\hdist);
\draw[fill=myblue1,rounded corners=10pt,draw=white,opacity=1]  (-\bitsize-\transfoffset,6*\hdist)--(\halfleft,6*\hdist)--(\halfleft,8*\hdist)--(-\bitsize-\transfoffset,8*\hdist);
\draw[black,rounded corners=10pt, line width=1pt,opacity=1] (-\bitsize-\transfoffset,6*\hdist)--(\halfleft,6*\hdist)--(\halfleft,8*\hdist)--(-\bitsize-\transfoffset,8*\hdist);
\end{tikzpicture}.
\end{equation*}
By introducing the projection $P_{\Lambda_x}$, we have
\begin{equation*}
\resizebox{0.8\hsize}{!}{\begin{tikzpicture}[baseline=(O.base)]
\node (O) at (0,2*\hdist) {};
\foreach \t in  {0,...,5}{
\draw[black]  (\bitdistance*\t+\bitsize*\t,10*\hdist) -- (\bitdistance*\t+\bitsize*\t, \bitsize*0.5+10pt) ;
\node[bit]  at (\bitdistance*\t+\bitsize*\t,0)   {};
}
\foreach \t in {0,...,2}{
\node[LR] at (\bitsize+\transfoffset+\transfdistance*\t,0) {\LR};
}
\foreach \t in {0,...,1}{
\node[M2] at (\bitsize+\transfoffset+\transfsize*0.5+\transfdistance*\t,\hdist) {$\mathcal{M}_2$};
}
\foreach \t in {0,...,1}{
\node[M1] at (\bitsize+\transfoffset+\transfdistance*\t+\transfsize*0.5,4*\hdist) {$\mathcal{M}_1$};
}
\foreach \t in {0,...,2}{
\node[M2] at (\bitsize+\transfoffset+\transfdistance*\t,7*\hdist) {$\mathcal{M}_2$};
}

\draw[fill=myblue2,rounded corners=10pt,draw=white,opacity=1] (\bitsize+\transfoffset+\transfsize*0.5+\transfdistance*2,\bitdistance+\halfsize)--(\bitsize+\transfoffset+\transfdistance*2,\bitdistance+\halfsize)--(\bitsize+\transfoffset+\transfdistance*2,2*\hdist)--(\bitsize+\transfoffset+\transfsize*0.5+\transfdistance*2,2*\hdist);
\draw[black,rounded corners=10pt, line width=1pt,opacity=1] (\bitsize+\transfoffset+\transfsize*0.5+\transfdistance*2,\bitdistance+\halfsize)--(\bitsize+\transfoffset+\transfdistance*2,\bitdistance+\halfsize)--(\bitsize+\transfoffset+\transfdistance*2,2*\hdist)--(\bitsize+\transfoffset+\transfsize*0.5+\transfdistance*2,2*\hdist);
\draw[fill=myblue1,rounded corners=10pt,draw=white,opacity=1] (\bitsize+\transfoffset+\transfsize*0.5+\transfdistance*2,3*\hdist)--(\bitsize+\transfoffset+\transfdistance*2,3*\hdist)--(\bitsize+\transfoffset+\transfdistance*2,5*\hdist)--(\bitsize+\transfoffset+\transfsize*0.5+\transfdistance*2,5*\hdist);
\draw[black,rounded corners=10pt, line width=1pt,opacity=1] (\bitsize+\transfoffset+\transfsize*0.5+\transfdistance*2,3*\hdist)--(\bitsize+\transfoffset+\transfdistance*2,3*\hdist)--(\bitsize+\transfoffset+\transfdistance*2,5*\hdist)--(\bitsize+\transfoffset+\transfsize*0.5+\transfdistance*2,5*\hdist);

\draw[fill=myblue2,rounded corners=10pt,draw=white,opacity=1] (-\bitsize-\transfoffset,\bitdistance+\halfsize)--(\halfleft,\bitdistance+\halfsize)--(\halfleft,2*\hdist)--(-\bitsize-\transfoffset,2*\hdist);
\draw[black,rounded corners=10pt, line width=1pt,opacity=1] (-\bitsize-\transfoffset,\bitdistance+\halfsize)--(\halfleft,\bitdistance+\halfsize)--(\halfleft,2*\hdist)--(-\bitsize-\transfoffset,2*\hdist);
\draw[fill=myblue1,rounded corners=10pt,draw=white,opacity=1](-\bitsize-\transfoffset,3*\hdist)--(\halfleft,3*\hdist)--(\halfleft,5*\hdist)--(-\bitsize-\transfoffset,5*\hdist);
\draw[black,rounded corners=10pt, line width=1pt,opacity=1] (-\bitsize-\transfoffset,3*\hdist)--(\halfleft,3*\hdist)--(\halfleft,5*\hdist)--(-\bitsize-\transfoffset,5*\hdist);
\end{tikzpicture}=\begin{tikzpicture}[baseline=(O.base)]
\node (O) at (0,2*\hdist) {};
\foreach \t in  {0,...,5}{
\draw[black]  (\bitdistance*\t+\bitsize*\t,10*\hdist) -- (\bitdistance*\t+\bitsize*\t, \bitsize*0.5+2pt) ;
}
\foreach \t in {0,...,2}{
\node[pi] at (\bitsize+\transfoffset+\transfdistance*\t,0) {$\Pi$};
}

\end{tikzpicture}
}\label{eq:evmod}
\end{equation*}
and by applying $\begin{tikzpicture}
\draw[black]  (\bitdistance+\bitsize,1)--(\bitdistance+\bitsize,-1);
\draw[black]  (0,1)--(0,-1);
\node[M1] at (\bitsize+\transfoffset,0) {$\mathcal{M}_1$};
\end{tikzpicture}$ on both sides we obtain the thesis. 
\end{proof}

It is convenient to define the gate
\begin{equation}
  \label{eq:definitionofG}
\mathcal G\coloneqq \begin{tikzpicture}
\draw[black]  (\bitdistance+\bitsize,1.5)--(\bitdistance+\bitsize,-1.5);
\draw[black]  (0,1.5)--(0,-1.5);
\node[M2] at (\bitsize+\transfoffset,-0.75) {$\mathcal{M}_2$};
\node[M1] at (\bitsize+\transfoffset,0.75) {$\mathcal{M}_1$};
\end{tikzpicture}\eqqcolon\begin{tikzpicture}
\draw[black]  (\bitdistance+\bitsize,1)--(\bitdistance+\bitsize,-1);
\draw[black]  (0,1)--(0,-1);
\node[G] at (\bitsize+\transfoffset,0) {$\mathcal{G}$};
\end{tikzpicture},
\end{equation}
so that we have
\begin{equation}
\resizebox{0.8\hsize}{!}{
\begin{tikzpicture}[baseline=(O.base)]
\node (O) at (0,2*\hdist) {};
\foreach \t in  {0,...,5}{
\draw[black]  (\bitdistance*\t+\bitsize*\t,13*\hdist) -- (\bitdistance*\t+\bitsize*\t, \bitsize*0.5+10pt) ;
}
\foreach \t in {0,...,2}{
\node[LR] at (\bitsize+\transfoffset+\transfdistance*\t,0) {\LR};
}
\foreach \t in {0,...,1}{
\node[M2] at (\bitsize+\transfoffset+\transfsize*0.5+\transfdistance*\t,\hdist) {$\mathcal{M}_2$};
}
\foreach \t in {0,...,1}{
\node[M1] at (\bitsize+\transfoffset+\transfdistance*\t+\transfsize*0.5,4*\hdist) {$\mathcal{M}_1$};
}
\foreach \t in {0,...,2}{
\node[M2] at (\bitsize+\transfoffset+\transfdistance*\t,7*\hdist) {$\mathcal{M}_2$};
}
\foreach \t in {0,...,2}{
\node[M1] at (\bitsize+\transfoffset+\transfdistance*\t,10*\hdist) {$\mathcal{M}_1$};
}

\draw[fill=myblue2,rounded corners=10pt,draw=white,opacity=1] (\bitsize+\transfoffset+\transfsize*0.5+\transfdistance*2,\bitdistance+\halfsize)--(\bitsize+\transfoffset+\transfdistance*2,\bitdistance+\halfsize)--(\bitsize+\transfoffset+\transfdistance*2,2*\hdist)--(\bitsize+\transfoffset+\transfsize*0.5+\transfdistance*2,2*\hdist);
\draw[black,rounded corners=10pt, line width=1pt,opacity=1] (\bitsize+\transfoffset+\transfsize*0.5+\transfdistance*2,\bitdistance+\halfsize)--(\bitsize+\transfoffset+\transfdistance*2,\bitdistance+\halfsize)--(\bitsize+\transfoffset+\transfdistance*2,2*\hdist)--(\bitsize+\transfoffset+\transfsize*0.5+\transfdistance*2,2*\hdist);
\draw[fill=myblue1,rounded corners=10pt,draw=white,opacity=1] (\bitsize+\transfoffset+\transfsize*0.5+\transfdistance*2,3*\hdist)--(\bitsize+\transfoffset+\transfdistance*2,3*\hdist)--(\bitsize+\transfoffset+\transfdistance*2,5*\hdist)--(\bitsize+\transfoffset+\transfsize*0.5+\transfdistance*2,5*\hdist);
\draw[black,rounded corners=10pt, line width=1pt,opacity=1] (\bitsize+\transfoffset+\transfsize*0.5+\transfdistance*2,3*\hdist)--(\bitsize+\transfoffset+\transfdistance*2,3*\hdist)--(\bitsize+\transfoffset+\transfdistance*2,5*\hdist)--(\bitsize+\transfoffset+\transfsize*0.5+\transfdistance*2,5*\hdist);

\draw[fill=myblue2,rounded corners=10pt,draw=white,opacity=1]  (-\bitsize-\transfoffset,\bitdistance+\halfsize)--(\halfleft,\bitdistance+\halfsize)--(\halfleft,2*\hdist)--(-\bitsize-\transfoffset,2*\hdist);
\draw[black,rounded corners=10pt, line width=1pt,opacity=1] (-\bitsize-\transfoffset,\bitdistance+\halfsize)--(\halfleft,\bitdistance+\halfsize)--(\halfleft,2*\hdist)--(-\bitsize-\transfoffset,2*\hdist);
\draw[fill=myblue1,rounded corners=10pt,draw=white,opacity=1] (-\bitsize-\transfoffset,3*\hdist)--(\halfleft,3*\hdist)--(\halfleft,5*\hdist)--(-\bitsize-\transfoffset,5*\hdist);
\draw[black,rounded corners=10pt, line width=1pt,opacity=1] (-\bitsize-\transfoffset,3*\hdist)--(\halfleft,3*\hdist)--(\halfleft,5*\hdist)--(-\bitsize-\transfoffset,5*\hdist);
\end{tikzpicture}=\begin{tikzpicture}[baseline=(O.base)]
\node (O) at (0,2*\hdist) {};
\foreach \t in  {0,...,5}{
\draw[black]  (\bitdistance*\t+\bitsize*\t,6*\hdist) -- (\bitdistance*\t+\bitsize*\t, \bitsize*0.5+10pt) ;
\node[bit]  at (\bitdistance*\t+\bitsize*\t,0)   {};
}
\foreach \t in {0,...,2}{
\node[LR] at (\bitsize+\transfoffset+\transfdistance*\t,0) {\LR};
}
\foreach \t in {0,...,1}{
\node[G] at (\bitsize+\transfoffset+\transfsize*0.5+\transfdistance*\t,\hdist) {$\mathcal{G}$};
}
\foreach \t in {0,...,2}{
\node[G] at (\bitsize+\transfoffset+\transfdistance*\t,4*\hdist) {$\mathcal{G}$};
}
\draw[fill=mycolor,rounded corners=10pt,draw=white,opacity=1] (\bitsize+\transfoffset+\transfsize*0.5+\transfdistance*2,\bitdistance+\halfsize)--(\bitsize+\transfoffset+\transfdistance*2,\bitdistance+\halfsize)--(\bitsize+\transfoffset+\transfdistance*2,2*\hdist)--(\bitsize+\transfoffset+\transfsize*0.5+\transfdistance*2,2*\hdist);
\draw[black,rounded corners=10pt, line width=1pt,opacity=1] (\bitsize+\transfoffset+\transfsize*0.5+\transfdistance*2,\bitdistance+\halfsize)--(\bitsize+\transfoffset+\transfdistance*2,\bitdistance+\halfsize)--(\bitsize+\transfoffset+\transfdistance*2,2*\hdist)--(\bitsize+\transfoffset+\transfsize*0.5+\transfdistance*2,2*\hdist);

\draw[fill=mycolor,rounded corners=10pt,draw=white,opacity=1](-\bitsize-\transfoffset,\bitdistance+\halfsize)--(\halfleft,\bitdistance+\halfsize)--(\halfleft,2*\hdist)--(-\bitsize-\transfoffset,2*\hdist);
\draw[black,rounded corners=10pt, line width=1pt,opacity=1] (-\bitsize-\transfoffset,\bitdistance+\halfsize)--(\halfleft,\bitdistance+\halfsize)--(\halfleft,2*\hdist)--(-\bitsize-\transfoffset,2*\hdist);

\end{tikzpicture}},
\label{eq:Gev}
\end{equation}
and to consider the
Schmidt decomposition of the projection $P_{\Lambda_x}$
\begin{equation}
P_{\Lambda_x}= 
=\sum_{\mu_x}\lambda_{\mu_x}\gt\rho_{\mu_x}.
\label{eq:SchmidtLR}
\end{equation}
We can recast the necessary and sufficient condition for $\tT$ to be $(2,\Pi_{\Lambda_0})$-renormalisable in terms of $\mathcal G$ and of the operators $\{\lambda_{\mu_x}, \rho_{\mu_x}\}$.
\begin{proposition}\label{prop:FDFC_ren}
  Let $\tT$ be an FDFC on $\Z$.  $\tT$ is  $(2,\Pi_{\Lambda_0})$-renormalisable iff 
  for every Schmidt decomposition $\sum_{\mu}{\lambda}_\mu\gt{\rho}_\mu$ of $P_{\Lambda_x}$, there exists a Schmidt decomposition $\sum_\mu \widetilde\rho_{\mu}\gt\widetilde\lambda_{\mu}$  of $\mathcal G^{\dagger}(P_{\Lambda_x})$ 
such that
\begin{align}
 \mathcal G (\rho_{\mu}\gt\lambda_{\nu}) =\widetilde\rho_{\mu}\gt\widetilde\lambda_{\nu}\quad \forall \mu,\nu\,.
\label{eq:cond}
\end{align}
\end{proposition}

\begin{proof}
Inverting the second layer of $\mathcal G$'s on both sides of~\eqref{eq:RenIndex} we have the equivalent condition
\begin{equation*}
\resizebox{0.85\hsize}{!}{
\begin{tikzpicture}
\foreach \t in  {0,...,5}{
\draw[black]  (\bitdistance*\t+\bitsize*\t,3*\hdist) -- (\bitdistance*\t+\bitsize*\t, \bitsize*0.5+10pt) ;
\node[bit]  at (\bitdistance*\t+\bitsize*\t,0)   {};
}
\foreach \t in {0,...,2}{
\node[LR] at (\bitsize+\transfoffset+\transfdistance*\t,0) {\LR};
}
\foreach \t in {0,...,1}{
\node[G] at (\bitsize+\transfoffset+\transfsize*0.5+\transfdistance*\t,\hdist) {$\mathcal{G}$};
}

\draw[fill=mycolor,rounded corners=10pt,draw=white,opacity=1] (\bitsize+\transfoffset+\transfsize*0.5+\transfdistance*2,\bitdistance+\halfsize)--(\bitsize+\transfoffset+\transfdistance*2,\bitdistance+\halfsize)--(\bitsize+\transfoffset+\transfdistance*2,2*\hdist)--(\bitsize+\transfoffset+\transfsize*0.5+\transfdistance*2,2*\hdist);
\draw[black,rounded corners=10pt, line width=1pt,opacity=1] (\bitsize+\transfoffset+\transfsize*0.5+\transfdistance*2,\bitdistance+\halfsize)--(\bitsize+\transfoffset+\transfdistance*2,\bitdistance+\halfsize)--(\bitsize+\transfoffset+\transfdistance*2,2*\hdist)--(\bitsize+\transfoffset+\transfsize*0.5+\transfdistance*2,2*\hdist);

\draw[fill=mycolor,rounded corners=10pt,draw=white,opacity=1]  (-\bitsize-\transfoffset,\bitdistance+\halfsize)--(\halfleft,\bitdistance+\halfsize)--(\halfleft,2*\hdist)--(-\bitsize-\transfoffset,2*\hdist);
\draw[black,rounded corners=10pt, line width=1pt,opacity=1] (-\bitsize-\transfoffset,\bitdistance+\halfsize)--(\halfleft,\bitdistance+\halfsize)--(\halfleft,2*\hdist)--(-\bitsize-\transfoffset,2*\hdist);

\end{tikzpicture}=\begin{tikzpicture}
\foreach \t in  {0,...,5}{
\draw[black]  (\bitdistance*\t+\bitsize*\t,1) -- (\bitdistance*\t+\bitsize*\t, \bitsize*0.5+2pt) ;
\node[bit]  at (\bitdistance*\t+\bitsize*\t,0)   {};
}
\foreach \t in {0,...,2}{
\node[LR] at (\bitsize+\transfoffset+\transfdistance*\t,0) {\LRbar};
}

\end{tikzpicture}}\ ,
\end{equation*}
where
\begin{align*}
\begin{tikzpicture}
\foreach \t in  {1,...,2}{
\draw[black]  (\bitdistance*\t+\bitsize*\t,1) -- (\bitdistance*\t+\bitsize*\t, \bitsize*0.5+2pt) ;
\node[bit]  at (\bitdistance*\t+\bitsize*\t,0)   {};
}
\node[LR] at (\bitsize+\transfoffset+\transfdistance*0.5,0) {\LRbar};
\end{tikzpicture}\coloneqq \mathcal G^\dag(P_{\Lambda_x})=G\left(\sum_{\mu}{\lambda}_\mu\gt{\rho}_\mu\right)G^\dag.
\end{align*} 
In formulas,
\begin{align}
  \begin{aligned}
&\sum_{\boldsymbol\mu}\bigboxtimes_{x\in\mathbb{L}'} \mathcal G(\rho_{\mu_{x-1}}\gt\lambda_{\mu_x})= \\
 & =\bigboxtimes_{y\in\mathbb{L}'}\sum_{\mu_y}\mathcal G^\dag(\lambda_{\mu_y}\gt\rho_{\mu_y}) \\
 &=\bigboxtimes_{y\in\mathbb{L}'}\sum_{\mu_y}\bar{\lambda}_{\mu_y}\gt\bar{\rho}_{\mu_y},
  \end{aligned}
\label{eq:Gconv}
\end{align}
where the sum over $\boldsymbol\mu$ is shorthand for the sum over all the indices $\mu_x$
for $x\in\mathbb{L}'$, and $\sum_\mu\bar{\lambda}_{\mu}\gt\bar{\rho}_{\mu}$ is 
some Schmidt decomposition of $\mathcal G^{\dag}(P_{\Lambda_x})$.
Inverting half of the $\mathcal G$'s we obtain another equivalent condition
\begin{equation*}
\resizebox{0.85\hsize}{!}{
\begin{tikzpicture}
\foreach \t in  {0,...,7}{
\draw[black]  (\bitdistance*\t+\bitsize*\t,3*\hdist) -- (\bitdistance*\t+\bitsize*\t, \bitsize*0.5+10pt) ;
\node[bit]  at (\bitdistance*\t+\bitsize*\t,0)   {};
}
\foreach \t in {0,...,3}{
\node[LR] at (\bitsize+\transfoffset+\transfdistance*\t,0) {\LR};
}

\node[G] at (\bitsize+\transfoffset+\transfsize*0.5+\transfdistance,\hdist) {$\mathcal{G}$};

\draw[fill=mycolor,rounded corners=10pt,draw=white,opacity=1] (\bitsize+\transfoffset+\transfsize*0.5+\transfdistance*3,\bitdistance+\halfsize)--(\bitsize+\transfoffset+\transfdistance*3,\bitdistance+\halfsize)--(\bitsize+\transfoffset+\transfdistance*3,2*\hdist)--(\bitsize+\transfoffset+\transfsize*0.5+\transfdistance*3,2*\hdist);
\draw[black,rounded corners=10pt, line width=1pt,opacity=1] (\bitsize+\transfoffset+\transfsize*0.5+\transfdistance*3,\bitdistance+\halfsize)--(\bitsize+\transfoffset+\transfdistance*3,\bitdistance+\halfsize)--(\bitsize+\transfoffset+\transfdistance*3,2*\hdist)--(\bitsize+\transfoffset+\transfsize*0.5+\transfdistance*3,2*\hdist);

\draw[fill=mycolor,rounded corners=10pt,draw=white,opacity=1] (-\bitsize-\transfoffset,\bitdistance+\halfsize)--(\halfleft,\bitdistance+\halfsize)--(\halfleft,2*\hdist)--(-\bitsize-\transfoffset,2*\hdist);
\draw[black,rounded corners=10pt, line width=1pt,opacity=1] (-\bitsize-\transfoffset,\bitdistance+\halfsize)--(\halfleft,\bitdistance+\halfsize)--(\halfleft,2*\hdist)--(-\bitsize-\transfoffset,2*\hdist);

\end{tikzpicture}=\begin{tikzpicture}
\foreach \t in  {0,...,7}{
\draw[black]  (\bitdistance*\t+\bitsize*\t,3*\hdist) -- (\bitdistance*\t+\bitsize*\t, \bitsize*0.5+2pt) ;
\node[bit]  at (\bitdistance*\t+\bitsize*\t,0)   {};
}
\foreach \t in {0,...,3}{
\node[LR] at (\bitsize+\transfoffset+\transfdistance*\t,0) {\LRbar};
}
\foreach \t in {0,...,1}{
\node[G] at (\bitsize+\transfoffset+\transfsize*0.5+\transfdistance*\t*2,\hdist) {$\mathcal G^\dag$};
}

\end{tikzpicture}},\ 
\end{equation*}
which in formulas reads
\begin{align*}
\begin{aligned}
\sum_{\boldsymbol{\mu}} \bigboxtimes_{x \in \mathbb{L}'} \lambda_{\mu_{2x}} \gt \mathcal G (\rho_{\mu_{2x}} \gt \lambda_{\mu_{2x+1}})  \gt \rho_{\mu_{2x+1}} =\\
= \sum_{\boldsymbol{\nu}} \bigboxtimes_{x \in \mathbb{L}'}\bar{\lambda}_{\nu_{2x-1}} \gt \mathcal G^\dag \left( \bar{\rho}_{\nu_{2x-1}} \gt \bar{\lambda}_{\nu_{2x}} \right) \gt \bar{\rho}_{\nu_{2x}}.
\end{aligned}
\end{align*}
Comparing the factorisations on the left and right hand side, and considering the tiles
$\Lambda_{x},\Lambda_{x+1}$, we should have
\begin{align*}
\sum_{\mu,\nu} \lambda_\mu \gt \mathcal G(\rho_\mu\gt \lambda_\nu)  \gt \rho_\nu = \sum_{\mu,\nu}{\lambda}_\mu\gt {\widetilde{\rho}}_\mu\gt {\widetilde{\lambda}}_\nu\gt {\rho}_\nu,
\end{align*}
for suitable $\widetilde{\lambda}_\mu, \widetilde{\rho}_\mu$, 
or diagrammatically
\begin{align}
\resizebox{0.85\hsize}{!}{
\begin{tikzpicture}
\foreach \t in  {0,...,3}{
\draw[black]  (\bitdistance*\t+\bitsize*\t,3*\hdist) -- (\bitdistance*\t+\bitsize*\t, \bitsize*0.5+2pt) ;
\node[bit]  at (\bitdistance*\t+\bitsize*\t,0)   {};
}
\foreach \t in {0,...,1}{
\node[LR] at (\bitsize+\transfoffset+\transfdistance*\t,0) {\LR};
}

\node[G] at (\bitsize+\transfoffset+\transfsize*0.5,\hdist) {$\mathcal{G}$};



\end{tikzpicture}=\begin{tikzpicture}
\foreach \t in  {0,...,3}{
\draw[black]  (\bitdistance*\t+\bitsize*\t,1) -- (\bitdistance*\t+\bitsize*\t, \bitsize*0.5+2pt) ;
\node[bit]  at (\bitdistance*\t+\bitsize*\t,0)   {};
}
\node[LR] at (\bitsize+\transfoffset,0) {\LRtilde};
\node[LR] at (\bitsize+\transfoffset+\transfdistance,0) {\LtilR};

\end{tikzpicture}}\ .
\end{align}
Using the  linear independence of all factors in the terms of the Schmidt decomposition $\lambda_\nu, \rho_\mu$ we get to
\begin{equation}
\mathcal G (\rho_{\nu}\gt\lambda_{\mu})={\widetilde{\rho}}_{\nu}\gt{\widetilde{\lambda}}_{\mu},
\label{eq:tilde}
\end{equation}
for all possible combinations of $\mu,\nu$.
By inserting this result in Eq.~\eqref{eq:Gconv} we get 
\begin{align}
\begin{aligned}
\sum_{\boldsymbol{\mu}}\bigboxtimes_{x\in\mathbb{L}'}{\bar{\lambda}}_\mu\gt {\bar{\rho}}_\mu 
=\sum_{\boldsymbol{\mu}}\bigboxtimes_{x\in\mathbb{L}'}\widetilde{\lambda}_\mu\gt \widetilde{\rho}_\mu\,,
\end{aligned}
\end{align}
and therefore
$\sum_{\mu}{\widetilde{\lambda}}_\mu\gt {\widetilde{\rho}}_\mu$ is another Schmidt 
decomposition of $\mathcal G^{\dag}(P_{\Lambda_x})$. The converse statement is now straightforward.
\end{proof}

\begin{Corollary}\label{cor:proj_decomp}
    Let the projection $P_{\Lambda_x}$ admit a  Schmidt decomposition of  the following form \[P_{\Lambda_x}=\frac{\rank(P_{\Lambda_x})}{d^\abs{\Lambda_x}}I\gt I+ \sum_{\mu \neq 0} \lambda_\mu \gt \rho_\mu \,,\]
    where $\Tr(\lambda_\mu)=\Tr(\rho_\mu)=0$. Then, up to swaps, 
    \begin{align*}
      \mathcal{G}(\rho_\mu\gt I)=\widetilde{\rho}_\mu \gt I\,,\quad  \mathcal{G}(I\gt \lambda_\mu)=I \gt \widetilde{\lambda}_\mu\,.
    \end{align*}
\end{Corollary}
\begin{proof}
    Let $\lambda_0,\rho_0=I$. As $\mathcal G(\rho_0 \gt \lambda_0  )=I \gt I$, then from \Eq\eqref{eq:cond} we have  $\widetilde{\lambda}_0,\widetilde{\rho}_0=I$.  Hence $\mathcal G(\rho_\mu \gt \lambda_0  )=\widetilde{\rho}_\mu \gt I$ and $\mathcal G(\rho_0 \gt \lambda_\mu  )=I \gt \widetilde{\lambda}_\mu $, where $\Tr(\widetilde{\rho}_\mu)=\Tr(\widetilde{\lambda}_\mu)=0$.
\end{proof}

It is easy to see that the renormalisability conditions~\eqref{eq:cond} are equivalent to
\begin{align}
  \mathcal G
  (  {\rho}^{(n)}\gt  {\lambda}^{(n)})
  =
\widetilde{\rho}^{(n)} \gt   \widetilde{\lambda}^{(n)}\quad \forall\,\rho^{(n)},\lambda^{(n)}\,,
\label{eq:mon}
\end{align}
where ${\rho}^{(n)},  {\lambda}^{(n)},
\widetilde{\rho}^{(n)}$ and  $ \widetilde{\lambda}^{(n)}$ denote monomials of order $n$ in
$\rho_\mu$, $ \lambda_\mu$, $ \widetilde{\rho}_\mu$, and 
$\widetilde{\lambda}_\mu $ respectively, i.e.
\begin{align*}
  &\lambda^{(n)}=\lambda_{\nu_1}...\lambda_{\nu_n}\,,
    \quad
   \rho^{(n)}= \rho_{\mu_1}...\rho_{\mu_n}\,, \\
     &\widetilde{\lambda}^{(n)}=\lambda_{\nu_1}...\lambda_{\nu_n}\,,\quad
\widetilde{\rho}^{(n)}=\rho_{\mu_1}...\rho_{\mu_n} \,.
\end{align*}
In particular let $\rho^{(n)}$ and $r^{(n)}$ be two different monomials, then $\mathcal{G}((\rho^{(n)}-r^{(n)}) \gt \lambda^{(n)})=(\widetilde{\rho}^{(n)}-\widetilde{r}^{(n)}) \gt \widetilde{\lambda}^{(n)}$, meaning that the action of $\mathcal G$ on one factor is independent of the other.
Then we have the following powerful result.
\begin{Corollary}
\label{lmm:supportalegbrasofP}
  Let us denote with
  $\mathsf{M}$, $\mathsf{N}$, $ \widetilde{\mathsf{M}}$,
  and $\widetilde{\mathsf{N}} $ the algebras generated
  by $\rho_\mu$, $ \lambda_\mu$, $ \widetilde{\rho}_\mu$, and 
$\widetilde{\lambda}_\mu $, respectively. Then, for any
  $ A \in \mathsf{M}$, $B \in \mathsf{N}$ there exist
$\widetilde{A} \in \widetilde{\mathsf{M}}$ and
$\widetilde{B}\in   \widetilde{\mathsf{N}}$
such that
\begin{equation}
\begin{split}
\mathcal G (A\gt B)=\widetilde{A}\gt\widetilde{B}\,.
\end{split}
\label{eq:GAlg}
\end{equation}
Moreover, let $I_{\m M}\in\m M$ and $I_{\m N}\in\m N$ be maximal projections\footnote{Note that $I_{\m M}$ and $I_{\m N}$ \emph{might} coincide with $I \in \m A_x$.}, then 
\begin{equation}\label{eq:maxproj}
    \mathcal G(I_{\m M} \gt B)=I_{\widetilde{\mathsf{M}}}\gt \widetilde{B}\,, \quad \mathcal G(A \gt I_{\m N} )=\widetilde{A}\gt I_{\widetilde{\mathsf{N}}}\,,
\end{equation}
where $I_{\widetilde{\m M}}\in\widetilde{\m M} $ and   $I_{\widetilde{\m N}}\in\widetilde{\m N} $ are maximal projections.
\end{Corollary}
\begin{proof}
  Since $P_{\Lambda_x}^2=P_{\Lambda_x}$, we have
  \begin{align}
    \label{eq:6}
    \lambda_\mu=\sum_{j,k}
  c_\mu^{j,k}\lambda_j\lambda_k \quad 
  \rho_\nu=\sum_{j,k}d_\nu^{j,k}\rho_j\rho_k\,,  
  \end{align}
  for suitable coefficients $c_\mu^{j,k}$ and
  $d_\nu^{j,k}$. By iteratively using Equation~\eqref{eq:6} we can express any homogeneous polynomial of degree $n$
  in the variables $\lambda_\mu$ as a homogeneous polynomial of degree $n+k$ for any $k$. The same holds for homogeneous polynomial in the variables $\rho_\mu$.
  Then, any arbitrary $A \in \mathsf{M}$ and $B \in \mathsf{N}$ can be written as homogeneous polynomials of the same degree $n$ for sufficiently large $n$.
  The thesis follows from Eq.~(\ref{eq:mon}).

\  Now let $I_{\m N}$ and $I_{\m M}$ be the maximal projections of $\m N, \m M$, respectively. Let $I_{\m N}=\sum_\mu c_\mu \lambda_\mu^{(m)}$ and $I_{\m M}=\sum_\mu d_\mu\rho_\mu^{(m)}$ for suitable $m=n+k$ as above. By unitarity $\mathcal G(I_{\m M} \gt I_{\m N})=I_{\widetilde{\m M}} \gt I_{\widetilde{\m N}}$, where  $I_{\widetilde{\m N}}=\sum_\mu c_\mu\widetilde{\lambda}_\mu^{(m)}$, $I_{\widetilde{\m M}}=\sum_\mu d_\mu \widetilde{\rho}_\mu^{(m)}$ are the maximal projections of $\widetilde{\m N},\,\widetilde{\m M}$. From \eqref{eq:tilde}  $\mathcal{G}(\rho_\mu \gt \lambda_\nu)= \widetilde{\rho}_\mu \gt  \widetilde{\lambda}_\nu$ for any $\mu$ and $\nu$ with no mutual dependence of the indices. It follows that the same equation holds for any monomial in $\rho_\mu,\lambda_\nu$ of the same degree, and by linearity for any polynomial. 
Hence, for any $B\in \m N$ we have
$\mathcal{G}(I_{\m M} \gt B)=I_{\widetilde{\m M}} \gt \widetilde{B}$ for some $\widetilde{B}\in \widetilde{\m N}$, and similarly for $A, \widetilde{A}$ as in the statement of the Corollary. 
Notice that Proposition \ref{prop:diffind} shows that $\mathcal G$ acts as a swap (up to cell-wise isomorphisms) iff $P_{\Lambda_x}=E^a \gt E^{a \oplus 1}$ with $a=0,1$, $E^0$ a maximal projector and $E^1$ an idempotent. Thus, no ambiguity arises in Eq.\eqref{eq:maxproj}.
\end{proof}

Corollary \ref{lmm:supportalegbrasofP} provides a practical criterion for assessing the $(2,\Pi_{\Lambda_0})$-renormalisability of an FDFC. Specifically, the preservation of the factorization of some algebra $\m M \gt \m N$ by $\mathcal G$ is a necessary condition for the $(2,\Pi_{\Lambda_0})$-renormalisability of the associated FDFC. However, this condition is not sufficient, since it does not ensure the validity of Eq.~\eqref{eq:cond}.
On the other hand, Corollary \ref{cor:proj_decomp} states a necessary and sufficient condition for $(2,\Pi_{\Lambda_0})$-renormalisability of an FDFC when the projection has a specific form.

In the case where $\mathcal G$ is symmetric under fermionic swap, Corollary \ref{lmm:supportalegbrasofP} can be further refined as follows.
The fermionic swap gate \[\tS(O_x O_{y}) \coloneqq (-)^{g(O_x)g(O_{y})}O_{y}O_{x}\] exchanges two operators. On $\M{1}{1} \gt \M{1}{1}$ the SWAP can be represented via conjugation by
\begin{align}
    S_{\leftrightarrow}= e^{\frac{\pi}{4} {X}_ x\gt {X}_ y}e^{\frac{\pi}{4} {Y}_ x\gt {Y}_ y}\,.
\end{align}
\begin{Corollary}
\label{cor:Gcommuteswithswap}
If $\mathcal G \circ \tS=\tS\circ \mathcal G$  
then 
\begin{gather*}
    \mathcal G^2(P_{\Lambda_x})=P_{\Lambda_x}\,,
\end{gather*}
and for all $A \gt B \in \mathsf{K}$ there exists $\widetilde{A} \gt \widetilde{B} \in \widetilde{\mathsf{K}}$ such that
\begin{align*}
  \begin{aligned}
\mathcal G (A\gt B)=\widetilde{A}\gt\widetilde{B}\,,     
  \end{aligned}
\end{align*}
with 
\begin{align*}
 \begin{aligned}
    \mathsf{K}:=
    \mathsf{M}\gt\mathsf{N} \lor
    \mathsf{N}\gt\mathsf{M} \,,
    \\
    \widetilde{\mathsf{K}}:=
        \widetilde{\mathsf{M}}\gt \widetilde{\mathsf{N}} \lor
  \widetilde{  \mathsf{N}} \gt  \widetilde{\mathsf{M}} \,,
  \end{aligned}
\end{align*}
where $\lor$ denotes the smallest algebraic closure of the union of two algebras.
\end{Corollary}
\begin{proof}
If $\mathcal G$ commutes with the fermionic swap $\tS$, Eq.\eqref{eq:cond} implies that for all $\mu,\nu$
  \begin{align}
    \label{eq:9}    (-1)^{g(\lambda_\mu)g(\rho_\nu)}\lambda_{\mu}\gt\rho_{\nu}=(-1)^{g(\widetilde{\lambda}_\mu)g(\widetilde{\rho}_\nu)}
    \mathcal G^\dag (
\widetilde{\lambda}_{\mu}\gt\widetilde{\rho}_{\nu}).
  \end{align}
  Evaluating Eq.\eqref{eq:maxproj} with $B=\lambda_\mu$ and $\widetilde{B}=\widetilde{\lambda}_\mu$ it follows that $g(\lambda_\mu)=g(\widetilde{\lambda}_\mu)$ for all $\mu$. Similarly, $g(\rho_\mu)=g(\widetilde{\rho}_\mu)$ for all $\mu$.
  Then we have
  $\mathcal G^{\dag2}(\lambda_{\mu}\gt\rho_{\mu}) = \mathcal G^\dag(\mathcal G^\dag (\lambda_{\mu}\gt\rho_{\mu})   = \mathcal G^\dag(
    \widetilde{\lambda}_{\mu}\gt\widetilde{\rho}_{\mu})    =(-1)^{g(\widetilde{\lambda}_\mu)g(\widetilde{\rho}_\mu)}
   (-1)^{g(\lambda_\mu)g(\rho_\mu)}\lambda_{\mu}\gt \rho_{\mu}=\lambda_{\mu}\gt \rho_{\mu}  $. Therefore, $\mathcal{G}^2(P_{\Lambda_x})=\mathcal{G}^2(\sum_\mu \lambda_\mu \gt \rho_\mu)=\sum_\mu \lambda_\mu \gt \rho_\mu=P_{\Lambda_x}$.
  From Eq. \eqref{eq:9} we have that
Eq. \eqref{eq:GAlg} 
holds for every operator generated by $\{\lambda_\mu\gt\rho_\nu\}
\cup \{\rho_\nu\gt\lambda_\mu\}$.
\end{proof}

We stress that, since in the case of $C^*$-algebras Wedderburn-Artin's Theorem implies that the group of gates that map factorised elements to factorised elements is generated by all factorised gates and the swap, Corollary \ref{cor:Gcommuteswithswap} establishes a necessary and sufficient condition for $(2,\Pi_{\Lambda_0})$-renormalisability in the case of quantum cellular automata~\cite{trezzini2024renormalisationquantumcellularautomata}. However, in the $\Z_2$ graded case, the group is larger. As the next section of the manuscript addresses the renormalisation of FCA with $\m A_x \cong \M{1}{1}$, we analyse such group in this setting. 
A more geneal result holding for $\m A_x \cong\M{2^{d-1}}{2^{d-1}}$ is postponed to future work.

\begin{lemma}\label{cor:fact}
Let $\m A_1, \m A_2\cong \M{1}{1}$  and let $\mathcal{F}$ be a fermionic gate on $\m A_1\gt \m A_2$. Then there exist two algebras $\widetilde{\m A}_1, \widetilde{\m A}_2\cong \M{1}{1}$ such that 
\begin{align}
\begin{aligned}
\mathcal{F}(A_1\gt A_2)= \widetilde{A}_1\gt \widetilde{A}_2 \quad \forall A_i\in\m A_i 
\end{aligned}
\end{align}
 with $ \widetilde{A}_i\in\widetilde{\m A}_i$,
iff, up to swaps, 
\begin{align}
\mathcal F=(\mathcal U_1 \gt \mathcal U_2 )\, \mathcal{C}_{\nu,n} 
\end{align}
with $\mathcal U_i$ fermionic gates, and $\mathcal{C}_{\nu,n}$ acts by conjugation via
\[C_{\nu,n}= \dyad{a} \gt I + \dyad{b} \gt e^{i\nu} Z^n\,, \]
$a,n =0,1$ and $b=a \oplus 1$.
\end{lemma}

\begin{proof}
We begin by considering the image under $\mathcal F$ of 
an odd Majorana generator $M_j\boxtimes I$, with $j=1,2$. We have 
\begin{align*}
F^\dag(M_j\gt  I)F =  F_j\gt E_j,
\end{align*}
where $F_j\boxtimes E_j \in \M{1}{1}$ must be odd. 
Consider now the gate action on a linear combination of Majorana generators
\begin{align*}
F^\dag[(M_1+M_2)\gt  I] F= F_1\gt E_1 + F_2 \gt E_2.
\end{align*}
The requirement of factorisation of the image imposes that either $F_1=F_2$ or $E_1=E_2$, and one fo the following holds:
\begin{align*}
F^\dag[(M_1+M_2)\gt  I] F &= (F_1+F_2) \gt E,\\
F^\dag[(M_1+M_2)\gt  I] F &= F\gt (E_1+E_2),
\end{align*} 
This implies that, up to swap, we have
\begin{align}\label{eq:condone}
F^\dag(M_j\gt  I)F =  F_j\gt E,
\end{align}
Then, since $\acomm{F^\dag(M_i\gt I)F}{F^\dag(M_j\gt I)F}=\delta_{i,j} I$ and $g(F_i)g(E)=0$, it follows
\begin{align*}
\acomm{F_i}{F_j}=\delta_{ij}I,\quad E^2=I.
\end{align*}
From this, one can easily conclude that $F_i, F_j$ must be odd graded-commuting generators, while $E$ is even.
Exactly the same argument can be used for the action on $ I\gt M_j$, to conclude that
\begin{align*}
F^\dag( I \gt M_i)F&= E'\gt F'_i,\; \text{or}\\
F^\dag( I \gt M_i)F&= E'_i\gt F',
\end{align*}
However, the second case cannot hold, because otherwise $F_1,F_2,E'_1,E'_2$ would be \emph{four} graded-commuting odd generators $\gamma_l$ abiding by $\gc{\gamma_i}{\gamma_j}=\delta_{ij}I$ in $\M{1}{1}$ (and $EF'$). 
Thus, we conclude that the only possibility is 
\begin{align}\label{eq:condtwo}
F^\dag(I\gt M_i)F= E'\gt F'_i.
\end{align}
 We can then compute the action over a generic polynomial of a pair of generators by
\begin{align*}
F^\dag(M_i\gt M_j)F&= F^\dag(M_i\gt  I)F F^\dag( I\gt M_j)F=\\ &=(F_i\gt E)(E' \gt F'_j)=\\
&=(F_iE'\gt EF'_j).
\end{align*}
One can easily realise that, up to a possible swap, conditions~\eqref{eq:condone} and~\eqref{eq:condtwo} along with $\{F_i,F_j\}=\{F_i,F_j\}=\delta_{ij}I$ and $E^2={E'}^2=I$ are necessary and sufficient for $\tF$ to be in the group of gates that map factorised elements to factorised elements.

As $E^2={E'}^2=I$, we find that either $E=\pm I$ or $E=\pm Z$, and similarly for $E'$. Now, since $F_i$ and $F'_i$ are, up to a unitary, exactly $M_i$, our problem boils down to determine those unitaries $H$ such that
\begin{align*}
    &H^\dag(M_i\gt I)H=M_i\gt Z^{n},\\
    &H^\dag(I\gt M_i)H=Z^{m}\gt M_i,
\end{align*}
with $m,n\in\{0,1\}$.
Hence,  
\begin{align*}
    H= \dyad{0}\gt I+ \dyad{1}\gt Z^n.
\end{align*}
Indeed,  the most general gate $\mathcal F$ which preserves factorisation along with the above conditions on the images of generators is given by $\mathcal F= \tS^n(\mathcal U_x \gt\, \mathcal U_y)\mathcal C(\mathcal V_x\gt\mathcal V_y)$, where $\tS$ is the swap, $n\in\{0,1\}$, $\mathcal C$ is a controlled-$Z$ map whose associated matrix is
$H= \dyad{0}\gt I+ \dyad{1}\gt Z$.
\end{proof}

\subsection{Renormalisation of a fermionic circuit to an FCA with a different index}

It is natural to ask whether our renormalisation prescription affects topological properties such as the index of the automaton, and under which conditions this may occur.
In particular, in this section we consider the case in which a one-dimensional nearest neighbour FCA with index equal to
$1$, i.e. an FDFC, admits a $(2,\Pi_{\Lambda_0})$-renormalisation to an FCA
which is given by a left or right shift followed by local unitaries.

Thus, 
we
need to be able to identify a $\Z_2$-graded matrix subalgebra
$\M{2^{d-1}}{2^{d-1}}$ inside the algebra
$\A(\Lambda_x)$, that after two steps is totally supported on the tile $\Lambda_{x\pm1}$. The following proposition
specialises the renormalisability condition to this
case.

\begin{proposition}\label{prop:diffind}
 Let $\tT$ be an FDFC on $\Z$.  $\tT$ is  $(2,\Pi_{\Lambda_0})$-renormalisable to a shift FCA $\mathcal S$  iff
\begin{equation}\label{prp:1tod}
\sum_{\boldsymbol\mu}\bigboxtimes_{x\in\mathbb{L}}\widetilde{\lambda}_{\mu_{x+1}}\gt\widetilde{\rho}_{\mu_{x-1}} =\sum_{\boldsymbol\mu}\bigboxtimes_{x\in\mathbb{L}}\lambda_{\mu_{x}}\gt\rho_{\mu_{x}} ,
\end{equation}
where $\widetilde{\lambda}_{\mu_{x+1}}=\mathcal U( \lambda_{\mu_{x+1}})$ and $\widetilde{\rho}_{\mu_{x-1}}=\mathcal V( \rho_{\mu_{x-1}})$ for some gates $\mathcal U,\mathcal V$, and the only rank-$d$ projections $P_{\Lambda_x}$ are
\begin{equation}
P_{\Lambda_x}==\begin{cases}
\dyad{u}\gt I\\
I\gt\dyad{v}\end{cases}\,,
\end{equation}
where $\dyad{u}$ ($\dyad{v}$) is an idempotent eigenoperator of $\mathcal U$ ($\mathcal V$).
\end{proposition}

\begin{proof}
Assume $\tT$ is  $(2,\Pi_{\Lambda_0})$-renormalisable to a shift FCA $\mathcal S$. We recall the definitions of support algebra and index of an FCA \cite{trezzini2025fermioniccellularautomatadimension, fidkowski2019interacting}, which streamline the proof.
\begin{definition} [Support algebra]
  Let $\mathsf B_1$ and $\mathsf B_2$ be $\Z_2$-graded
  $C^*$-subalgebras of some supermatrix algebra $\M{p}{q}$, and consider a $\Z_2$-graded $*$-subalgebra
  $\mathsf{A}\subset \mathsf B_1\gt \mathsf
  B_2$. Each element $A\in\mathsf{A}$ can be
  expanded uniquely in the form
  $A=\sum_{\mu}A_\mu\gt E_{\mu}$, where
  $\{E_{\mu}\}$ is a fixed basis of $\mathsf
  B_2$. The \emph{support algebra}
  $\mathsf S(\mathsf A,\mathsf B_1)$ of
  $\mathsf A$ in $\mathsf B_1$ is the smallest $\Z_2$-graded
  $C^*$-algebra generated by all $A_\mu$ in such
  an expansion.
\end{definition}

It is easy to prove that the support algebra does not depend on the choice of the basis $\{E_\mu\}$ \cite{Index}.
 Let us introduce the left and right
support algebras around cells $2x$ and $2x+1$ for a nearest neighbour FCA $\tT$ on $\mathbb Z$  as
\begin{equation*}
\begin{split}
&\mathsf{L}_{2x}= \mathsf{S}(\tT(\mathsf{A}_{2x} \gt \mathsf{A}_{2x+1}),(\mathsf{A}_{2x-1}\gt \mathsf{A}_{2x})),\\
&\mathsf{R}_{2x}= \mathsf{S}(\tT(\mathsf{A}_{2x} \gt \mathsf{A}_{2x+1}),(\mathsf{A}_{2x+1}\gt\mathsf{A}_{2x+2})).
\end{split}
\end{equation*}
For any FCA $\tT$ on $\mathbb Z$, the index is defined as 
\begin{equation*}
\ind(\tT) \coloneqq\sqrt{\frac{\operatorname{dim}[\mathsf{L}_{2x}]}{\operatorname{dim}[\mathsf{A}_{2x}]}}
\end{equation*}
and is independent of $x$ \cite{trezzini2025fermioniccellularautomatadimension, fidkowski2019interacting}.

 The support algebra of $\tT^2[\A(\Lambda_x)]$ in $\A(\Lambda_{x-1})$ is a subalgebra $\mathsf L'$ of $\mathsf L_{2x-1}\cap\tT(\mathsf L_{2x})\cong\mathsf \tT^{-1}(\mathsf L_{2x-1})\cap\mathsf L_{2x}$, and that of $\tT^2[\A(\Lambda_x)]$ in $\A(\Lambda_{x+1})$  is a subalgebra $\mathsf R'$ of $\mathsf R_{2x+1}\cap\tT(\mathsf R_{2x})\cong\mathsf \tT^{-1}(R_{2x+1})\cap\mathsf R_{2x}$. Starting from an FCA $\tT$ of index $1$, we have $\dim\mathsf L=\dim\mathsf R=d^2$. If we  consider a projection $\Pi$ with rank $d$, and a general subalgebra $\mathsf S$ of $\mathsf A(\Lambda_{x\pm1})$ clearly we have
$\dim(\Pi\mathsf S\Pi)\leq\dim(\mathsf S)$ and thus, if we want 
$\dim(\Pi\mathsf L'\Pi)=d^2$ or $\dim(\Pi\mathsf R'\Pi)=d^2$,
we need 
\begin{align*}
&\mathsf L'=\tT^{-1}(\mathsf L_{2x-1})\cap\mathsf L_{2x}\equiv\tT^{-1}(\mathsf L_{2x-1})\equiv\mathsf L_{2x}\,,\\
&\, \textrm{or }\quad\mathsf R'=\tT^{-1}(\mathsf R_{2x+1})\cap\mathsf R_{2x}\equiv\tT^{-1}(\mathsf R_{2x+1})\equiv\mathsf R_{2x}\,.
\end{align*}
This implies that, to have a renormalisation to a left shift or right shift,  $\tT(\mathsf L)\equiv\mathsf L$ or $\tT(\mathsf R)\equiv\mathsf R$, respectively. However, $\tT(\mathsf L)\equiv\mathsf L$ if and only if 
$\tT(\mathsf R)\equiv\mathsf R$.

It follows that $\mathcal M_2$ must map the algebra in its right input cell to 
$\tT^{-1}(\mathsf L)$ and that in its left input in $\tT^{-1}(\mathsf R)$, i.e., it must be a swap followed by local isomorphisms followed by $\mathcal M^{-1}_1$. Therefore
the composition of $\begin{tikzpicture}
\draw[black]  (\bitdistance+\bitsize,-1)--(\bitdistance+\bitsize,1);
\draw[black]  (0,-1)--(0,1);
\node[M2] at (\bitsize+\transfoffset,0) {$\mathcal{M}_2$};
\end{tikzpicture}$ and $$ results in a swap up to local unitaries---diagrammatically
\begin{equation}
==\begin{tikzpicture}
\foreach \t in  {0,1}{
\draw[black]  (\bitdistance*\t+\bitsize*\t,0) -- (\bitdistance*\t+\bitsize*\t, 1.5*\hdist) ;
}
\node[U] at (0,\hdist*0.25) {$\mathcal U$};
\node[U] at (\bitdistance+\bitsize,\hdist*0.25) {$\mathcal{V}$};
\draw[black]  (0,1*\nhdist) -- (\bitdistance+\bitsize,0) ;
\draw[black]  (0,0) -- (\bitdistance+\bitsize,1*\nhdist) ;
\end{tikzpicture}\,.\\
\label{eq:Uswap}
\end{equation}
In this way Eq.$(\ref{eq:RenIndex})$ becomes
\begin{equation}
\resizebox{0.8\hsize}{!}{\begin{tikzpicture}[baseline=(O.base)]
\node (O) at (0,2*\hdist) {};
\foreach \t in  {0,...,5}{
\draw[black]  (\bitdistance*\t+\bitsize*\t,0.5*\hdist) -- (\bitdistance*\t+\bitsize*\t, \bitsize*0.5+2pt) ;
\node[bit]  at (\bitdistance*\t+\bitsize*\t,0)   {};

}
\foreach \t in {0,...,2}{
\node[LR] at  (\bitsize+\transfoffset+\transfdistance*\t,0) {\LR};
}
\foreach \t in {1,3}{
\tikzmath{
\j=\t+1;
}
\draw[black] (\bitdistance*\t+\bitsize*\t,0.5*\hdist)--(\bitdistance*\j+\bitsize*\j,2.5*\hdist);
\draw[black] (\bitdistance*\j+\bitsize*\j,0.5*\hdist)--(\bitdistance*\t+\bitsize*\t,2.5*\hdist);

}
\foreach \t in {0,2,4}{

\draw[black]  (\bitdistance*\t+\bitsize*\t,2.5*\hdist) -- (\bitdistance*\t+\bitsize*\t,5.5*\hdist) ;
\node[U]  at (\bitdistance*\t+\bitsize*\t,4*\hdist) {$\mathcal{V}$};
}
\foreach \t in {1,3,5}{

\draw[black]  (\bitdistance*\t+\bitsize*\t,2.5*\hdist) -- (\bitdistance*\t+\bitsize*\t,5.5*\hdist) ;
\node[U]  at (\bitdistance*\t+\bitsize*\t,4*\hdist) {$\mathcal{U}$};
}
\foreach \t in {0,2,4}{
\tikzmath{
\j=\t+1;
}
\draw[black] (\bitdistance*\t+\bitsize*\t,5.5*\hdist)--(\bitdistance*\j+\bitsize*\j,7.5*\hdist);
\draw[black] (\bitdistance*\j+\bitsize*\j,5.5*\hdist)--(\bitdistance*\t+\bitsize*\t,7.5*\hdist);
}
\foreach \t in {0,2,4}{
\draw[black]  (\bitdistance*\t+\bitsize*\t,7.5*\hdist) -- (\bitdistance*\t+\bitsize*\t,10.5*\hdist) ;
\node[U]  at (\bitdistance*\t+\bitsize*\t,9*\hdist) {$\mathcal{U}$};
}
\foreach \t in {1,3,5}{
\draw[black]  (\bitdistance*\t+\bitsize*\t,7.5*\hdist) -- (\bitdistance*\t+\bitsize*\t,10.5*\hdist) ;
\node[U]  at (\bitdistance*\t+\bitsize*\t,9*\hdist) {$\mathcal{V}$};
}
\draw[black] (0,0.5*\hdist)--(-\bitdistance*0.5-\bitsize*0.5,1.5*\hdist);
\draw[black] (-\bitdistance*0.5-\bitsize*0.5,1.5*\hdist)--(0,2.5*\hdist);
\draw[black] (5*\bitdistance+5*\bitsize,0.5*\hdist)--(\bitdistance*5.5+\bitsize*5.5,1.5*\hdist);
\draw[black] (\bitdistance*5.5+\bitsize*5.5,1.5*\hdist)--(5*\bitdistance+5*\bitsize,2.5*\hdist);
\end{tikzpicture}=\begin{tikzpicture}[baseline=(O.base)]
\node (O) at (0,2*\hdist) {};
\foreach \t in  {0,...,5}{
\draw[black]  (\bitdistance*\t+\bitsize*\t,11*\hdist) -- (\bitdistance*\t+\bitsize*\t, \bitsize*0.5+2pt) ;
}
\foreach \t in {0,...,2}{
\node[LR] at (\bitsize+\transfoffset+\transfdistance*\t,0) {\LR};
}

\end{tikzpicture}},
\label{eq:not1}
\end{equation}
Using the decomposition of projection $P_{\Lambda_x}$ in Eq.$(\ref{eq:SchmidtLR})$ and computing the left side of the previous equation we get
\begin{equation*}
\begin{tikzpicture}[baseline=(O.base)]
\node (O) at (0,2*\hdist) {};
\foreach \t in  {0,...,5}{
\draw[black]  (\bitdistance*\t+\bitsize*\t,0.5*\hdist) -- (\bitdistance*\t+\bitsize*\t, \bitsize*0.5+2pt) ;
\node[bit]  at (\bitdistance*\t+\bitsize*\t,0)   {};

}
\foreach \t in {0,...,2}{
\node[LR] at  (\bitsize+\transfoffset+\transfdistance*\t,0) {\LR};
}
\foreach \t in {1,3}{
\tikzmath{
\j=\t+1;
}
\draw[black] (\bitdistance*\t+\bitsize*\t,0.5*\hdist)--(\bitdistance*\j+\bitsize*\j,2.5*\hdist);
\draw[black] (\bitdistance*\j+\bitsize*\j,0.5*\hdist)--(\bitdistance*\t+\bitsize*\t,2.5*\hdist);

}
\foreach \t in {0,2,4}{

\draw[black]  (\bitdistance*\t+\bitsize*\t,2.5*\hdist) -- (\bitdistance*\t+\bitsize*\t,5.5*\hdist) ;
\node[U]  at (\bitdistance*\t+\bitsize*\t,4*\hdist) {$\mathcal{V}$};
}
\foreach \t in {1,3,5}{

\draw[black]  (\bitdistance*\t+\bitsize*\t,2.5*\hdist) -- (\bitdistance*\t+\bitsize*\t,5.5*\hdist) ;
\node[U]  at (\bitdistance*\t+\bitsize*\t,4*\hdist) {$\mathcal{U}$};
}
\foreach \t in {0,2,4}{
\tikzmath{
\j=\t+1;
}
\draw[black] (\bitdistance*\t+\bitsize*\t,5.5*\hdist)--(\bitdistance*\j+\bitsize*\j,7.5*\hdist);
\draw[black] (\bitdistance*\j+\bitsize*\j,5.5*\hdist)--(\bitdistance*\t+\bitsize*\t,7.5*\hdist);
}
\foreach \t in {0,2,4}{
\draw[black]  (\bitdistance*\t+\bitsize*\t,7.5*\hdist) -- (\bitdistance*\t+\bitsize*\t,10.5*\hdist) ;
\node[U]  at (\bitdistance*\t+\bitsize*\t,9*\hdist) {$\mathcal{U}$};
}
\foreach \t in {1,3,5}{
\draw[black]  (\bitdistance*\t+\bitsize*\t,7.5*\hdist) -- (\bitdistance*\t+\bitsize*\t,10.5*\hdist) ;
\node[U]  at (\bitdistance*\t+\bitsize*\t,9*\hdist) {$\mathcal{V}$};
}
\draw[black] (0,0.5*\hdist)--(-\bitdistance*0.5-\bitsize*0.5,1.5*\hdist);
\draw[black] (-\bitdistance*0.5-\bitsize*0.5,1.5*\hdist)--(0,2.5*\hdist);
\draw[black] (5*\bitdistance+5*\bitsize,0.5*\hdist)--(\bitdistance*5.5+\bitsize*5.5,1.5*\hdist);
\draw[black] (\bitdistance*5.5+\bitsize*5.5,1.5*\hdist)--(5*\bitdistance+5*\bitsize,2.5*\hdist);
\end{tikzpicture}=\sum_{\boldsymbol\mu}\bigboxtimes_{x\in\mathbb{L}}\widetilde{\lambda}^{(x)}_{\mu_{x+1}}\gt\widetilde{\rho}^{(x)}_{\mu_{x-1}} ,
\end{equation*}
where $\widetilde{\lambda}^{(x)}_{\mu_{x+1}}= \mathcal U^2\left({\lambda}^{(x)}_{\mu_{x+1}}\right)$ and $\widetilde{\rho}^{(x)}_{\mu_{x-1}}=\mathcal V^2\left({\rho}^{(x)}_{\mu_{x-1}}\right)$.
Here we introduced a superscript to keep track of the the tile $\Lambda_x$ on which an operator acts $\textit{after}$ the evolution, while the subscript is a summation index that is reminiscent of the tile on which the operator acted $\textit{before}$ the evolution. In this way Eq.$(\ref{eq:not1})$ becomes
\begin{equation}
\sum_{\boldsymbol\mu}\bigboxtimes_{x\in\mathbb{L}}\widetilde{\lambda}^{(x)}_{\mu_{x+1}}\gt\widetilde{\rho}^{(x)}_{\mu_{x-1}} =\sum_{\boldsymbol\mu}\bigboxtimes_{x\in\mathbb{L}}\lambda^{(x)}_{\mu_{x}}\gt\rho^{(x)}_{\mu_{x}} \,.
\label{eq:diffind}
\end{equation}

The other way, if  Eq.\eqref{eq:diffind} holds, then no $(2,\Pi_{\Lambda_0})$-renormalisation 
other than the one mapping an FDFC to a shift is allowed.
On the left of Eq.\eqref{eq:diffind} we have that the summation index $\mu$ is the same for $\widetilde{\lambda}^{(x)}_\mu$ and $\widetilde{\rho}^{({x+2})}_\mu$. So if we multiply both sides by a fixed 
$\left(\widetilde{\rho}^{({x+2})}_\nu\right)^\dag \in \A(\Lambda_{x+2})$ and 
$\left(\widetilde{\lambda}^{({x-2})}_\gamma\right)^\dag \in \A(\Lambda_{x-2})$, and take the partial trace on the corresponding spaces, the expression on both sides is left such that the operator in $\A(\Lambda_{x})$ is factorised from the rest, and 
we have
\begin{equation*}
\widetilde{\lambda}^{({x})}_{\nu}\gt\widetilde{\rho}^{({x})}_{\gamma}=\sum_{\mu_x}\lambda^{(x)}_{\mu_x}\gt\rho^{(x)}_{\mu_x}.
\end{equation*}
The equality holds if only one term in the expansion \eqref{eq:SchmidtLR} of $P_{\Lambda_x}$ survives, i.e.
\begin{equation*}
\begin{tikzpicture}[baseline=(O.base)]
\node (O) at (0,0.25*\hdist) {};
\foreach \t in  {1,...,2}{
\draw[black]  (\bitdistance*\t+\bitsize*\t,3*\hdist) -- (\bitdistance*\t+\bitsize*\t, \bitsize*0.5+2pt) ;
}
\node[M1] at (\bitsize+\transfoffset+\transfdistance*0.5,\hdist) {$\mathcal{M}_1$};
\node[pi] at  (\bitsize+\transfoffset+\transfdistance*0.5,0) {$\Pi$};
\end{tikzpicture}==\begin{cases}
\ket{u}\bra{u}\gt I\\
I\gt\ket{v}\bra{v}\end{cases}\,.
\end{equation*}
Moreover, in order for Eq.$(\ref{eq:not1})$ to be
satisfied, $\dyad{u}$  ($\dyad{v}$)
should be eigenoperator of the cell-wise unitary map $\mathcal U$ ($\mathcal V)$ in
Eq.$(\ref{eq:Uswap})$.  It follows that our
projection will select only the right (left)
algebra and the renormalised evolution will
consist in a right (left) shift.
\end{proof}

One may also wonder whether an FDFC on $\Z$ can be renormalised to a Majorana shift. In the following section, we show that in the case of $(2,\Pi_{\Lambda_0})$-renormalisations of FDFCs on $\Z$ with a single fermionic mode per cell this cannot occur (Table~\ref{tab:res}). We do not address the general case here, but we conjecture that for any $(N,\Pi_{\Lambda_0})$-renormalisations with even $N$ such a possibility is excluded. On the other hand, the $(2,\Pi_{\Lambda_0})$-renormalisations of Majorana shifts on $\Z$ with a single fermionic mode per cell always change the index (Table~\ref{tab:res}).

\section{Renormalisation of one-dimensional spinless FCA}\label{sec:Qubits}
In this section, leveraging the tools developed so far, we classify the $(2,\Pi_{\Lambda_0})$-renormalisations of nearest-neighbour FCA on $\Z$ with a single (spinless) fermionic mode per cell  $\m A_x \cong \M{1}{1}$. 

We refer to renormalisations shown in Fig.~\ref{fig:CG} as \emph{trivial renormalisations} and we do not discuss them further. 

In Table~\ref{tab:res} we list all the non-trivial $(2,\Pi_{\Lambda_0})$-renormalisable FCA, along with their $(2,\Pi_{\Lambda_0})$-renormalised FCA.
The detailed proofs are presented in Appendix \ref{sec:appendice_ren_spinless}. Despite there are more renormalisable FCA than QCA,
similarly to results found for qubit chains \cite{trezzini2024renormalisationquantumcellularautomata}, all $(2,\Pi_{\Lambda_0})$-renormalisable FCA fail to propagate information except through trivial shifts. Renormalisable FCA are either factorised or possess commuting Margolus layers, and this allows to reduce multiple time steps by merging all odd (even) layers into a new Margolus scheme, to recast the dynamics as that of an effective single-step evolution with unchanged neighbourhood size. Furthermore, as noted in the remark \ref{rmk:bosoniz} below, the projections $P_e,P_o$ map the fermionic algebra to a bosonic one, whereas the remaining projections in Table~\eqref{tab:res} preserve the fermionic nature of cells. Direct inspection of Table~\eqref{tab:res} shows that the only nontrivial fixed point corresponds to case $\phi=\theta=(2n\pi)/3$ via projection $P_o$, thus ultimately describing a quantum cellular automaton of qubits. Genuine fermionic cellular automata admit only shifts as fixed points of the $(2,\Pi_{\Lambda_0})$-renormalisation flow.

\begin{table*}[ht]
\centering
\renewcommand{\arraystretch}{2.2}
\setlength{\tabcolsep}{8pt}

\begin{tabular}{@{}c c c c@{}}
\toprule[1.5pt]
\multicolumn{4}{c}{\large\textbf{Renormalised FCA}} \\
\midrule[1pt]
\textbf{FCA} & \multicolumn{3}{c}{\textbf{projection}} \\
\cmidrule(lr){2-4}

\textbf{Schumacher–Werner} 
& $P_e$ & $P_o$ & $P_L$, $P_R$ \\

\midrule

\multirow{2}{*}{
$\begin{cases}
\phi\neq n\pi,\\[2pt]
U=e^{i\theta Z}
\end{cases}$
}
& 
$\begin{cases}
\phi'=2\phi,\\[2pt]
\theta'=\phi-4\theta
\end{cases}$ 
& 
\multirow{2}{*}{
$\begin{cases}
\phi'=-2\phi,\\[2pt]
\theta'=\phi 
\end{cases}$
}
& 
\multirow{2}{*}{
$\begin{cases}
\phi'=0,\\[2pt]
\theta'=\pm 2(\theta+\delta_{c,0}\phi)
\end{cases}$
} \\

& 

$\begin{cases}
\phi'=2\phi,\\[2pt]
\theta'=4\theta-3\phi
\end{cases}$ & & \\

\midrule

$\begin{cases}
\phi\neq n\pi,\\[2pt]
U=\cos(\theta) X+\sin (\theta) Y
\end{cases}$
& 
$\begin{cases}
\phi'=2\phi,\\[2pt]
\theta'=-\phi
\end{cases}$
&
$\begin{cases}
\phi'=-2\phi,\\[2pt]
\theta'=\phi
\end{cases}$
&
$\begin{cases}
\phi'=0,\\[2pt]
\theta'=\mp(2c-1)\phi
\end{cases}$ \\

\midrule[1pt]
\textbf{Forking} & $P_L$ & $P_R$ & \\

\midrule

$U=e^{i(2n+1)\frac{\pi}{4}Z}$ 
& $\tau_-$ & $\tau_+$ & \\

$U=\tfrac{1}{\sqrt{2}}(X \pm Y)$
& $\tau_-$ & $\tau_+$ & \\

\midrule[1pt]
\textbf{Majorana Shift} & $\Pi^Y$ & $\Pi^X$ & \\

\midrule

$\begin{cases}
U=e^{-i(2n+1)\frac{\pi}{2}Z},\\[2pt]
\sigma_{\pm}
\end{cases}$
& $U'$ & $\tau_{\pm}$ & \\

\bottomrule[1.5pt]
\end{tabular}

\caption{
Summary of the renormalisation flow of nearest-neighbour spinless FCA on $\Z$.  
For each FCA and its parameters, the non-trivial $(2,\Pi_{\Lambda_0})$-renormalised FCA associated with specific projections $\Pi_{\Lambda_0}$ (or equivalently $P_{\Lambda_0}$) are listed.  
Definitions:  
$P_e=\ketbra{0}{0}\gt\ketbra{0}{0}+\ketbra{1}{1}\gt\ketbra{1}{1}$,  
$P_o=\ketbra{0}{0}\gt\ketbra{1}{1}+\ketbra{1}{1}\gt\ketbra{0}{0}$,  
$P_L=I\gt\ketbra{c}{c}$,  
$P_R=\dyad{c}\gt I$,  with $c=0,1$,
$\Pi^X=\tfrac{1}{2}(I\gt I \pm iX\gt X)$,  
$\Pi^Y=\tfrac{1}{2}(I\gt I \pm iY\gt Y)$, and $n\in \Z$. Angles are defined modulo $2 \pi n$.  
Primed quantities denote renormalised counterparts. For $\phi=n \pi$ or shifts $\tau_\pm$, $(2,\Pi_{\Lambda_0})$-renormalisation is trivial, yielding to cell-wise unitaries and shifts, respectively, independent of $\Pi_{\Lambda_0}$.
}
\label{tab:res}
\end{table*}

\begin{remark}\label{rmk:bosoniz}
    The projections $P_{e,o}=\frac{1}{2}(I\gt I \pm Z\gt Z)$ select the even ($+$) and odd ($-$) sectors of the algebra. 
    
    The coisometries $\tE^\dag$ associated to $P_{e}=\ketbra{e}{e}$ and $P_{o}=\ketbra{o}{o}$ map a $\Z_2$-graded algebra onto an ungraded one---to be precise, onto  a $\Z_2$-graded algebra whose odd sector is empty. 
Indeed, as a generic coisometry associated to a rank-2 projection on $\m A^\Pi(\Lambda_x)$ acts as
\begin{align*}
\tE^\dag(O)= (\ket{\psi_0}\bra{a_0}+\ket{\psi_1}\bra{a_1})\,O\,(\ketbra{a_0}{\psi_0}+\ketbra{a_1}{\psi_1}),
\end{align*}
where $a_i\in\{e,o\}$ 
and $\ketbra{\psi_0}{e}, \ketbra{\psi_1}{o}$ are even and 
$\ketbra{\psi_1}{e}, \ketbra{\psi_0}{o}$ are odd maps,

the coarse-grained parity operator $Q_x^r = \tE^\dag( Q_{\Lambda_x})$ can be one of the following: 
\begin{align*}
Q^r_x= \begin{cases}
\text{sgn}(a)\,I \quad\: \text{ if }\; a_0=a_1=a\,,\\
\text{sgn}(a)\, Z \quad \text{ if }\; a_1\neq a_0=a\,,
\end{cases}
\end{align*}
where $a \in \{e,o\}$ and $\text{sgn}(e)\coloneqq +$, $\text{sgn}(o)\coloneqq -$.
Different choices of $\tV^\dag$ map the algebra $\m A(\Lambda_x)$, graded by $Q_{\Lambda_x}$, onto an algebra $\m A'_x$ graded by  $Q^r_x$. 
When $Q_x^r$ is proportional to the identity it commutes with all operators in the algebra $\m A'_x$, so every element is even. This recovers the well-known result that the even sector of two fermionic modes is isomorphic to the algebra of a single bosonic mode. 
\end{remark}

\section{Conclusions}
This work introduces a new framework for analysing and manipulating the large-scale behavior of fermionic cellular automata, providing both conceptual insight into the structure of discrete fermionic dynamics and practical tools for the simulation and control of those systems. Our procedure does not aim to produce a continuum field theory, but instead preserves the discrete structure. This allows us to stay entirely within a class of physically implementable, digital dynamical systems, while still probing their large-scale behavior. 

In the simple case study of one-dimensional spinless FCA, we find that extending to fermionic cellular automata enlarges the zoology of renormalisable evolutions with respect to the qubit case \cite{trezzini2024renormalisationquantumcellularautomata}, yet the same conceptual limitation remains: no renormalisable automaton propagates information in a nontrivial way. This may stem from the intrinsic simplicity of the chosen model, and one may conjecture that local algebras of larger dimension admit richer renormalisable dynamics. 

Interestingly, FCA that after $N$ time steps differ by composition with finite–depth fermionic circuits preserving the projection \(\Pi\), are $(N,\Pi_{\Lambda_0})$-renormalised to the same FCA. In this sense, those FCA belong to the same symmetry–protected phase, where the protecting symmetry is enforced by $\Pi$.

Finally, our procedure is designed as an \emph{exact} renormalisation: no error is allowed in the comparison between the original and the renormalised evolutions. As such, it provides a rigorous baseline for developing \emph{approximate} renormalisation schemes, where some constraints are relaxed, e.g., to obtain the best local-unitary approximation after tracing out selected degrees of freedom \cite{Rotundo}.
Another promising avenue is to connect this work with the operator-algebraic renormalisation of lattice field theories, investigated in \cite{PhysRevLett.127.230601, Osborne:2023aa, Luijk:2024aa}. Such generalisations could reveal richer renormalisation flows, universality classes, and continuum limits. Nevertheless, as an exact renormalisation, our approach should correctly identify the fixed points of those approximate renormalisation schemes.

\acknowledgments
PP acknowledges financial support from European Union—Next Generation EU through the
MUR project Progetti di Ricerca d’Interesse
Nazionale (PRIN) QCAPP No. 2022LCEA9Y.
AB acknowledges financial support from
European Union—Next Generation EU
through the MUR project Progetti di Ricerca
d’Interesse Nazionale (PRIN) DISTRUCT No.
P2022T2JZ9. LT acknowledges financial support
from European Union—Next Generation EU
through the National Research Centre for HPC,
Big Data and Quantum Computing, PNRR
MUR Project CN0000013-ICSC.

\bibliographystyle{unsrt}
\bibliography{CGFQCA.bib}
\appendix

\section{$\Z_2$-graded algebras}\label{app:z2algebras}

 \begin{definition}[$\mathbb{Z}_2$\textit{-graded vector space}]
     A $\mathbb{Z}_2$-graded $\K$-vector space $\mathsf{A}$ is a  pair
\begin{equation*}
\mathsf{A}=(\mathsf{A}^0,\mathsf{A}^1) \coloneqq \m A^0 \sqcup \m A^1 \,,
\end{equation*} 
where $\mathsf{A}^{p}$, $p \in \Z_2$, are two vector spaces over a field $\mathbb K$. The sum is defined only between elements of $\mathsf{A}$ belonging to the same vector space 
\begin{align*}
	O_1,O_2\in \mathsf{A}^p, \qquad O_1+O_2\in \mathsf{A}^p.
\end{align*}
 \end{definition}
\begin{definition}[$\mathbb{Z}_2$\textit{-graded algebra}]\label{def:g_algebra}
A unital $\mathbb{Z}_2$-graded  $\mathbb K$-algebra $\mathsf{A}$ is a  $\Z_2$ graded $\K$-vector space
\begin{equation*}
\mathsf{A}=(\mathsf{A}^0,\mathsf{A}^1) \coloneqq \m A^0 \sqcup \m A^1 \,,
\end{equation*} 
endowed with an associative distributive product: \begin{align*}
	O_1\in\mathsf{A}^p,O_2\in\mathsf{A}^q, \qquad O_1O_2\in\mathsf{A}^{p\oplus q},
\end{align*}
where $\oplus$ denotes sum modulo 2, and $I_{\m A}\in \m A^0$ is the unit element.
If $\m A^1=\{0\}$, then $\m A=\A^0$ as associative $\K$-algebra. 
\end{definition}
We stress that according to Definition \ref{def:g_algebra}, referring to an element $O$ of $\m A$ means $O \in \m A^p$. We refer to the dimension of the algebra (vector space) as $\dim(\m A)=\dim(\m A^0)+\dim(\m A^1)$. When no ambiguity can arise, we write $I$ for the unit element of the algebra. 
Note that a $\Z_2$-graded algebra is an instance of $G$-graded algebra where $(G,\circ)$ is any monoid, $\m A=\bigsqcup_{p\in G} \m A^p$ with $A_pA_q\in \m A^{p\circ q}$ for $A_r\in \m A^r$ \cite{HWANG199973}.
\begin{definition}[Parity.]
    Let $\m A$ be $\Z_2$-graded algebra. Defining the map $g:\m A \to \Z_2$ such that $g(O)=p$ for any $O \in \m A^p$, we call $g(O)$ the parity of the element $O$.  We refer to $O$ as an even operator if $g(O)=0$ and as an odd operator if $g(O)=1$.
\end{definition}
The notion of parity naturally extends to maps. Given a map between two $\Z_2$-graded algebras $\eta:\m A \to \m B$, we call $\eta$ an even map if $g(\eta(O))=g(O)$ for any $O \in \m A$ and we call $\eta$ an odd map if $g(\eta(O))=g(O)\oplus 1$ for any $O \in \m A$.
\begin{definition}[$\Z_2$-graded subalgebra]
    Let $\m A$ and $\m B$ two $\Z_2$-graded algebras. We say that $\m B$ is a $\Z_2$-graded subalgebra of $\m A$ if
    \[
    \m B^p=\m B\cap \m A^p.
    \]
    In this case we write $\m B \subseteq \m A$.
\end{definition}

\begin{definition}[Graded morphism]
    Let $\m A$ and $\m B$ be $\Z_2$-graded algebras,  $O_i\in\m A$, and $c_i \in \K$. A graded morphism $\mathcal M:\m A \to \m B$ acts as $\mathcal M(c_1O_1+c_2O_2O_3)=c_1\mathcal M(O_1)+c_2\mathcal M(O_2)\mathcal M (O_3)$ with $\mathcal M(\m A^p)\subseteq\m B^p$. Hence  $g(\mathcal M(O))=g(O)$ for any $O \in \m A$.
\end{definition}
\begin{definition}[Graded commutator]
  Let $\m A$ a $\Z_2$-graded algebra and let $O_1,O_2\in\mathsf{A}$.
  The graded commutator is
\begin{equation*}
\gc{O_1}{O_2} \coloneqq O_1O_2-(-1)^{g(O_1)g(O_2)}O_2O_1.
\end{equation*}
\end{definition}
The graded commutator reduces to the anticommutator if $O_1,O_2$ are both odd elements, 
and to the commutator otherwise.

\begin{definition}[Graded center]
    The graded center $\m Z (\m A)$ of a $\Z_2$-graded $\mathbb K$-algebra $\m A$ is the $\Z_2$-graded subalgebra of $\m A$ given by all the elements of $\m A$ that graded commute with any element of $\m A$. The algebra $\m A$ is called central if $\m Z(\m A)\cap\m A^0=\mathbb K$.
\end{definition}

\begin{definition}[Simple]
    A two-sided ideal $\m I(\m A)$ of a $\Z_2$-graded algebra $\m A$ is a $\Z_2$-graded subalgebra of $\m A$ such that $AO,OA \in \m I(\m A)$ for any $A\in \m A$ and $O \in \m I$.
    A $\Z_2$-graded algebra $\m A$ is called simple if the only two-sided ideals of $\m A$ are $\{0\}$ and $\m A$.
\end{definition}

\begin{definition}[Graded tensor product]
     Let $\m A$ and $\m B$ two $\Z_2$-graded algebras. The graded tensor product algebra $\m A \gt \m B$ is the $\Z_2$-graded algebra given by
     \begin{equation*}
\begin{split}
&\mathsf{A}\boxtimes\mathsf{B}=((\mathsf{A}\boxtimes\mathsf{B})^0,(\mathsf{A}\boxtimes\mathsf{B})^1), \\
&(\mathsf{A}\boxtimes\mathsf{B})^r\coloneqq\bigoplus_{\substack{p,q=0,1 \\ p\oplus q=r}} \mathsf{A}^p\otimes\mathsf{B}^q,
\end{split}
\end{equation*}
where $\otimes$ is the standard tensor product of $\mathbb K$-vector spaces. For every $A \in \m A$ and $B \in \m B$ we denote $A \gt B$ the corresponding element of $\m A \gt \m B$
and
 for any $A_i \in \m A$, $B_i \in \m B$
\begin{align}\label{eq:grtp}
(A_1\boxtimes B_1)(A_2\boxtimes B_2)&=(-1)^{g(B_1)g(A_2)}(A_1A_2\boxtimes B_1B_2)\,.
\end{align}
\end{definition}
\begin{remark}
As a consequence of~\eqref{eq:grtp} the graded tensor product has the following braiding relation 
\begin{align*}
\gc{(A\boxtimes I_{\m B})}{(I_{\m A}\boxtimes B)}=0\,.
\end{align*}
\end{remark}
\begin{remark}
   By definition, $g(I)=0$, and consequently for any $A_i \in \m A$ and $A \gt B \in \m A \gt \m B $ we have $g(A_1A_2)=g(A_1)\oplus g(A_2)$, $g(A\boxtimes B)=g(A)\oplus g(B)$. 
\end{remark}
As in the ungraded case, equipping a $\Z_2$-graded algebra with a suitable norm and an involution, we get a $\Z_2$-graded $C^*$-algebra. 
\begin{definition}[$C^*$-algebra]
    A $C^*$-algebra $\mathsf A$ is an associative $\C$-algebra decorated with an involution $\dagger$ and a norm $\norm{\cdot}$, such that it is complete in the metric induced by the norm, $\norm{AB}\leq \norm{A}\norm{B}$, and $\norm{A^\dagger A}=\norm{A}^2$ for all $A,B\in\mathsf A$.
\end{definition}

\begin{definition}[$*$-morphism]
    Let $\m A$ and $\m B$ be $C^*$-algebras,  $O_i\in\m A$, and $c_i \in \K$. A $*$-morphism $\mathcal M:\m A \to \m B$ acts as $\mathcal M(c_1O_1+c_2O_2^\dag O_3)=c_1\mathcal M(O_1)+c_2\mathcal M(O_2)^\dag \mathcal M (O_3)$.
\end{definition}

\begin{remark}
    If $\m A$ and $\m B$ are $\Z_2$-graded $C^*$-algebras with involution $\dag$, then $\m A \gt \m B$ is a $\Z_2$-graded $C^*$-algebra with the natural inherited $C^*$ structure. To be precise, in general one must specify a $C^*$-norm (minimal, maximal, etc.) on the algebraic tensor product $\m A \gt \m B$. In the finite-dimensional case all such norms coincide, so the distinction is immaterial \cite{blackadar2017operator}. In particular, we choose a norm on the tensor product algebra such that $\norm{A\gt I}=\norm{A}_{\m A}$ and $\norm{I\gt B}=\norm{B}_{\m B}$ for any $A\in \m A$ and $B \in \m B$.  Let $A \gt B \in \m A \gt \m B$, we highlight that
\begin{align*}
    (A\gt B)^\dag&=(I_{\m A}\gt B)^\dag(A \gt I_{\m B})^\dag=\\
    &=(-1)^{g(B)g(A)}(A^\dag\gt B^\dag)\,.
\end{align*}
\end{remark}

\begin{definition}[Generators]
    Let $\m A$ be a finite-dimensional $\Z_2$-graded algebra. We refer to elements $\gamma_1,\ldots,\gamma_n \in \m A$ as generators of $\m A$ if they form a minimal set such that every element of $\m A$ can be expressed as a linear combination of products of the $\gamma_j$. In this case, we write $\m A = \expval{\gamma_1,\ldots,\gamma_n}$. 
\end{definition}

\begin{definition}[Notable algebras]
For, $\mathbb N \ni p,q<\infty$, the  matrix algebra $\M{p}{q}$ is the $\Z_2$-graded $\C$-algebra defined as
\[\M{p}{q}\coloneqq \C^{p^2+q^2}\sqcup\C^{2pq}\,,\]
with  $\C^0\coloneqq\{0\}$. The  Clifford algebra $\Cl_n(\C)$ is the $\Z_2$-graded $\C$-algebra defined as
\[\Cl_n(\C)\coloneqq\C^{2^{n-1}} \sqcup\C^{2^{n-1}}\,,\]
with $ \N^+\ni n <\infty$.
\end{definition}
Notice that $\M{p}{0}=\M{0}{p}=\text{Mat}_{p\times p}(\C)$ is the $\C$-algebra of complex $p\times p$ matrices. Moreover, the Clifford algebra $\Cl_n(\C)$ is the $n$-fold graded tensor product of $\Cl_1(\C)$, and when $n=2k$ is even then $\Cl_{2k}(\C)=\M{2^{k-1}}{2^{k-1}}$. By $\Cl_1^\gamma(\C)=\expval{\gamma}$ we denote the algebra $\Cl_1(\C)$ specifying an odd generator $\gamma$ (note that necessarily $\gamma^2=I$).

\begin{definition}[Superselected state]
 Let $\m A$ be a $\Z_2$-graded $C^*$-algebra.
A state is a linear map $\omega:\m A \to \mathbb C$ with $\omega(A^\dag A)\geq 0$ for all $A \in \m A$, and $\omega(I_{\m A})=1$. A superselected state is a state such that $\omega(\A^p)\subseteq \C^{p\oplus 1}$.
\end{definition}

\begin{definition}[Tracial state]
    Let $\m A$ be $\Z_2$-graded algebra. 
    The tracial state of $\m A$ is a superselected state $\Tr:\m A \to \mathbb C$ such that $\Tr(AB)=\Tr(BA)$ for any $A,B \in \m A$.
\end{definition}

\begin{definition}[Basis]\label{def:basis}
      A basis of $\M{p}{q}$ is a set of elements $B_{ij} \in \M{p}{q}$ with $B_{ij}B_{kl}=\delta_{jk}B_{il}$ and $B_{ij}^\dag=B_{ji}$, which form a basis of $ \M{p}{q}$ as a $\Z_2$-graded vector space.
As a consequence $i,j = 1, \ldots, p+q$, $g(B_{ij})=g(B_{ji})$, $B_{ii}$ are even,  and the number of odd basis elements is $2pq$. 
\end{definition}
In the body of the manuscript we use the notation $\ketbra{i}{j} \coloneqq B_{ij}$.

\begin{remark}
    $\m A= \M{p}{q},\Cl_n(\C)$ are both $C^*$-algebras endowed with the operator norm $\norm{O} \coloneqq \sup_{\omega}\sqrt{\omega(O^\dag O)}$, where $O \in \m A$ and $\omega$ superselected states.\\
    For the algebra $ \M{p}{q}$, a state $\omega$ is superselected if and only if $\omega(O)=\omega(Q\,O\,Q^\dag)$ for all $O \in \M{p}{q}$, where $\M{p}{q}\ni Q \coloneqq\sum_{i=1}^p \dyad{i}-\sum_{j=1}^q \dyad{p+j}$ is called the parity element.
\end{remark}

\begin{remark}
      A common way to introduce $\Z_2$-graded algebras in the literature is through direct sums of vector spaces. Although this definition is algebraically equivalent, we adopt a formulation based on disjoint unions because it makes the superselection structure manifest. In physical terms, the $\mathbb{Z}_2$–grading distinguishes superselection sectors corresponding to even (bosonic) and odd (fermionic) elements, which cannot be coherently superposed.
    In discussing matrix algebras $\M{p}{q}$, the vector spaces $\M{p}{q}^0, \M{p}{q}^1$  can be represented as 
    \begin{align*}
        &\M{p}{q}^0\cong\\
        &\cong\left\{\begin{pmatrix}
            A& 0\\
            0& B
        \end{pmatrix}\;:\; A\in\text{Mat}_{p\times p}(\C)\,, B \in \text{Mat}_{q\times q}(\C)\right\}\,,\\
        &\M{p}{q}^1\cong\\
        &\cong\left\{\begin{pmatrix}
            0& C\\
            D & 0
        \end{pmatrix}\;:\; C\in\text{Mat}_{p\times q}(\C)\,, D \in \text{Mat}_{q\times p}(\C)\right\}\,.
    \end{align*}
    In this case, the parity element of the algebra is represented by
    \begin{align*}
        Q=\begin{pmatrix}
            I_{p\times p} &\\
            &-I_{q \times q}
        \end{pmatrix}\,.
    \end{align*}
    We stress that in this representation the basis elements $\ketbra{i}{j}$ are those per Definition \ref{def:basis}, and should not be confused with ungraded matrix entries as they obey graded commutation relations.
\end{remark}

\begin{definition}[Idempotent]
    Let $\m A$ be $\Z_2$-graded algebra. We refer to $P \in \m A$ as idempotent if $P^2=P$. A primitive idempotent $E \in \m A$ is an idempotent such that $E=Q+R$ implies either $Q=0$ or $R=0$.
\end{definition}
Any idempotent is an even element of the algebra, and can be decomposed as the sum of primitive idempotents. The basis elements $B_{ii}$ of $\M{p}{q}$ are primitive idempotents. Viceversa, given $p+q$ orthogonal primitive idempotents $E_iE_j=\delta_{ij}E_i$, there exists a basis $B_{ij}$ such that $E_i=B_{ii}$.

\begin{definition}[Projection]
    Let $\m A$ be $\Z_2$-graded $C^*$-algebra. A projection $\Pi \in \m A$ is an idempotent of $\m A$ such that $\Pi^\dag=\Pi$. As such, $\Pi=\sum_{i=1}^r E_i$ with $E_i \in \m A$ primitive idempotents. The rank of the projection $\Pi$ is given by $\text{rank}(\Pi) \coloneqq r$.
\end{definition}

Since projections can always be decomposed in terms of primitive idempotents as $\Pi=\sum_{i=1}^rB_{ii}$, we can write projections as $\Pi=\sum_{i=1}^r\dyad{i}$. The element 
$I-\Pi$ is also a projection, and if $\Pi\Pi'=\Pi'\Pi$ and $\rank(\Pi)=\rank(\Pi')$ then $\Pi=\Pi'$.

\begin{remark}
    Let $\m A \cong \M{p}{q}$ and $\m A' \cong \M{p'}{q'}$, $\Pi \in \m A \gt \m A'$ a projection, and $ E_i\in \m A \gt \m A'$, $D_i \in \m A$, $ D'_i \in \m A'$ primitive idempotents. As in general $E_i\neq \sum _{kl} d_{kl} D_k \gt D'_l $, instead of using primitive idempotents, we often conveniently write $\Pi$ in a basis of $\m A \gt \m A'$ $\Pi=\sum_{ijkl}c_{ijkl} B_{ij} \gt B'_{kl}$ where $B_{ij}$ and $B_{kl}$ are a basis of $\m A$ and $ \m A'$ respectively, $d_{kl},c_{ijkl}\in\C$.
\end{remark}

\begin{definition}
    [Unitary elements] Let $\m A$ be a unital $\mathbb Z_2$-graded $C^*$-algebra. An element $U\in\m A$ is \emph{unitary} if $U^\dag U=UU^\dag=I_\m A$
\end{definition}
\begin{lemma}
    Let $\m A$ be a unital $\mathbb Z_2$-graded $C^*$-algebra, $\Pi\in\m A$ be a projection and $U\in\m A$ a unitary. Then $U\Pi U^\dag$ is a projection, and $\rank(U\Pi U^\dag)=\rank(\Pi)$.
\end{lemma}
\begin{definition}
    Let $\m A$ be a $\mathbb Z_2$-graded $C^*$-algebra, and let $\Pi,\Pi'\in\m A$ be projections. We define the ordering $\Pi\geq\Pi'$ if $\Pi\Pi'=\Pi'$.
\end{definition}
\begin{definition}
    Let $\m A$ be a finite-dimensional $\mathbb Z_2$-graded $C^*$-algebra, and $\m A'\subseteq\m A$ a $\mathbb Z_2$-graded $C^*$-subalgebra. We say that a projection $\Pi\in\m A'$ is \emph{maximal} in $\m A'\subseteq\m A$ if, for every projection $\Pi'\in\m A$ such that $\Pi'\geq\Pi$, one has $\Pi'\not\in\m A'$.
\end{definition}

\begin{definition}[Isometries]\label{def:isom}
    Let $\m A \cong \M{p}{q}$ and $\m A'\cong \M{p'}{q'}$ with $p'+q'\leq p+q$, and choose any basis $B_{ij}$ and $B'_{ij}$, respectively. Let $\tF_{i'i,j'j}:\m A' \to \m A$ with $i,j=1,\ldots,p+q$ and $i',j'=1,\ldots,p'+q'$ be maps such that $\tF_{i'i,j'j}(B'_{k'l'})=\delta_{i'k'}\delta_{j'l'}B_{ij}$. Let $H=\{1,2,\ldots,p'+q'\}$, and $f:H \to \{1,2,\ldots,p+q\}$ be injective. The map \[\tE=\sum_{k,l\in H} \tF_{lf(l),kf(k)}\] is an \emph{isometry}. $\tE$ is a $*$-homomorphism, and it is a $*$-isomorphism if $\dim(\m A)=\dim(\m A')$.
\end{definition}
\begin{remark}
    For convenience, we use the formal notation 
    \begin{align*}
        &\tF_{lf(l),kf(k)}(X)=\ketbra{f(l)}{l}X\ketbra{k}{f(k)},
   \end{align*}
    even though $\ketbra{f(l)}{l}$, $\ketbra{l}{f(l)}$, 
    are not rigorously defined as standalone objects. Their use 
    in this respect is completely analogous to the use of bras and kets.
\end{remark}
\begin{definition}
    Let $\m A$ and $\m A'$ be finite-dimensional $\mathbb Z_2$-graded algebras, and let $\Tr_{\m A}$ and $\Tr_{\m A'}$ denote tracial states for $\m A$ and $\m A'$, respectively. Let $\tF:\m A' \to \m A$ be an even linear map. Then its dual $\tF^\dag:\m A \to \m A'$ is the linear map implicitly defined by
\begin{align}
    \Tr_{\m A'}[\tF^\dag(Y)X]\coloneqq \Tr_{\m A}[Y\tF(X)].
\end{align}
\end{definition}
One can straightforwardly prove that the dual of an even linear map is well defined and unique, and that $(\tF^\dag)^\dag=\tF$. The following lemma holds.
\begin{lemma}
    Let $\m A$ and $\m A'$ be as in Definition~\ref{def:isom}. Then $\tF_{i'i,j'j}^\dag=\tF_{jj',ii'}$.
\end{lemma}

\begin{remark}
    Notice that, by definition, we have that 
\begin{align}
        &\tF_{lf(l),kf(k)}(X)\tF_{l'f(l'),k'f(k')}(Y)=\\
        &\quad=\delta_{kl'}\tF_{lf(l),k'f(k')}(X\ketbra{k}{k}Y),\label{eq:FF}\\
        &\tF^\dag_{lf(l),kf(k)}(X)\tF^\dag_{l'f(l'),k'f(k')}(Y)=\\
        &\quad=\delta_{k'l}\tF^\dag_{l'f(l'),kf(k)}(X\ketbra{k'}{k'}Y).\label{eq:FDFD}
\end{align}
\end{remark}

\begin{definition}
    Let $\tE:\m A'\to\m A$ be an isometry. Then $\tE^\dag$ is a \emph{coisometry}.
\end{definition}
From the above definition one can check that if $\tE$ is an isometry, then $\tE^\dag(I_{\m A})=I_{\m A'}$ and $\tE(I_{\m A'})=\Pi$, where $\Pi \in \m A$ is a projection. 

\begin{definition}(Graded direct sum)
    Let $\m A$ and $\m B$ be $\Z_2$-graded algebras, then
    \[\m A \boxplus \m B \coloneqq(\m A^0\oplus\m B^0, \m A^1 \oplus \m B^1)\,,\]
    where $\oplus$ is the standard direct sum of vector spaces.
\end{definition}
The fundamental Theorem of finite-dimensional associative algebras, also known as Wedderburn-Artin Theorem, can be generalised to $\Z_2$-graded algebras as follows.
\begin{theorem}[\cite{Brunshidle01112012}]
  Let $\m A$ be a finite-dimensional $\Z_2$-graded semisimple $\C$-algebra, then
  \begin{align*}
      \m A \cong \bigboxplus_i \M{p_i}{q_i} \gt D^{m_i}
  \end{align*}
  where $D^{m_i}=(\C, \C^{m_i})$, and $m_i\in \{0,1\}$. If the index set of $i$ is larger than one, then $\m A$ is called semi-simple.
\end{theorem}
As every finite-dimensional $\mathbb Z_2$-graded $C^*$-algebra is semisimple, the theorem applies in this case, which is of  interest to us.
Thanks to the graded generalization of the Wedderburn-Artin theorem, it is possible to show that another fundamental result, the Skolem–Noether Theorem, also extends naturally to graded algebras.
\begin{theorem}[\cite{HWANG199973}]\label{thm:skolem}
     Let $\m A \subseteq\m B$ be $\Z_2$-graded central simple $C^*$-algebras. Let $\varphi:\m A \to \m B$ an injective graded $*$-morphism. Then there exists a unitary element $U\in \m B$ such that $\varphi(O)=U^\dag O U$ for any $O \in \m A$. 
\end{theorem}

\section{Renormalisation of one-dimensional spinless FCA}\label{sec:appendice_ren_spinless}
In this appendix, we apply the propositions presented in the body of the manuscipt to the $(2,\Pi_{\Lambda_0})$-renormalisation of FCA on $\Z$ with a single (spinless) fermionic mode per cell. 
Firstly, we examine FCA with index $1$, i.e., those FCA that can be realised as FDFCs. Given the FCA $\tT$---equivalently the associated gate $\mathcal G$ (Eq.~\eqref{eq:definitionofG})---we identify for which gate parameters and projections $\Pi_{\Lambda_0}$---equivalently the associated projections $P$ (Eq.~\eqref{eq:4})---$\tT$ is $(2,\Pi_{\Lambda_0})$-renormalisable to an FCA $\tS$, which then will be explicitly computed. Secondly, we study the necessary and sufficient conditions for $(2,\Pi_{\Lambda_0})$-renormalisation of Majorana shifts, and compute the renormalised FCA for those cases where they exist.

In view of their role in Eq.~\eqref{eq:GAlg} for the renormalisability of FDFCs (Prop.~\ref{prop:FDFC_ren}), we list the subalgebras of $\M{1}{1}$ introducing a lighter notation for the following analysis.
\begin{lemma}
The $\Z_2$-graded $C^*$-subalgebras of $\M{1}{1}\cong\mathcal{C}\ell_2(\C)$ are
\begin{enumerate}
\item $\M{1}{1}$: the matrix algebra itself,\\
\item $\mathcal{C}\ell^\gamma_1(\C)$: a Clifford algebra generated by an odd element $\gamma$,\\
\item $\mathcal{A} \coloneqq \M{1}{1}^0$: the abelian algebra generated by the Majorana element $Z$,
\item $\C\mathcal{I} \coloneqq \M{1}{0}$: the trivial algebra of complex scalars.
\end{enumerate}
\end{lemma}
In the following we denote the cosets sets of unitary elements of $\m A^p_x$ as $SU^p(2)$ with $p=0,1$. 

\subsection{Schumacher-Werner FCA}

In the case of Schumacher-Werner FCA (Eq. \eqref{eq:sw}), the fermionic gate $\mathcal G$ (Eq.~\eqref{eq:definitionofG}) retains the same formal structure as in the qubit setting \cite{trezzini2024renormalisationquantumcellularautomata}, but expressed with Majorana modes:
\begin{align*}
=\begin{tikzpicture}
\draw[black]  (\bitdistance+\bitsize,1.5)--(\bitdistance+\bitsize,-1);
\draw[black]  (0,1.5)--(0,-1);
\node[Cif] at (\bitsize+\transfoffset,0) {$\mathcal{C}_\phi$};
\node[U] at (0,\hdist*0.5) {$\mathcal{U}$};
\node[U] at (\bitdistance+\bitsize,\hdist*0.5) {$\mathcal{U}$};
\end{tikzpicture}\ ,\
=\begin{tikzpicture}
\draw[black]  (\bitdistance+\bitsize,1)--(\bitdistance+\bitsize,-1);
\draw[black]  (0,1)--(0,-1);
\node[Cif] at (\bitsize+\transfoffset,0) {$\mathcal C_\phi$};
\end{tikzpicture}\ ,
\end{align*}

\begin{align*}
\mathcal G= \mathcal M_1 \mathcal M_2= \mathcal C_{\phi}\,(\mathcal U\gt \mathcal U)\,\mathcal C_\phi,
\end{align*}
with two admissible forms taking into account $\Z_2$-grading
\begin{align}
\begin{aligned}\label{eq:Werner's}
&G=C_{\phi}( e^{i\theta  Z}\gt e^{i\theta  Z}) C_\phi\,, \quad  \text{ or} \\
&G= C_{\phi}\,[\mathbf{n}\cdot \mathbf M \gt \mathbf{n}\cdot \mathbf M ]\, C_\phi\,,
\end{aligned}
\end{align}
with $\mathbf{n}\in\mathbb{C}^2$, $n_2/n_1=\tan \theta \in \R$, and $\mathbf M=(X,Y)$. Recall that $\mathcal G(\cdot)=G^\dag(\cdot)G$.
Interestingly, the forms of $G$ in the superselected qubit case are exactly those that are known to support $(2,\Pi_{\Lambda_0})$-renormalisation in the qubit framework~\cite{trezzini2024renormalisationquantumcellularautomata}. This leads to the following result.

\begin{lemma}
\label{lem:WernerFCA}
All Schumacher-Werner FCA satisfy Eq.~\eqref{eq:GAlg} with $\mathsf{M},\mathsf{N}=\mathcal A$. Additionally, if $G = C_{\frac{n\pi}{2}} \, ({U} \gt {U}) \, C_{\frac{n\pi}{2}}$ with ${U} \in SU^{0,1}(2)$ and $n \in \Z$,
then Eq.~\eqref{eq:GAlg} holds with $\mathsf{M}, \mathsf{N} = \M{1}{1}$.

\end{lemma}

\begin{proof}
Parity superselection restricts the form of $G$ to that given in Eq.~\eqref{eq:Werner's}, i.e., $U$ has a fixed parity. According to the analysis in~\cite{trezzini2024renormalisationquantumcellularautomata}, such evolutions are precisely those that preserve factorised elements in a suitable subalgebra. However, within the fermionic setting, we need to differentiate between preserving  factorised elements of an abelian even algebra and those in a factorised Clifford algebra.

Recall that $\dyad{0} U = U \dyad{1}$ for every odd $U$. We can rewrite $G$ as follows:
\begin{align}
\begin{aligned}\label{eq:Geasy}
&G=C_\phi \, (U \gt U) \, C_\phi = (U \gt U) \, C_{2\phi}, \quad \text{if } U \in SU^0(2),\\
&G=C_\phi \, (U \gt U) \, C_\phi = (U \gt U) \, \widetilde{C}_{2\phi}, \quad \text{if } U \in SU^1(2),
\end{aligned}
\end{align}
where we recall that $C_\phi$ can be represented by
\begin{align*}
C_\phi = \ket{0}\bra{0} \gt I +\ket{1}\bra{1} \gt e^{i\frac{\phi}{2} (I-Z)},
\end{align*}
and the modified elements are defined as
\begin{align*}
\widetilde{C}_\beta = \ket{0}\bra{0} \gt e^{i\frac{\beta}{2}(I+ Z)}+\ket{1}\bra{1} \gt  e^{i\frac{\beta}{2}(I- Z)}.
\end{align*}
In both scenarios \eqref{eq:Geasy}, we observe that any factorised element $Z^a\gt Z^b$ remains factorised, as both $C_\phi$ and $U \gt U$ preserve it.
If $G$ also maps factorised elements in the Clifford algebras $\mathcal{C}\ell^\xi_1(\mathbb{C}) \gt \mathcal{C}\ell^\eta_1(\mathbb{C})$ to factorised elements, for any odd elements $\xi, \eta$, then it necessarily maps factorised elements to factorised elements. Indeed, if of both $Z \gt  I $ and $\xi \gt  I $ are mapped to factorised elements, also $Z\xi \gt  I $, is mapped to a factorised element. Similarly for elements of the form $I\gt Z$, $I\gt \eta$ and $I\gt Z\eta$.
By Lemma~\eqref{cor:fact}, in this case $G$ must take the form
\begin{align*}\label{eq:GZ}
\begin{aligned}
&G = (U \gt V) \, H^b\,,\quad b=0,1\\
&H= \dyad{0}\gt I + \dyad{1}\gt Z\,.
\end{aligned}
\end{align*}
Therefore
$\phi = \frac{k\pi}{2}$, $ k \in \mathbb{Z}$.
\end{proof}

From the previous lemma, we recover the same results as in the qubit case \cite{trezzini2024renormalisationquantumcellularautomata} when considering projections in the
$Z$-basis, i.e., when $\m{M}, \m{N} \cong \mathcal{A}$. This case was extensively discussed in \cite{trezzini2024renormalisationquantumcellularautomata}. In addition, in the fermionic setting there exists an FCA that preserves a factorised  matrix algebra. In the following we focus on this novel case, recalling that Eq.~\eqref{eq:GAlg} is a necessary but not sufficient condition for $(2,\Pi_{\Lambda_0})$-renormalisation. We have the following:

\begin{lemma}\label{lem:SWFCA}
For Schumacher-Werner FCA where
\begin{align*}
G = C_{\frac{n\pi}{2}} \, (U \gt U) \, C_{\frac{n\pi}{2}},
\end{align*}
 $(2,\Pi_{\Lambda_0})$-renormalisation is possible with projections in the $Z$-basis.
\end{lemma}

\begin{proof}
We know from Ref.\cite{trezzini2024renormalisationquantumcellularautomata} that these FCA are $(2,\Pi_{\Lambda_0})$-renormalisable with 
\begin{equation}\label{eq:pari}
    \Pi_{\Lambda_0}= \sum_{k,j}c_{kj}Z^k\gt Z^j\,,
\end{equation} 
where $k,j=0,1$ and $Z^0=I$. Therefore it is only left to show that $(2,\Pi_{\Lambda_0})$-renormalisability is not possible for other projections. Assume $ \Pi_{\Lambda_0}$ to be different from \Eq\eqref{eq:pari}. Then $ \Pi_{\Lambda_0}$ must contain an element $\xi\gt \eta$ in its Schmidt decomposition with $\xi,\eta \in \M{1}{1}^1$ odd elements. In the present case, the algebras $\m M$ and $\m N$ defined in Corollary \ref{lmm:supportalegbrasofP} are unital. Then by Corollary \ref{lmm:supportalegbrasofP} it should hold $\mathcal G(I\gt \xi)=I\gt \widetilde{\xi}$ for some odd $\widetilde{\xi}$. However, for the present FCA $ \mathcal G(I\gt \xi)= E \gt \widetilde{\xi}$ with a suitable even element $E\neq I$, as can be easily checked using Eqs.\eqref{eq:Geasy}.
\end{proof}

\subsection{Forking FCA}
For the \textit{Forking Automaton} (Eq.\eqref{eq:frk}) we have
\begin{align*}
=\begin{tikzpicture}
\draw[black]  (\bitdistance+\bitsize,1.5)--(\bitdistance+\bitsize,-1);
\draw[black]  (0,1.5)--(0,-1);
\node[Mswap] at (\bitsize+\transfoffset,0) {$\mathcal{S}(X, Y)$};
\node[U] at (0,\hdist*0.5) {$\mathcal U$};
\node[U] at (\bitdistance+\bitsize,\hdist*0.5) {$\mathcal{U}$};
\end{tikzpicture}\ ,\ 
=\begin{tikzpicture}
\draw[black]  (\bitdistance+\bitsize,1)--(\bitdistance+\bitsize,-1);
\draw[black]  (0,1)--(0,-1);
\node[Mswap] at (\bitsize+\transfoffset,0) {$S(\m X, \m Y)$};
\end{tikzpicture}\ .
\end{align*}
The main result of this section is the following.
\begin{proposition}\label{prp:fork_main}
The Forking automata that admit a non-trivial $(2,\Pi_{\Lambda_0})$-renormalisation are those with unitaries given by
\begin{align*}
U=\begin{cases}
\exp{\left[i\left(\frac{(1+2n)\pi}{4} Z\right)\right]},\; \text{or}\\
\frac{1}{\sqrt{2}}(X \pm Y).
\end{cases}
\end{align*}
Those FCA are $(2,\Pi_{\Lambda_0})$-renormalised to right and left shifts by choosing the projections
\begin{align*}
 P_R= \frac{1}{2}\{ ( I\pm Z)\gt   I \}\,,\;  P_L = \frac{1}{2}\{  I \gt  ( I \pm Z)\},
\end{align*}
respectively.
\end{proposition}
The proof of Proposition \ref{prp:fork_main} is given in the following split in a few lemmas.

The $G$ element associated to the gate $\mathcal G$ associated to the Forking automaton is 
\begin{align*}
\begin{aligned}
&G=S_{ X  ,   Y  }  (U\gt U) S_{  X  , Y  }= e^{\frac{\pi}{4}  X  \gt  Y  } (U\gt U) e^{\frac{\pi}{4}  X  \gt  Y  }=\\
&=\frac{1}{2}( I \gt  I +  X   \gt  Y  ) ( U\gt U)( I \gt  I +  X   \gt  Y  )=\\
&=\frac{1}{2}[U\gt U + ( X  \gt  Y  )(U\gt U) + \\
&+(U\gt U)(  X  \gt  Y  ) + ( X  \gt  Y  )(U\gt U) ( X  \gt  Y  )] .
\end{aligned}
\end{align*}
\begin{lemma}
For Forking automata 
\begin{align*}\label{eq:Gdec}
G=\begin{cases} G_{even}\coloneqq i G_+\, ( Z \gt I) \quad \text{if}\quad U=e^{i\theta   Z}\in SU^0(2)\\
G_{odd} \coloneqq G_-\quad \text{if}\quad U=\sin(\theta)X+\cos(\theta)Y \in SU^1(2)\end{cases}
\end{align*}
with
\begin{align*}
\begin{aligned}
G_{\pm}=& \pm  \sin\theta\cos\theta ( I \gt  I + Z \gt Z) + \\
&   \mp \sin^2\theta  (X  \gt  X) \mp\cos^2\theta  (Y  \gt  Y) \,.\\
\end{aligned}
\end{align*}
\end{lemma}
\begin{proof}
Since $U\in SU^{0,1}(2)$ has definite parity, we have $(U\gt  I)(  I \gt  X  )= (-1)^{g(U)g( X  )}(  I \gt  X  )(U\gt   I)=(-1)^{g(U)}(  I \gt  X  )(U\gt   I)$. Then we can write $G$ as
\begin{align}
\begin{aligned}
G=&\frac{1}{2}[U \gt U +\\
&+ ( -1) ^{g(U)}(  X   U \gt  Y   U)+\\
&+ ( -1) ^{g(U)}( U  X   \gt U  Y  )+ \\
&-  X  U  X  \gt  Y   U  Y  ].
\end{aligned}
\end{align}
Inserting
\(
U=e^{i\theta Z}\in SU^0(2)
\)
in the above expression of $G$  we get 
\begin{align*}
\begin{aligned}
&G=\frac{1}{2}[e^{i\theta Z}\gt e^{i \theta Z} -  X   e^{i\theta Z}  X   \gt  Y  e^{i\theta Z}  Y   +\\
 &+\acomm{ X  \gt  Y  }{ e^{i\theta Z}\gt e^{i\theta Z}}]=\\
&=\frac{1}{2}[e^{i\theta Z}\gt e^{i\theta Z}- e^{-i\theta Z}\gt e^{-i\theta Z} + \\
 &+2 \cos^2(\theta)  X  \gt  Y   - 2 \sin^2(\theta)  Y  \gt  X  ]=\\
&= i \sin\theta\cos\theta (Z\gt   I + I \gt Z)+\\
&+ \cos^2(\theta)  X  \gt  Y   - \sin^2(\theta)  Y  \gt  X  .
\end{aligned}
\end{align*}
We can further simplify $G$ by rewriting it as
\begin{align*}
\begin{aligned}\label{eq:Geven}
&G=i [  \sin\theta\cos\theta (  I \gt   I + Z \gt Z) + \\
&-\cos^2\theta  Y  \gt  Y   - \sin^2\theta  X  \gt  X  ] (Z\gt I) \eqqcolon\\
&= i G_+\, ( Z \gt I) .
\end{aligned}
\end{align*}
On the other hand, if we insert
\(
U=\gamma\in SU^1(2)
\),
we have
\begin{align*}
\begin{aligned}
G=&\frac{1}{2}[\gamma \gt \gamma+\\
&+ ( -1) (  X   \gamma \gt  Y   \gamma)+\\
&+ ( -1)( \gamma X   \gt \gamma  Y  )+\\
&-  X  \gamma  X  \gt  Y   \gamma  Y  ].
\end{aligned}
\end{align*}
Writing $\gamma=a  X+b   Y$ with suitable $a,b\in \C$, we obtain
\begin{equation}
    \label{eq:Gfork}
\begin{aligned}
G&=a^2   X\gt   X+ b^2   Y\gt   Y - ab(  I\gt   I+   Z\gt   Z)\\
&= G_-\,.
\end{aligned}
\end{equation}
By unitarity of $\gamma$, we have $a=e^{i\phi}\sin\theta$, $b=e^{i\phi}\cos\theta$, with $\phi,\theta\in \R$.
\end{proof}
The following Lemma holds.
\begin{lemma}\label{lem:even_U}
Consider $U =e^{i\theta Z} \in SU^0(2)$, then \Eq\eqref{eq:GAlg} is satisfied \emph{iff} $\theta=\frac{n\pi}{4}$. If this is the case, we have $\m M,\, \m N,\, \widetilde{\m M},\, \widetilde{\m N}=\M{1}{1}$.
\end{lemma}
\begin{proof}
Let us compute the action of $\mathcal G$ on the generators of any possible subalgebra of $\M{1}{1} \gt \M{1}{1} $. We have 
\begin{align*}
\begin{aligned}
&G_+^\dag ( \Sigma \gt I)G_+= \\
&[\cos^4\theta  (Y  \Sigma  Y)  +\sin^4\theta  (X  \Sigma  X)  ]\gt   I +\\
&-\cos\theta\sin^3\theta(\acomm{\Sigma}{  X  } \gt   X +i \acomm{\Sigma Z}{X} \gt Y)  +\\
&-\cos\theta^3\sin\theta(\acomm{\Sigma}{  Y  } \gt   Y+i \acomm{\Sigma Z}{Y} \gt X)\,,  \\
\\
&G_+^\dag(I \gt \Sigma) G_+=- \text{SWAP}(\mathcal G( \Sigma \gt I) )=\\
& - I \gt[\cos^4\theta  (Y  \Sigma  Y)   + \sin^4\theta  (X  \Sigma  X)  ]+\\
&+\cos\theta \sin^3\theta  (X \gt \acomm{\Sigma}{  X  }     +i Y \gt \acomm{\Sigma Z}{X} ) +\\
&+\cos^3\theta \sin\theta  (Y \gt \acomm{\Sigma}{  Y  }     +i X \gt \acomm{\Sigma Z}{Y} )\,, \\
\end{aligned}
\end{align*}
with $\Sigma= X  , Y $.
Evaluating the results for the different values of $\Sigma$ we get
\begin{align*}
\begin{aligned}
G_+^\dag (  X   \gt I)G_+ =&-\{\cos^4\theta - \sin^4\theta \} X  \gt   I+\\ &-2(\cos^3 \theta \sin \theta + \cos\theta\sin^3\theta )   I \gt  X\,,   \\
G_+^\dag (  Y   \gt I)G_+ =&-\{\sin^4\theta  - \cos^4\theta \}  Y   \gt   I+\\ &-2(\cos^3 \theta \sin \theta - \cos\theta\sin^3\theta )   I \gt  Y\,,   \\
G_+^\dag ( I \gt  X  )G_+ =&-\{\sin^4\theta -\cos^4\theta  \}  I \gt  X +  \\ &+2(\cos^3 \theta \sin \theta + \cos\theta\sin^3\theta )   X \gt  I\,,   \\
G_+^\dag ( I \gt  Y  )G_+ =&-\{\cos^4\theta -\sin^4\theta   \}   I \gt  Y +\\ &+2(\cos^3 \theta \sin \theta - \cos\theta\sin^3\theta )   Y \gt  I\,. 
\end{aligned}
\end{align*}
Note that $\cos^4\theta-\sin^4\theta=\cos 2\theta$, $\cos^3 \theta \sin \theta + \cos\theta\sin^3\theta=(1/2)\sin 2\theta$, and $\cos^3 \theta \sin \theta - \cos\theta\sin^3\theta=(1/4)\sin 4\theta$.
For the images of $  X  \gt   I,   I\gt  X  $ to be factorised, it is necessary that
\begin{align*}
\theta = \frac{(2n+1)\pi}{4} \lor \theta=\frac{n\pi}{2},
\end{align*}
while for the images of $ Y   \gt   I,   I \gt  Y   $ we need 
\begin{align*}
\theta = \frac{(2n+1)\pi}{4} \lor \theta=\frac{n\pi}{4}.
\end{align*}
The two sets of conditions have the common set of solutions for 
\begin{align}\label{eq:tsol}
\theta=\frac{n \pi}{4}\,.
\end{align}
There exists no value of $\theta$ for which only one generator, but not the others, is kept factorised. Hence, if the factorisation condition is met, then it is met for a matrix algebra $\M{1}{1}$. The argument holds for any choice of generators. Indeed for a generic odd element $\eta=a X  +b Y  $ with $a,b \in \C$, we have
\begin{align*}
\begin{aligned}
&G_+^\dag(\eta\gt I)G_+= aG_+^\dag (X  \gt   I)G_++ bG_+^\dag (Y  \gt   I)G_+=\\
&=-\cos2 \theta(a X  -b Y  )\gt   I +\\
&-   I \gt (a\sin 2\theta  X  +\frac{b}{2}\sin 4\theta  Y  ),\\
&\\
&G_+^\dag(  I\gt \eta)G_+=\\
&=-\cos2 \theta   I\gt (-a X  +b Y  ) +\\
&+ (a\sin 2\theta  X  +\frac{b}{2}\sin 4\theta  Y  )\gt   I,
\end{aligned}
\end{align*}
which lead the same factorisation conditions found in \Eq\eqref{eq:tsol}, as one can expect due to the linear independence of $\mathcal G(X \gt I)$ and $\mathcal{G}(Y \gt I)$ ($\mathcal G(I \gt X)$ and $\mathcal{G}(I \gt Y)$).
Finally, we check if \Eq\eqref{eq:tsol} is the only solution that keeps factorised the $Z$-generators. Let us consider
\begin{align*}
\begin{aligned}
G_+^\dag(Z\gt   I) G_+ =&i\,G_+^\dag (Y  \gt   I) G_+G_+^\dag (X  \gt   I) G_+=\\
&\cos^2 2\theta  (Z\gt   I) +\\
&+i\frac{1}{2}\sin 4\theta [   ( Y   \gt  X)  +\\
&+\cos 2\theta   (X   \gt  Y)  +\\
&+i \sin 2\theta (I\gt Z) ]. 
\end{aligned}
\end{align*}
Given the independence of each subspace, the only way to have three coefficients out of four equal to zero is again to have
\(\theta=\frac{n\pi}{4}\).

This concludes the proof as \Eq\eqref{eq:GAlg} requires $\mathcal G$ to map a factorised algebra $\m M \gt \m N$ into a factorised algebra $\widetilde{\m M} \gt \widetilde{\m N}$. Thus, if for instance $\mathcal G$ would keep factorised the algebra $\langle X\gt X \rangle$ but not $\langle X \gt I\rangle$, the statement of Corollary \ref{lmm:supportalegbrasofP} would not be met since $\langle X \gt X \rangle \neq\langle X\rangle \gt \langle X \rangle$. We used the symbol $\langle O \rangle$ to denote the algebra generated by the element $O$ (see Appendix \ref{app:z2algebras}).
\end{proof}
\begin{lemma}\label{lem:odd_U}
Consider $U  \in SU^1(2)$, then  \Eq\eqref{eq:GAlg} is satisfied \emph{iff}
\begin{align}
U=  X,   Y, \frac{1}{\sqrt{2}}\{  X\pm   Y\}.
\end{align}
If this is the case,  $\m M,\, \m N,\, \widetilde{\m M},\, \widetilde{\m N}=\M{1}{1}$.
\end{lemma}
\begin{proof}
The proof is analogous to that of Lemma \ref{lem:even_U}. We provide the computations for two odd anticommuting generators of $\M{1}{1} \gt \M{1}{1}$, and one can easily retrace the steps of the former proof to obtain that the set of solutions is independent of the choice of generators.

The most general odd unitary $\gamma\in \M{1}{1}$ is $\gamma=e^{i\alpha}(\cos \beta X+\sin \beta Y)$, thus if we pick $aX+bY=\gamma \in SU^1(2)$ then $a,b \in \R$.
Note that
\begin{align*}
    \mathcal S(X,Y)(\Sigma\gt I)&=\frac{1}{2}[\acomm{\Sigma}{X}\gt Y+(\Sigma-X\Sigma X)\gt I]\,,\\
    \mathcal S(X,Y)(I \gt \Sigma)&=\frac{1}{2}[- X \gt \acomm{\Sigma}{Y}+I \gt (\Sigma-Y\Sigma Y) ]\,,
\end{align*}
for $\Sigma=X,Y$. Hence $\tS$ maps $X\gt I \mapsto I \gt Y$, $Y \gt I \mapsto Y \gt I$, $I \gt X \mapsto I \gt X$, and $I \gt Y \mapsto - X\gt I$.
Computing the action over the generators of $\M{1}{1} \gt \M{1}{1}$ we have 
\begin{align*}
\begin{aligned}
&\mathcal G (\Sigma\gt I)=-\tS( X,   Y)\{ (\gamma^\dag \gt \gamma^\dag) \mathcal S(  X,   Y) [\Sigma \gt I]     (\gamma   \gt \gamma)\}\,,\\
&\\
&\mathcal G(  X \gt   I)= -\tS  (X,   Y)[   I \gt \gamma   Y\gamma]=\\
&=-\tS  (X,   Y)[   I \gt (iaZ+bI)(aX+bY)]=\\
&=-\tS  (X,   Y)\{   I \gt [(b^2-a^2)Y+2abX]\}=\\
&=(b^2-a^2)  X \gt   I - 2ab   I \gt   X\,,\\
&\\
&\mathcal G(  Y\gt   I)= -\tS  (X,   Y)[   \gamma   Y\gamma \gt I]=\\
&=- (b^2-a^2)   Y \gt   I-2ab  I \gt   Y \,,\\
&\\
&\mathcal G(  I \gt   X)= -\tS  (X,   Y)[   I \gt \gamma   X\gamma]=\\
&=  -\tS  (X,   Y)\{   I \gt [(a^2-b^2)X+2abY]\}=\\
&=- (a^2-b^2)   I \gt   X + 2ab X \gt I\,,\\
&\\
&\mathcal G(   I\gt   Y)=\tS  (X,   Y)[   \gamma   X\gamma \gt I]=\\
&=(a^2-b^2)  I \gt   Y + 2ab  Y \gt   I\,.
\end{aligned}
\end{align*}
The generators remain factorised iff
\begin{align} \label{eq:cond_sol}
\lvert a\rvert = \lvert b\rvert \lor a=0 \lor b=0.
\end{align}
One can prove the same for a generic linear combination of generators. Following  the same argument discussed in the proof of Lemma \ref{lem:even_U}, the only solutions satisfying the conditions of Corollary~\ref{lmm:supportalegbrasofP} give $\m M\simeq\m N\simeq \M{1}{1}$. Plugging the solutions \eqref{eq:cond_sol} in the expression of $G$ as in \Eq\eqref{eq:Gfork} we get
\begin{align}\label{eq:gs}
\begin{cases}
G=\frac{1}{2}\{  X \gt   X +   Y\gt   Y \pm   I \gt   I \pm   Z\gt   Z)\,,\\
G=   Y \gt   Y\,,\\
G=  X \gt   X\,,
\end{cases}
\end{align}
respectively.
The first expression for $G$ in \eqref{eq:gs} corresponds to the operator $\exp{(\pi/4)X\gt X}\exp{(\pi/4)Y\gt Y}= S_{\leftrightarrow}$ (or minus its adjoint ), whose conjugation implements the fermionic swap, i.e., in this case $\mathcal G= \tS$ (or $-\tS^\dag$). Thus, in any case $\mathcal G$ keeps $\M{1}{1} \gt \M{1}{1} $ factorised.
\end{proof}
The previous two lemmas \ref{lem:even_U}, \ref{lem:odd_U} are summarised in the following proposition.
\begin{proposition}
The Forking Automata that satisfy \Eq\eqref{eq:GAlg} are given by:
\begin{align}
\begin{aligned}\label{eq:Gevfork}
G_{\text{even}}=\begin{cases}
   X\gt    Y \quad \text{if}\quad\theta=n\pi,\\
\\
-  Y \gt      X \quad \text{if}\quad\theta = \frac{\pi}{2}+n\pi,\\
\\
i S_{\leftrightarrow}^\dag (  Z \gt   I) \quad \text{if}\quad\theta = \frac{\pi}{4}+\frac{n\pi}{2},
\end{cases}\\
\\
G_{\text{odd}}= \begin{cases}
  X \gt   X \quad \text{if}\quad U=  X,\\
\\
  Y \gt   Y\quad\text{if}\quad U=  Y,\\
\\
S_{\leftrightarrow} \quad\text{if}\quad U= \frac{1}{\sqrt{2}}\{  X\pm   Y\}.
\end{cases}
\end{aligned}
\end{align}
\end{proposition}
\begin{proof}
Insert the results of the previous two lemmas in the expression of $G$ in \Eq\eqref{eq:Gdec}.
\end{proof}
Finally, to get Proposition \ref{prp:fork_main} it is sufficient to apply Proposition \ref{prop:diffind} to the non-factorised $G$ in \Eq\eqref{eq:Gevfork}. Proposition \ref{prop:diffind} states that if $\mathcal G$ is a fermionic swap up to on-site gates $\mathcal V$, then its $(2,\Pi_{\Lambda_0})$-renormalisations can only be shifts with $\Pi_{\Lambda_0}=I \gt \dyad{\psi}(\dyad{\psi}\gt I)$ where $\dyad{\psi}$ is an eigenoperator of $\mathcal V$. In the present case $\mathcal V$ acts by conjugation via $V=I,Z$. Although when $V=I$ the operator $\dyad{\psi}$ can be any, the only states with definite parity---i.e. well-defined projections in $\m A_{x}$---are eigenstates of $\mathcal Z$, which leads to the projections $P_{L,R}$ in Proposition \ref{prp:fork_main}.

\subsection{Majorana Shift}

The last class of FCA which is left to analyse are the \textit{Majorana Shifts}, with $\operatorname{ind}(\sigma_\pm)=2^{\pm\frac{1}{2}}$.

We recall that the action of a Majorana shift is given by
\begin{align*}
\begin{aligned}\label{eq:Majorana}
&\sigma_\pm(   X_{x})=   Y_{x},\\
&\sigma_\pm(  Y_{x})=   X_{x \pm 1},
\end{aligned}
\end{align*}
and consequently
\begin{align*}
\sigma_\pm(  Z_{x})=   Y_{x}\gt   X_{x \pm 1},
\end{align*}
possibly followed by a local unitary acting on each cell as $\mathcal U(\cdot)= U(\cdot) U^\dag$ with $U=\exp\left(-\frac{\theta}{2} iZ\right)$, i.e.
\begin{align*}
\begin{aligned}
\mathcal{U}\circ\sigma_\pm(   X_{x})= \xi _{x},\\
\mathcal{U}\circ\sigma_\pm(  Y_{x})= \eta_{x \pm 1},\\
\mathcal{U}\circ\sigma_\pm(  Z_{x})= \xi_{x}\gt \eta_{x \pm 1},
\end{aligned}
\end{align*}
for two arbitrary odd anticommuting unitary operators $\xi,\eta$.
\begin{proposition}
    The Majorana shift is $(2,\Pi_{\Lambda_0})$-renormalisable \emph{iff} \[U=\exp\left(-i\frac{2n+1}{2}\pi Z\right)\] through the projections:
    \[
    \Pi^X=\frac{I\gt I \pm iX\gt X}{2}\,, \quad \Pi^Y=\frac{I\gt I \pm iY\gt Y}{2}.
    \]
    Those renormalisable Majorana Shifts act as:
    \[
    \mathcal{U}\circ\sigma_\pm(   X_{x})= \mp X_x\,, \quad \mathcal{U}\circ\sigma_\pm(   Y_{x})= \pm Y_{x\pm1}.
    \]
\end{proposition}
\begin{proof}
Notice that the only renormalisation condition we can impose on such evolution is the one in Theorem \ref{prp:thecoarsegrainingisfinite}, since no circuital implementation can be invoked to simplify the condition. 
We focus, without loss of generality, on the action of the right Majorana shift $\sigma_{+}$.
Two steps of such evolution are given by:
\[
\mathcal{U}\circ\sigma_+\circ\mathcal{U}\circ\sigma_+.
\]  
Let $\mathcal{V}$ be the map such that:
\[
\mathcal{V}\circ\sigma_+=\sigma_+\circ\mathcal{U}.
\]
Since
\[
\sigma_+\circ\mathcal{U}(X_x)=aY_x+bX_{x+1}\quad\sigma_+\circ\mathcal{U}(Y_x)=a X_{x+1}-bY_x,
\]
as, without loss of generality, one can always write $\mathcal U(X)=a  X+b   Y$ and $\mathcal U(Y)= a   Y-b X $ for $a=\cos\theta,b=\sin\theta $ with $\theta \in [0,\pi] $\footnote{Let $\xi=a  X+b   Y$ and $\eta= c   X+d Y $ with $a,b,c,d \in \C$. By requiring $\{\xi, \eta\}=0$ and $\xi^2=\eta^2=I$ one gets that the matrix of coefficients $\begin{pmatrix}
    a & b\\ c& d
\end{pmatrix}$ must be orthogonal. Moreover, as $\xi,\eta$ are hermitian, the coefficients must be real.}, 
we get:
\[
\mathcal{V}(X_x)=aX_x-bY_{x-1}\quad \mathcal{V}(Y_x)=aY_x+bX_{x+1}.
\]
Consider now the renormalisation condition as per  Theorem~\ref{prp:thecoarsegrainingisfinite}
\[
\mathcal{U}\circ\sigma_+\circ\mathcal{U}\circ\sigma_+(\Pi)=\Pi.
\]
We can recast it as:
\[
\mathcal{U}\circ\mathcal{V}\circ\sigma_+^2(\Pi)=\Pi,
\]
i.e.
\begin{equation}\label{eq:majo_shift}
 \mathcal{U}\circ\mathcal{V}\circ\tau_+(\Pi)=\Pi.   
\end{equation}
By composing with $\mathcal{U}^{-1}$ on the left of both members of Eq.\eqref{eq:majo_shift} and defining $\Pi'\coloneqq\mathcal{U}^{-1}(\Pi)$, we can represent Eq.\eqref{eq:majo_shift} in a diagrammatic way as:
\[
\mathcal{V}\left(\resizebox{0.35\hsize}{!}{\begin{tikzpicture}[baseline=(O.base)]
\node (O) at (0,2*\hdist) {};
\foreach \t in  {0,...,5}{
\draw[black]  (\bitdistance*\t+\bitsize*\t,10*\hdist) -- (\bitdistance*\t+\bitsize*\t, \bitsize*0.5+2pt) ;
}
\foreach \t in {0,1}{
\node[pi] at (0.5*\transfdistance+\bitsize+\transfoffset+\transfdistance*\t,0) {$\Pi$};
}
\draw[fill=yellow,rounded corners=10pt,draw=white,opacity=1] (\bitsize+\transfoffset+\transfsize*0.5+\transfdistance*2,\bitsize)--(\bitsize+\transfoffset+\transfdistance*2,\bitsize)--(\bitsize+\transfoffset+\transfdistance*2,-\bitsize)--(\bitsize+\transfoffset+\transfsize*0.5+\transfdistance*2,-\bitsize);
\draw[black,rounded corners=10pt, line width=1pt,opacity=1] (\bitsize+\transfoffset+\transfsize*0.5+\transfdistance*2,\bitsize)--(\bitsize+\transfoffset+\transfdistance*2,\bitsize)--(\bitsize+\transfoffset+\transfdistance*2,-\bitsize)--(\bitsize+\transfoffset+\transfsize*0.5+\transfdistance*2,-\bitsize);

\draw[fill=yellow,rounded corners=10pt,draw=white,opacity=1] (-\bitsize-\transfoffset,\bitsize)--(\halfleft,\bitsize)--(\halfleft,-\bitsize)--(-\bitsize-\transfoffset,-\bitsize);
\draw[black,rounded corners=10pt, line width=1pt,opacity=1] (-\bitsize-\transfoffset,\bitsize)--(\halfleft,\bitsize)--(\halfleft,-\bitsize)--(-\bitsize-\transfoffset,-\bitsize);

\end{tikzpicture}}\right)=\resizebox{0.35\hsize}{!}{\begin{tikzpicture}[baseline=(O.base)]
\node (O) at (0,2*\hdist) {};
\foreach \t in  {0,...,5}{
\draw[black]  (\bitdistance*\t+\bitsize*\t,10*\hdist) -- (\bitdistance*\t+\bitsize*\t, \bitsize*0.5+2pt) ;
}
\foreach \t in {0,...,2}{
\node[pi] at (\bitsize+\transfoffset+\transfdistance*\t,0) {$\Pi'$};
}

\end{tikzpicture}}.
\]
Let $V=\exp\left(-\frac{\theta}{2}Y\gt X\right)$, $a=\cos \theta, b=\sin \theta$ so that 
\[
\begin{aligned}
&V(X\gt I)V^\dag =X\gt I\,, \\
&V(I\gt X)V^\dag =-bY\gt I+ aI\gt X\,,\\
&V(Y\gt I)V^\dag= aY\gt I+ bI\gt X\,,\\
&V(I\gt Y)V^\dag=I\gt Y.
\end{aligned}
\]
We have the diagrammatic equation:
\begin{equation*}\label{eq:V_action}
    \resizebox{0.8\hsize}{!}{\begin{tikzpicture}[baseline=(O.base)]
\node (O) at (0,3*\hdist) {};
\foreach \t in  {0,...,5}{
\draw[black]  (\bitdistance*\t+\bitsize*\t,10*\hdist) -- (\bitdistance*\t+\bitsize*\t, \bitsize*0.5+10pt) ;
\node[bit]  at (\bitdistance*\t+\bitsize*\t,0)   {};
}
\foreach \t in {0,...,2}{
\node[M1] at (\bitsize+\transfoffset+\transfdistance*\t,4*\hdist) {$\mathcal{V}_\text{loc}$};
}
\foreach \t in {0,...,1}{
\node[M1] at (\bitsize+\transfoffset+\transfsize*0.5+\transfdistance*\t,\hdist) {$\mathcal{V}_\text{loc}$};
}
\foreach \t in {0,1}{
\node[pi] at (0.5*\transfdistance+\bitsize+\transfoffset+\transfdistance*\t,0) {$\Pi$};
}
\draw[fill=yellow,rounded corners=10pt,draw=white,opacity=1] (\bitsize+\transfoffset+\transfsize*0.5+\transfdistance*2,\bitsize)--(\bitsize+\transfoffset+\transfdistance*2,\bitsize)--(\bitsize+\transfoffset+\transfdistance*2,-\bitsize)--(\bitsize+\transfoffset+\transfsize*0.5+\transfdistance*2,-\bitsize);
\draw[black,rounded corners=10pt, line width=1pt,opacity=1] (\bitsize+\transfoffset+\transfsize*0.5+\transfdistance*2,\bitsize)--(\bitsize+\transfoffset+\transfdistance*2,\bitsize)--(\bitsize+\transfoffset+\transfdistance*2,-\bitsize)--(\bitsize+\transfoffset+\transfsize*0.5+\transfdistance*2,-\bitsize);

\draw[fill=yellow,rounded corners=10pt,draw=white,opacity=1] (-\bitsize-\transfoffset,\bitsize)--(\halfleft,\bitsize)--(\halfleft,-\bitsize)--(-\bitsize-\transfoffset,-\bitsize);
\draw[black,rounded corners=10pt, line width=1pt,opacity=1] (-\bitsize-\transfoffset,\bitsize)--(\halfleft,\bitsize)--(\halfleft,-\bitsize)--(-\bitsize-\transfoffset,-\bitsize);

\draw[fill=myblue1,rounded corners=10pt,draw=white,opacity=1] (\bitsize+\transfoffset+\transfsize*0.5+\transfdistance*2,0.1*\hdist)--(\bitsize+\transfoffset+\transfdistance*2,0.1*\hdist)--(\bitsize+\transfoffset+\transfdistance*2,2*\hdist)--(\bitsize+\transfoffset+\transfsize*0.5+\transfdistance*2,2*\hdist);
\draw[black,rounded corners=10pt, line width=1pt,opacity=1] (\bitsize+\transfoffset+\transfsize*0.5+\transfdistance*2,0.1*\hdist)--(\bitsize+\transfoffset+\transfdistance*2,0.1*\hdist)--(\bitsize+\transfoffset+\transfdistance*2,2*\hdist)--(\bitsize+\transfoffset+\transfsize*0.5+\transfdistance*2,2*\hdist);

\draw[fill=myblue1,rounded corners=10pt,draw=white,opacity=1] (-\bitsize-\transfoffset,0.1*\hdist)--(\halfleft,0.1*\hdist)--(\halfleft,2*\hdist)--(-\bitsize-\transfoffset,2*\hdist);
\draw[black,rounded corners=10pt, line width=1pt,opacity=1] (-\bitsize-\transfoffset,0.1*\hdist)--(\halfleft,0.1*\hdist)--(\halfleft,2*\hdist)--(-\bitsize-\transfoffset,2*\hdist);
\end{tikzpicture}=},
\end{equation*}
with $\mathcal{V}_{loc}(O)=V(O)V^\dag$ for all $O\in\m A_{x,x+1}$. Inverting the second layer of $\mathcal{V}_\text{loc}$ we get to
\begin{equation}\label{eq:V_action_fact}
    \resizebox{0.8\hsize}{!}{\begin{tikzpicture}[baseline=(O.base)]
\node (O) at (0,3*\hdist) {};
\foreach \t in  {0,...,5}{
\draw[black]  (\bitdistance*\t+\bitsize*\t,10*\hdist) -- (\bitdistance*\t+\bitsize*\t, \bitsize*0.5+10pt) ;
\node[bit]  at (\bitdistance*\t+\bitsize*\t,0)   {};
}
\foreach \t in {0,...,1}{
\node[M1] at (\bitsize+\transfoffset+\transfsize*0.5+\transfdistance*\t,\hdist) {$\mathcal{V}_\text{loc}$};
}
\foreach \t in {0,1}{
\node[pi] at (0.5*\transfdistance+\bitsize+\transfoffset+\transfdistance*\t,0) {$\Pi$};
}
\draw[fill=yellow,rounded corners=10pt,draw=white,opacity=1] (\bitsize+\transfoffset+\transfsize*0.5+\transfdistance*2,\bitsize)--(\bitsize+\transfoffset+\transfdistance*2,\bitsize)--(\bitsize+\transfoffset+\transfdistance*2,-\bitsize)--(\bitsize+\transfoffset+\transfsize*0.5+\transfdistance*2,-\bitsize);
\draw[black,rounded corners=10pt, line width=1pt,opacity=1] (\bitsize+\transfoffset+\transfsize*0.5+\transfdistance*2,\bitsize)--(\bitsize+\transfoffset+\transfdistance*2,\bitsize)--(\bitsize+\transfoffset+\transfdistance*2,-\bitsize)--(\bitsize+\transfoffset+\transfsize*0.5+\transfdistance*2,-\bitsize);

\draw[fill=yellow,rounded corners=10pt,draw=white,opacity=1] (-\bitsize-\transfoffset,\bitsize)--(\halfleft,\bitsize)--(\halfleft,-\bitsize)--(-\bitsize-\transfoffset,-\bitsize);
\draw[black,rounded corners=10pt, line width=1pt,opacity=1] (-\bitsize-\transfoffset,\bitsize)--(\halfleft,\bitsize)--(\halfleft,-\bitsize)--(-\bitsize-\transfoffset,-\bitsize);

\draw[fill=myblue1,rounded corners=10pt,draw=white,opacity=1] (\bitsize+\transfoffset+\transfsize*0.5+\transfdistance*2,0.1*\hdist)--(\bitsize+\transfoffset+\transfdistance*2,0.1*\hdist)--(\bitsize+\transfoffset+\transfdistance*2,2*\hdist)--(\bitsize+\transfoffset+\transfsize*0.5+\transfdistance*2,2*\hdist);
\draw[black,rounded corners=10pt, line width=1pt,opacity=1] (\bitsize+\transfoffset+\transfsize*0.5+\transfdistance*2,0.1*\hdist)--(\bitsize+\transfoffset+\transfdistance*2,0.1*\hdist)--(\bitsize+\transfoffset+\transfdistance*2,2*\hdist)--(\bitsize+\transfoffset+\transfsize*0.5+\transfdistance*2,2*\hdist);

\draw[fill=myblue1,rounded corners=10pt,draw=white,opacity=1] (-\bitsize-\transfoffset,0.1*\hdist)--(\halfleft,0.1*\hdist)--(\halfleft,2*\hdist)--(-\bitsize-\transfoffset,2*\hdist);
\draw[black,rounded corners=10pt, line width=1pt,opacity=1] (-\bitsize-\transfoffset,0.1*\hdist)--(\halfleft,0.1*\hdist)--(\halfleft,2*\hdist)--(-\bitsize-\transfoffset,2*\hdist);
\end{tikzpicture}=\begin{tikzpicture}[baseline=(O.base)]
\node (O) at (0,3*\hdist) {};
\foreach \t in  {0,...,5}{
\draw[black]  (\bitdistance*\t+\bitsize*\t,10*\hdist) -- (\bitdistance*\t+\bitsize*\t, \bitsize*0.5+10pt) ;
\node[bit]  at (\bitdistance*\t+\bitsize*\t,0)   {};
}
\foreach \t in {0,...,2}{
\node[M1] at (\bitsize+\transfoffset+\transfdistance*\t,\hdist) {$\mathcal{V}_\text{loc}^{-1}$};
\node[pi] at (\bitsize+\transfoffset+\transfdistance*\t,0) {$\Pi'$};
}

\end{tikzpicture}}.
\end{equation}
Comparing the factorisations on the left and right of Eq.\eqref{eq:V_action_fact}, we obtain the following necessary and sufficient conditions for our evolution to be renormalised:
\begin{equation}
\begin{split}\label{eq:avoid_nightmare}
    \mathcal{V}_{\text{loc}}(\Pi_{\Lambda_x})=P_1\gt P_2,\\
    \mathcal{V}_{\text{loc}}^{-1}\circ\mathcal{U}^{-1}(\Pi_{\Lambda_x})=P_2\gt P_1,
\end{split}
\end{equation}
i.e. the image of $\Pi_{\Lambda_x}$ under $\mathcal{V}_{\text{loc}}$ and $\mathcal{V}_{\text{loc}}^{-1}\circ\mathcal{U}^{-1}$ should be factorised, and equal up to a fermionic swap
\[
SWAP\circ\mathcal{V}_{\text{loc}}(\Pi_{\Lambda_x})=\mathcal{V}_{\text{loc}}^{-1}\circ\mathcal{U}^{-1}(\Pi_{\Lambda_x}).
\]
In general we can write any rank-2 projection as $\Pi_{\Lambda_x}=\frac1{2}\{I\gt I+T_{\text{less}}\}$, where $T_{\text{less}}$ is a traceless operator. Then the only possible solutions of \Eq\eqref{eq:avoid_nightmare} are:
\[
P_1\gt P_2 =  \frac{(I\pm Z)\gt I}{2}, \frac{I\gt (I\pm Z)}{2},
\]
where we considered $(P_1\gt P_2)^2=P_1\gt P_2$ and the parity superselection ($\Z_2$-grading of the algebra). We explicitly verify the claim below.
Since $\text{rank}(\Pi_{\Lambda_x})=2$ then in order for $P_2\gt P_2$ to have the same rank as $\Pi_{\Lambda_x}$ it must be $P_1\gt P_2= I \gt P_{\pm}$ (or $=P_{\pm} \gt I$), where $P_{\pm}=\frac{1}{2}(I\pm Z)$. Hence, the only possible solutions for Eqs.~\eqref{eq:avoid_nightmare} are the ones listed above. 

We compute $\Pi_{\Lambda_x}=\mathcal V^{-1}_{loc}(P_1 \gt P_2)$:
\begin{align}\label{eq:p1p2}
    \mathcal V^{-1}_{loc}(\frac{1}{2}(I\pm Z) \gt I)=\frac{1}{2}(I \gt I \pm (a Z\gt I +i b X\gt X))
\end{align}
\[\mathcal V^{-1}_{loc}(\frac{1}{2}I \gt (I\pm Z))=\frac{1}{2}(I \gt I \pm (a I\gt Z -i b Y\gt Y))\]
and $\Pi_{\Lambda_x}=\mathcal U \circ  \mathcal V_{loc}(P_2 \gt P_1)$:
\begin{align}\label{eq:p2p1}
\begin{aligned}
    &\mathcal U \circ  \mathcal V_{loc}(\frac{1}{2}I \gt (I\pm Z))=\\
    &=\frac{1}{2}(I \gt I \pm (a I\gt Z -i b \mathcal U(Y)\gt \mathcal U(Y)))=\\
    &=\frac{1}{2}\{I \gt I \pm [a I\gt Z +\\
    &-i b[a^2Y\gt Y-ab(Y\gt X+X \gt Y)+b^2X\gt X ]]\}
    \end{aligned}
\end{align}
As Eq.s\eqref{eq:p1p2} and \eqref{eq:p2p1} must be equal, it follows $a=0$ and $b^3=b$, i.e. $\theta=\pm\pi/2+ k\pi$.\\
Therefore the only  $(2, \Pi_{\Lambda_0})$-renormalisable right Majorana shifts are  $\mathcal U \circ\sigma_+(X_x)=\mp X_x, \,\mathcal U \circ\sigma_+(Y_x)=\pm Y_{x+1}$. Analogously, the only  $(2, \Pi_{\Lambda_0})$-renormalisable left Majorana shifts are  $\mathcal U \circ\sigma_-(X_x)=\mp X_x, \,\mathcal U \circ\sigma_-(Y_x)=\pm Y_{x-1}$. 
\end{proof}
The $(2, \Pi_{\Lambda_0})$-renormalised FCA of Majorana shifts are summarised in the following Proposition.
\begin{proposition}
    The $(2,\Pi_{\Lambda_0})$-renormalisable Majorana Shifts $\tU\circ\sigma_+$ are renormalised in a shift through $\Pi^X$ and into a trivial FCA through $\Pi^Y$.
\end{proposition}
\begin{proof}
    We first show that the Majorana shift is renormalised into a shift through $\Pi_{\Lambda_x}=\Pi^X$.
    The algebra
    \begin{align}\label{eq:isom_majo}
        \tP^X(\m A({\Lambda_x}))=\Pi^X\mathsf A({\Lambda_x})\Pi^X \gt \bigboxtimes_{x\neq y\in \mathbb{L'}}\Pi_{\Lambda_y} \cong \A'_{x} \gt \C\tI
    \end{align}
     is isomorphic to $\M{1}{1}$ with two odd graded-commuting generators given by
     \begin{align*}
         \xi_{1,\Lambda_x}=\tP^X(I_{2x}\gt Y_{2x+1})\,,\\
         \xi_{2,\Lambda_x}=\tP^X(Y_{2x}\gt I_{2x+1})\,,
     \end{align*}
     where $\xi_{j,\Lambda_x} \equiv\tE(\xi_{j,x}')$ with $\xi_{j,x}'$ two odd graded-commuting generators of $\m A'_x$. The isomorphism in Eq.\eqref{eq:isom_majo} is realised by $\tE, \tE^\dag$ as per Eq.\eqref{eq:Jiso3}. We recall that the renormalised FCA $\tS$ is given by $\tS=\tE^\dag \circ (\tU \circ \sigma_+)^2 \circ \tE $ and for a renormalisable FCA we have $(\tU \circ \sigma_+)^2 \circ \tP^X=\tP^X \circ (\tU \circ \sigma_+)^2$, with $\P^X=\tE \circ \tE^\dag$ and $\tE^\dag (\Pi_{\Lambda_x})=I_x \in \m A'(\L')$. Now, since $(\tU\circ\sigma_+)^2(Y_x)=Y_{x+2}$, we have 
     \begin{align*}
        (\tU\circ\sigma_+)^2\circ \tE(\xi'_{1,x}) =(\tU\circ\sigma_+)^2\circ \tP(I_{2x}\gt Y_{2x+1})=\\
        =\tP \circ(\tU\circ\sigma_+)^2(I_{2x}\gt Y_{2x+1})=\\ 
        =(\Pi_{\Lambda_x} \gt \xi_{1,\Lambda_{x+1}}) \gt \bigboxtimes_{y \in \mathbb{L'}\setminus\{x,x+1\}}\Pi_{\Lambda_y}. 
     \end{align*}
      Therefore 
      \begin{align*}
          \tS=&\tE^\dag(\Pi_{\Lambda_x} \gt \xi_{1,\Lambda_{x+1}} \gt \bigboxtimes_{y \in \mathbb{L'}\setminus\{x,x+1\}}\Pi_{\Lambda_y})=\\
          =&I \gt \xi'_{1,x+1}\,.
      \end{align*}
     Repeating the same steps for $\xi_{2,x}'$, we get $\tS(\xi_{2,x}')=I \gt \xi'_{2,x+1}$.\\
     By noticing that $(\tU\circ\sigma_+)^2(X_x)=X_{x}$, one can easily show that the Majorana shift is renormalised into one layer of cell-wise gates through $\Pi_{\Lambda_x}=\Pi^Y$ by applying the same above argument using as generators of $\tP^Y(\m A({\Lambda_x}))$
     \begin{align*}
         &\xi_{1,\Lambda_x}=\tP^Y(I_{2x}\gt X_{2x+1})\,,\\
         &\xi_{2,\Lambda_x}=\tP^Y(X_{2x}\gt I_{2x+1})\,.
     \end{align*}
\end{proof}

\section{Proof of Theorem \ref{prp:thecoarsegrainingisfinite}}\label{sec:appendiceThm1}

In this appendix we go through essentially the same proofs presented in \cite{trezzini2024renormalisationquantumcellularautomata} for  quantum cellular automata. We show that those results can be extended to the case of $\Z_2$-graded $C^*$-algebras.

\vspace{0.4cm}
We  define the map $\tT_r$ induced by $\tT_w$ on $\Aredloc$ as follows.
\begin{definition}[Induced map]
  Let $\tT$ be a QCA on $\mathbb{Z}^s$, $\mathbb{L} = \mathbb{Z}_{Nm}^s$
   be a regular wrapping, $\mathbb{L}' = \mathbb{Z}_{m}^s$ a coarse-grained lattice,
   and let $\tJ$ and $\tV$ be defined as in Eq.~\eqref{eq:Jiso3}.
   The \emph{induced map} $\tT_r:\Aredloc \to\Aredloc$ is defined as 
\begin{align}
  \begin{aligned}
  & 
\tT_r\coloneqq \tJ\circ\tT_{w}\circ\tV.    
  \end{aligned}
\label{eq:induced}
\end{align}
\end{definition}
Notice that  we have a
\textit{family} of induced maps
$\{\tT_r\}_{\mathbb L,\mathbb{L'}}$  depending on the choice of the
lattices $\mathbb{L}$ and $\mathbb L'$.  To simplify notation, we omit this dependence, since the transition rule is ultimately independent of their choice, as we will prove later.
We remark that, in general, we need to differentiate between the map induced by $N$ 
steps of the FCA $\tT$ and $N$ subsequent applications of the map induced by $\tT$. 
These two transformations are given by
\begin{align}
&(\tT^N)_r= \tJ\circ\tT_w^N\circ\tV && (\tT_r)^N= (\tJ\circ\tT_w\circ\tV)^N,
\end{align}
respectively.
In the following we will see how the induced map pins down the idea of coarse-graining.

\begin{lemma}\label{lem:CgT}
  If a FCA $\tT$ of $\m A(\mathbb{Z}^s)$ is 
$(N,\Pi_{\Lambda_0})-$renormalisable to some FCA $\tS$,
then every induced map $(\tT^N)_r$ is an FCA. In particular, $(\tT^N)_r=\tS_w$. 
\end{lemma}
\begin{proof}
The induced map of  $N$ steps of $\tT$ is
\begin{equation}
(\tT^N)_r= \tJ\circ\tT^N_{w}\circ\tV.
\label{eq:inducedN}
\end{equation}
Composing Eq.~\eqref{eq:cg} on the left with $\tJ$ we have
\begin{equation*}
\tJ\circ\tT^N_{w}\circ\tV=\tJ\circ\tV\circ\tS_w=\tS_w,
\end{equation*}
whence
\begin{equation*}
\tS_w=(\tT^N)_r,
\end{equation*}
for every possible choice of regular wrapping $\mathbb{L}$.
\end{proof}

Notice that, if $(\tT^N)_r$ is an FCA for every
choice of the lattice $\mathbb{L}$, the Wrapping
Lemma allows us to extend it to an FCA $\tS$ on the
quasi-local algebra $\A'(\Z^ s)$ over the infinite
lattice built from the local algebra $\m A'(\L')$. However, the transition rule of $\tS$ and the algebra $\m A'(\L')$ may, in principle, depend on the specific wrapping. 
In the following, we prove that this is not the case.

Let $\tT:\A(\Z^s) \to \A(\Z^s)$ and $\tS:\A'(\Z^s) \to \A'(\Z^s)$ be two FCA, and let us study when 
condition \eqref{eq:cg}  holds for a given
wrapping $\mathbb{L}$ of $\Z^s$. Since $\tT_w$
is a graded $*$-automorphism of a central simple finite-dimensional $\Z_2$-graded $C^*$-algebra, by the $\Z_2$-graded counterpart of the Skolem-Noether Theorem (Appendix \ref{app:z2algebras}),
we can
always find a unitary $U_\tT\in \A(\mathbb L)$ such that
\begin{equation}\label{eq:evunit}
\tT_w(A)=\Tmat^\dag A \Tmat \hspace{10pt} \forall A\in \mathsf{A}(\mathbb{L}),
\end{equation}
and the same holds for $\tS_{w}$
\begin{equation*}
\tS_w(A)=(U_\tS)^\dag A (U_\tS)\hspace{10pt} \forall A\in \mathsf{A}'(\mathbb{L}').
\end{equation*}
Moreover, Eq.\eqref{eq:inducedN} allows us to write
\begin{align}
&\Tred (A)=\Tredmat^\dag A \Tredmat\label{eq:tred}\,,\\
&\Tredmat \coloneqq \tE^\dag( \Tnmat)\,, \label{eq:tredmat}
\end{align}
where in principle $\Tredmat$ might not be a unitary operator. However, condition~\eqref{eq:cg} implies that $\Tredmat=U_\tS e^{i\phi}$ is a unitary operator for every regular wrapping $\mathbb{L}$. If this is the case, Eq.\eqref{eq:tred} guarantees that $\Tred$ is a graded $*$-automorphism.

The following Lemma identifies a necessary and sufficient condition on $U_\tT$ for Eq.\eqref{eq:cg} to hold.
\begin{lemma} \label{lem:comm}
Eq.\eqref{eq:cg} holds with $\tS_w=\mathcal V $ and $\mathcal V \coloneqq (\tT^N)_r$ being a graded $*$-automorphism \textit{iff} for all regular wrappings
\begin{equation}\label{eq:commuwrap}
  \begin{aligned}
    & \tT^N_w(\Pi)=\Pi,
  \end{aligned}
\end{equation}
or equivalently
\begin{equation}\label{eq:commuwrap2}
    \left[\Tnmat,\, \Pi\right]=0\,.
\end{equation}
\end{lemma}
\begin{proof}
If \eqref{eq:cg} holds, we can compose it on the right with $\tJ$ and get
\begin{equation}
\tT^N_w\circ \tP=\tV\circ\tS_w\circ\tJ.
\end{equation}
Now, composing the r.h.s.~of the latter expression with $\tP=\tV\circ\tJ$ on the left gives
$
\tP\circ\tV\circ\tS_w\circ\tJ=\tV\circ\tS_w\circ\tJ=\tT^N_w\circ\tP,
$
while on the other hand, considering l.h.s., we obtain
\begin{align}\label{eq:pro}
    \tP\circ\tT^N_w\circ\tP=\tT^N_w\circ\tP.
\end{align}

Evaluating Eq.\eqref{eq:pro} on the identity element we get $
\Pi U_{\tT^N}^\dag\Pi U_{\tT^N}\Pi= U_{\tT^N}^\dag\Pi U_{\tT^N}$,
where $U_{\tT^N}$ is the unitary associated with $\tT^N_w$.
Since $\Pi'\coloneqq U_{\tT^N}^\dag\Pi U_{\tT^N}$ is a projection such that $\Pi\Pi'\Pi=\Pi'$, then $\Pi\Pi'=\Pi'\Pi=\Pi'$ and $\rank(\Pi')=\rank(\Pi)$. Thus we have $\Pi'=\Pi$, i.e.   
\begin{align*}
   U^\dag_{\tT^N}\Pi U_{\tT^N}=\Pi\,,
\end{align*}
which is Eq.\eqref{eq:commuwrap}.
This amounts to the commutation condition~\eqref{eq:commuwrap2}. Moreover, this condition is equivalent to 
\begin{align}\label{eq:eqv}
    \tP \circ \tT^N=\tT^N \circ \tP\,.
\end{align}
Indeed, on the one hand Eq.~\eqref{eq:eqv} clearly implies $\tP\circ\tT^N_w\circ\tP=\tT^N_w\circ\tP$; on the other hand from Eq.~\eqref{eq:commuwrap} for any $O \in \A(\L)$ we have $\Pi O\Pi=\P(O)=U_{\tT^N}^\dag\Pi U_{\tT^N} OU_{\tT^N}^\dag\Pi U_{\tT^N}=\tT^N\circ\P\circ (\tT^{N})^{-1}(O)$.

To prove the converse, 
we remind that $\Tredmat =\tE^\dag(\Tnmat)$, 
Then, exploiting~\eqref{eq:FF}, we have
\begin{align*}
&\Tredmat \Tredmat^\dag= \sum_{f,g,l,k\in S}\tF_{f\psi_f,g\psi_g}(\Tnmat)\tF_{l\psi_l,k\psi_k}(\Tnmat^\dag)\\
&= \tE^\dag( \Tnmat \Pi \Tnmat^\dag )  = \tE^\dag( \Pi \Tnmat\Tnmat^\dag)\\
&=\tE^\dag(\Pi)= \tE^\dag\circ\tE(I_{\m A'(\mathbb L')})=I_{\m A'(\mathbb{L}')},
\end{align*}
where we exploited the commutation of $\Tnmat$
with $\Pi$. The same holds for
$\Tredmat^\dag\Tredmat$ and thus $\Tredmat$ is
unitary over $\Aredloc $.  We can thus make the
identification $U_\tS=\Tredmat$, and use the fact that an isometry is a $*$-homomorphism to obtain 
the identity
\begin{align*}
    &\tE\circ\tS_w(X)=\tE(\Tredmat^\dag)\tE(X)\tE(\Tredmat)\\
    &=\tE\circ\tE^\dag(\Tnmat^\dag)\tE(X)\tE\circ\tE^\dag(\Tnmat)\\
    &=\tP(\Tnmat^\dag)\tE(X)\tP(\Tnmat)\\
    &=\Pi\Tnmat^\dag\Pi\tE(X)\Pi\Tnmat\Pi\\
    &=\Tnmat^\dag\Pi\tE(X)\Pi\Tnmat\\
    &=\tT^N_w\circ\tP\circ\tE(X)\\
    &=\tT^N_w\circ\tE\circ\tE^\dag\circ\tE(X),
\end{align*}
which finally yields
$
\tV\circ\tS_w=\tT^N_w\circ\tV,
$
with  $\tS_w=\mathcal V$ an automorphism of the  algebra 
$\mathsf{A}'(\mathbb{L}')$.
\end{proof}
Notice that Eq.~\eqref{eq:commuwrap} is necessary and sufficient for condition~\eqref{eq:cg} for~\emph{some} automorphism $\mathcal V$, which by now is not necessarily the wrapped version of a given FCA $\tS$ on $\mathbb Z^s$. Thus, we do not have yet a necessary and sufficient condition for renormalisability.

 We now show that all
maps $\mathcal V$ share the same transition rule. This has
two important consequences. The first one is that Eqs.~\eqref{eq:inducedN} 
and~\eqref{eq:commuwrap} become necessary and sufficient for $(N,\Pi_{\Lambda_0})$-renormalisability. The second and most powerful one is that we can check the
condition in Eq.\eqref{eq:commuwrap} only on
the smallest regular wrapping for
$\tT^N$.  We first provide two preliminary lemmas.
\begin{lemma}
$(\tT^N)_r$ is invariant under shifts on $\mathbb{L}'$.
\end{lemma}
\begin{proof}
A shift $\tau_x^w$ along the vector $x\in\mathbb{L}'$ 
corresponds to a shift $\tau_{Nx}^w$
along the vector $Nx\in\mathbb{L}$.
Since  
\begin{align*}
\begin{aligned}
\tau_{Nx}^w\tV=\tV\tau_{x}^w,\\
\tau_{x}^w\tJ=\tJ\tau_{Nx}^w,\\
\end{aligned}
\end{align*}
 and the FCA
$\tT^N$ commutes with $\tau_{Nx}^w$,
by their composition for the induced map $(\tT^N)_r$ it follows
\begin{align*}
[(\tT^N)_r,\tau_x]=0 \quad \forall x\in \mathbb{L}'.
\end{align*}
\end{proof}

\begin{lemma}
$\Tnmat$ can be written in the following form:
\begin{align*}
&\Tnmat=[(I_{\Lambda_{0}}\gt V)(U\gt I_{\Lambda_{\mathcal W}})],\\
&\mathcal W\coloneqq\mathbb{L}\setminus(\Lambda_0\cup\mathcal{N}^N_{\Lambda_0}),
\end{align*}
where $U$ is independent of the wrapping, and $\mathcal{N}^N_{\Lambda_0}$ denotes the neighbourhood of $\Lambda_0$ for $\tT^N$.
\end{lemma}
\begin{proof}

By definition of FCA $\tT(\mathsf{A}(\Lambda_0))\subseteq\mathsf{A}(\mathcal{N}_{\Lambda_0})$, then $\tT^N(\mathsf{A}(\Lambda_0))\subseteq\mathsf{A}(\mathcal{N}^N_{\Lambda_0})$. 
In other words, the algebra $\tT^N(\mathsf{A}(\Lambda_0))$ is an homomorphic representation of the algebra $\mathsf{A}(\Lambda_0)$ in 
$\mathsf{A}(\mathcal N^N_{\Lambda_0})$, which means that there exists a unitary 
$U\in\mathsf{A}(\Lambda_0\cup\mathcal N^N_{\Lambda_0})$ such that, for every 
$O\in\mathsf A(\Lambda_0)$
\begin{align*}
&\tT^N(O)=U^\dag(O\gt I_{\mathcal{N}^N_{\Lambda_0}\setminus\Lambda_0})U\in\mathsf A(\Lambda_0\cup\mathcal{N}_{\Lambda_0}).
\end{align*}
For any $B\in \A(\mathbb L)$ the wrapped evolution $\tT^N_w$ is implemented over the wrapping by a unitary operator (Eq.~\eqref{eq:evunit})
\begin{align}
\tT^N_w(B)=\Tnmat^\dag(B\gt I_{\mathbb{L}\setminus\mathcal{N}_B})\Tnmat.
\end{align}
Since $\tT$ and $\tT_w$ share the same transition rule, we then have
\begin{align*}
\tT_w^N(B)&=\Tnmat^\dag(B\gt I_{\mathbb{L}\setminus\Lambda_0})\Tnmat\\
&=
U^\dag(B\gt I_{\mathcal{N}^N_{\Lambda_0}\setminus\Lambda_0})U\gt I_{\mathbb{L}\setminus(\Lambda_0\cup\mathcal{N}^N_{\Lambda_0})}.
\end{align*}
One can easily prove that this implies
\begin{equation*}
\Tnmat(U\gt I_{\mathbb{L}\setminus(\Lambda_0\cup\mathcal{N}^N_{\Lambda_0})})^\dag=(I_{\Lambda_0}\gt V),
\end{equation*}
for some unitary $U\in\mathsf A(\mathbb{L}\setminus \Lambda_0)$, and finally
\begin{equation}\label{eq:decompun}
\Tnmat=(I_{\Lambda_0}\gt V)(U\gt I_{\mathbb{L}\setminus(\Lambda_0\cup\mathcal{N}^N_{\Lambda_0})}).
\end{equation}
\end{proof}
We now prove that the condition \eqref{eq:commuwrap} is independent of the wrapping.

\begin{lemma} If condition in \eqref{eq:commuwrap} is satisfied for a regular wrapping $\mathbb{L}$ for $\tT^N$, then it is true for every other regular wrapping.
\end{lemma}
\begin{proof}
Let us divide the projection $\Pi$ over $\mathbb{L}$ in
\begin{equation*}
\begin{split}
&\Pi=\Pi_{0}\gt\Pi_{{\mathcal{M}}}\gt \Pi_{\mathcal W},\\
&{\mathcal M}\coloneqq\mathcal{N}\setminus\{0\},\\
&\mathcal W\coloneqq\mathbb{L}\setminus(\mathcal{N}\cup \{0\}),
\end{split}
\end{equation*} 
with the convention:
\begin{align*}
\Pi_{{R}}=\left(\bigboxtimes_{x\in R}\Pi_{{x}}\right).
\end{align*}
The commutation relation~\eqref{eq:commuwrap} can then be expressed as
\begin{equation}
\begin{split}
&\Tnmat(\Pi_{0}\gt\Pi_{\mathcal M}\gt \Pi_{\mathcal W})\Tnmat^\dag=\Pi_{0}\gt\Pi_{\mathcal M}\gt \Pi_{\mathcal W}.
\end{split}
\label{eq:Tpi}
\end{equation}

We can represent Eq.~\eqref{eq:Tpi} diagrammatically, with the decomposition~\eqref{eq:decompun} of $\Tnmat$, as follows
\begin{equation}
\begin{tikzpicture}[baseline=(O.base)]
\node (O) at (0,0.25*\hdist) {};
\foreach \t in  {0,...,2}{
\draw[black]  (\bitdistance*\t+\bitsize*\t,6*\hdist) -- (\bitdistance*\t+\bitsize*\t, \bitsize*0.5+2pt) ;
}
\node[bit]  at (0,0)   {\tiny $\Pi_0$};
\node[bit] at (\bitdistance+\bitsize,0) {\tiny $\Pi_{{\mathcal{M}}}$};
\node[bit] at (2*\bitdistance+2*\bitsize,0) {\tiny $\Pi_{\mathcal W}$};
\node[M1] at (\bitsize+\transfoffset+\transfsize*0.5,4*\hdist) {$\mathcal{V}$};
\node[M2] at (\bitsize+\transfoffset,\hdist) {$\mathcal{U}$};
\end{tikzpicture} = \begin{tikzpicture}[baseline=(O.base)]
\node (O) at (0,0.25*\hdist) {};
\foreach \t in  {0,...,2}{
\draw[black]  (\bitdistance*\t+\bitsize*\t,6*\hdist) -- (\bitdistance*\t+\bitsize*\t, \bitsize*0.5+2pt) ;
}
\node[bit]  at (0,0)   {\tiny$\Pi_0$};
\node[bit] at (\bitdistance+\bitsize,0) {\tiny $\Pi_{{\mathcal{M}}}$};
\node[bit] at (2*\bitdistance+2*\bitsize,0) {\tiny $\Pi_{\mathcal W}$};
\end{tikzpicture}\ .
\end{equation}
Exploiting the decompositon in Eq.\eqref{eq:decompun} and conjugating both members of Eq.~\eqref{eq:Tpi} with $(I_{\Lambda_{0}}\gt V)^\dag$ we get to:
\begin{equation*}
\begin{split}
&(U \gt I_{\Lambda_{\mathcal W}})(\Pi_{0}\gt\Pi_{\mathcal M}\gt \Pi_{\mathcal{W}})(U\gt I_{\Lambda_{\mathcal W}})^\dag=\\
&=(I_{\Lambda_{0}}\gt V)^\dag(\Pi_{0}\gt\Pi_{\mathcal M}\gt \Pi_{\mathcal{W}})(I_{\Lambda_{0}}\gt V).
\end{split}
\end{equation*}
Diagrammatically:
\begin{equation}\label{eq:diagfact}
\begin{tikzpicture}[baseline=(O.base)]
\node (O) at (0,0.25*\hdist) {};
\foreach \t in  {0,...,2}{
\draw[black]  (\bitdistance*\t+\bitsize*\t,3*\hdist) -- (\bitdistance*\t+\bitsize*\t, \bitsize*0.5+2pt) ;
}
\node[bit]  at (0,0)   {\tiny $\Pi_0$};
\node[bit] at (\bitdistance+\bitsize,0) {\tiny $\Pi_{{\mathcal{M}}}$};
\node[bit] at (2*\bitdistance+2*\bitsize,0) {\tiny $\Pi_{\mathcal W}$};
\node[M2] at (\bitsize+\transfoffset,\hdist) {$\mathcal{U}$};
\end{tikzpicture}=\begin{tikzpicture}[baseline=(O.base)]
\node (O) at (0,0.25*\hdist) {};
\foreach \t in  {0,...,2}{
\draw[black]  (\bitdistance*\t+\bitsize*\t,3*\hdist) -- (\bitdistance*\t+\bitsize*\t, \bitsize*0.5+2pt) ;
}
\node[bit]  at (0,0)   {\tiny $\Pi_0$};
\node[bit] at (\bitdistance+\bitsize,0) {\tiny $\Pi_{{\mathcal{M}}}$};
\node[bit] at (2*\bitdistance+2*\bitsize,0) {\tiny $\Pi_{\mathcal W}$};
\node[M1] at (\bitsize+\transfoffset+\transfsize*0.5,\hdist) {$\mathcal{V}^\dag$};
\end{tikzpicture}\ .
\end{equation}
The two members of Eq.~\eqref{eq:diagfact} are factorised in two different ways. This implies that both factorisations must hold, and then
the transformation $V$ on the right produces a factorised projection map, i.e.
\begin{equation*}
\begin{split}
&(I_0\gt V)^\dag(\Pi_{0}\gt\Pi_{\mathcal M}\gt \Pi_{{\mathcal W}})(I_0\gt V )=\\
&=\Pi_{0}\gt Z^\dag\Pi_{{\mathcal{M}}}Z\gt \Pi_{{\mathcal W}},
\end{split}
\end{equation*}
for some unitary $Z\in\mathsf A(\mathcal M)$. In particular, we have the following diagrammatic identity
\begin{equation*}
\begin{tikzpicture}[baseline=(O.base)]
\node (O) at (0,0.25*\hdist) {};
\foreach \t in  {0,...,1}{
\draw[black]  (\bitdistance*\t+\bitsize*\t,3*\hdist) -- (\bitdistance*\t+\bitsize*\t, \bitsize*0.5+2pt) ;
}
\node[bit]  at (0,0)   {\tiny $\Pi_0$};
\node[bit] at (\bitdistance+\bitsize,0) {\tiny $\Pi_{\mathcal{M}}$};
\node[M2] at (\bitsize+\transfoffset,\hdist) {$\mathcal{U}$};
\end{tikzpicture}=\begin{tikzpicture}[baseline=(O.base)]
\node (O) at (0,0.25*\hdist) {};
\foreach \t in  {0,...,1}{
\draw[black]  (\bitdistance*\t+\bitsize*\t,3*\hdist) -- (\bitdistance*\t+\bitsize*\t, \bitsize*0.5+2pt) ;
}
\node[bit]  at (0,0)   {\tiny$\Pi_0$};
\node[bit] at (\bitdistance+\bitsize,0) {\tiny $\Pi_{\mathcal{M}}$};
\node[U] at (\bitsize+\bitdistance,\hdist) {$\mathcal{Z}^\dag$};
\end{tikzpicture}\ .
\end{equation*}
The latter is independent of the lattice $\mathbb{L}$, and proves that the image of 
$\Pi$ under $\Tnmat \rvert_\Pi$ is a projection of the form 
$\widetilde\Pi_{{\mathbb{L}'\setminus\{0\}}}\gt\Pi_{0}$. Finally, by translation 
invariance we can conclude that for every regular wrapping $\mathbb{L}$ for $\tT^N$ one has
\begin{align*}
\Tnmat^\dag \Pi \Tnmat=\Pi\,.
\end{align*}

\end{proof}
This immediately implies the following corollary.
\begin{Corollary}
If condition in \eqref{eq:commuwrap} is satisfied for a regular wrapping $\mathbb{L}$ for $\tT^N$, then
\begin{align}
\begin{split}
&(U\gt I_\mathcal W)^\dag(\Pi_{0}\gt Z^\dag\Pi_{{\mathcal{M}}}Z\gt \Pi_{{\mathcal W}})(U\gt I_\mathcal W)\\
&=\Pi_{0}\gt \Pi_{{\mathcal{M}}}\gt \Pi_{{\mathcal W}}.
\end{split}
\end{align}
\end{Corollary}
The above result allows us to show that condition~\eqref{eq:commuwrap} is necessary and sufficient for $(N,\Pi_{\Lambda_0})$-renormalisability.
\begin{lemma}
Let $\mathbb L$ and $\widetilde{\mathbb L}$  be two regular wrappings for $\tT^N$, and $\mathbb L'$ and $\widetilde{\mathbb L}'$ the associated coarse-grained lattices. Let condition~\eqref{eq:commuwrap} be satisfied for both. Then the transition rule of the induced map $(\tT^N)_r$ over 

$\A'(\mathbb{L}')$ and $\A'(\widetilde{\mathbb L}')$ is the same.
\end{lemma} 
\begin{proof}
It suffices to notice that a local operator $A_x$ in $\A'(\mathbb L')\cong\Aplocc$ is obtained as $\Pi A_{\Lambda_x}\Pi$ for some $A_{\Lambda_x}$ on the macro-cell 
$\Lambda_x$. As a consequence, one has
\begin{align*}
&\Tnmat^\dag A_x \Tnmat=\Tnmat^\dag \Pi A_{\Lambda_x}\Pi\Tnmat\\
&\quad=U^\dag(\Pi_{\Lambda_x}A_{\Lambda_x}\Pi_{\Lambda_x}\gt Z^\dag\Pi_{\mathcal N_{\Lambda_x}\setminus\Lambda_x} Z)U \gt\Pi_{\mathbb{L}\setminus(\Lambda_x\cup\mathcal N_{\Lambda_x})}),
\end{align*}
where $U$ is independent of $\mathbb{L}'$. Finally, by Eq.~\eqref{eq:tred} and~\eqref{eq:tredmat}, one has
\begin{align*}
(\tT^N)_r\rvert_{\widetilde{\mathbb L}'}(A_x)=(\tT^N)_r\rvert_{\mathbb{L}'}(A_x),
\end{align*}
for every $\widetilde{\mathbb L}'$, where we used $(\tT^N)_r\rvert_{\mathbb{G}}$ to denote the FCA $(\tT^N)_r$ restricted to the particular lattice $\mathbb{G}$.
\end{proof}

\end{document}